\documentclass[journal]{IEEEtran}

\usepackage[utf8]{inputenc}

\usepackage{amsmath,amssymb,amsthm}
\usepackage{graphicx}
\usepackage[dvipsnames]{xcolor}
\usepackage[backend=biber, style=ieee, defernumbers=false]{biblatex}
\DeclareBibliographyCategory{appendixonly}
\addtocategory{appendixonly}{%
  aigner1999proofs,buckholtz1979first,cayley1878theorem,debicki2021finite,%
  goldschmidt2016short,he1998double,iyengar1985hitting,kou2016first,%
  metzler2010first,morters2010brownian,rogers2006correlation,shao2013estimates,%
  zhou2001analysis,el2018inferring,semple2003phylogenetics}

\usepackage{subcaption}
\usepackage{enumitem}
\usepackage{newtxtext,newtxmath}
\usepackage{thmtools}
\usepackage{stfloats}
\usepackage[final]{hyperref}

\newtheorem{thm}{Theorem}[]
\newtheorem{lem}{Lemma}
\newtheorem*{lem*}{Lemma}

\newtheorem{ass}{Assumption}
\newtheorem{rmk}{Remark}
\newtheorem{defi}{Definition}

\newcommand\numberthis{\addtocounter{equation}{1}\tag{\theequation}}
\newcommand{\rev}[1]{\textcolor{black}{#1}}

\begin{document}

\title{Constraints on the perfect phylogeny mixture model and their effect on reducing degeneracy}

\author{
\IEEEauthorblockN{
John Marangola\IEEEauthorrefmark{1}\IEEEauthorrefmark{4}\IEEEauthorrefmark{3},
Azadeh Sheikholeslami\IEEEauthorrefmark{1}\IEEEauthorrefmark{3}\IEEEauthorrefmark{2},
Jos\'e Bento\IEEEauthorrefmark{3}\IEEEauthorrefmark{5}
}\\
\IEEEauthorblockA{\IEEEauthorrefmark{4}MIT}
\IEEEauthorblockA{\IEEEauthorrefmark{2}Suffolk University}
\IEEEauthorblockA{\IEEEauthorrefmark{3}Boston College}
\thanks{\IEEEauthorrefmark{1}These authors contributed equally to this work.}
\thanks{\IEEEauthorrefmark{5}Corresponding author: Jos\'e Bento (jose.bento@bc.edu).}
}

\maketitle

\begin{abstract}
The perfect phylogeny mixture (PPM) model is useful due to its simplicity and applicability in scenarios where mutations can be assumed to accumulate monotonically over time.
It is the underlying model in many tools \cite{miller2014sciclone,roth2014pyclone,malikic2015clonality,popic2015fast,el2015reconstruction,deshwar2015phylowgs,el2016inferring,reiter2017reconstructing,satas2017tumor,husic2019mipup,myers2019calder,toosi2019bamse,xiao2020fastclone,andersson2021devolution,baghaarabani2021conifer,hurtado2025phyclone} that have been used, for example, to infer phylogenetic trees for tumor evolution and reconstruction \cite{el2016inferring}.
Unfortunately, the PPM model gives rise to substantial ambiguity --  in that many different phylogenetic trees can explain the same observed data -- even in the idealized setting where data are observed perfectly, i.e. fully and without noise.
This ambiguity has been studied in this perfect setting \cite{pradhan2018non}, which proposed a procedure to bound the number of solutions given a fixed instance of observation data. Beyond this, studies have been primarily empirical.
Recent work proposed adding extra constraints to the PPM model to tackle ambiguity \cite{myers2019calder}.

In this paper, we first show that the extra constraints of \cite{myers2019calder}, called \emph{longitudinal constraints} (LC), often fail to reduce the number of distinct trees that explain the observations.
We then propose novel alternative constraints to limit solution ambiguity and study their impact when the data are observed perfectly.
Unlike the analysis in \cite{pradhan2018non}, our theoretical results---regarding both
the inefficacy of the LC and the extent to which our new constrains reduce ambiguity are not tied to a single observation instance.
Rather, our theorems hold over large ensembles of possible inference problems.
To the best of our knowledge, we are the first to study degeneracy in the PPM model in this ensemble-based theoretical framework.
\end{abstract}

\section{Introduction}

Phylogenetic tree inference advances scientific knowledge. In the medical field, for example, it can improve cancer treatment outcomes \cite{schwartz2017evolution}. Here,
understanding the evolutionary relationships among tumor clones allows for treatment strategies that target vulnerabilities shared among ancestral clones.

The standard approach to analyzing tumor genomes is bulk tumor sequencing, which provides a snapshot of the overall genetic mutations present within a tumor sample. To reconstruct ancestral relationships from bulk sequencing data, it is necessary to impose structural constraints on the evolutionary model. One such constraint that is frequently employed is the \emph{infinite sites assumption} \cite{kimura1969number, hudson1983properties}.
The infinite sites assumption states that a mutation arises at a specific locus only once during the tumor's evolutionary progression. This implies that throughout evolution the same genetic alteration cannot arise independently in separate lineages but can only be inherited, resulting in what is known as a perfect phylogeny.
Although the perfect phylogeny assumption is restrictive in certain scenarios \cite{bonizzoni2014and, bonizzoni2014explaining, zafar2017sifit, bonizzoni2017beyond, davis2017tumor, el2015reconstruction}, it captures many cancer evolutionary processes and has been utilized extensively in the literature \cite{dang2017clonevol, deshwar2015phylowgs, el2015reconstruction, jia2018efficient}.

\begin{figure}[t]
    \centering
    \includegraphics[width=\linewidth]{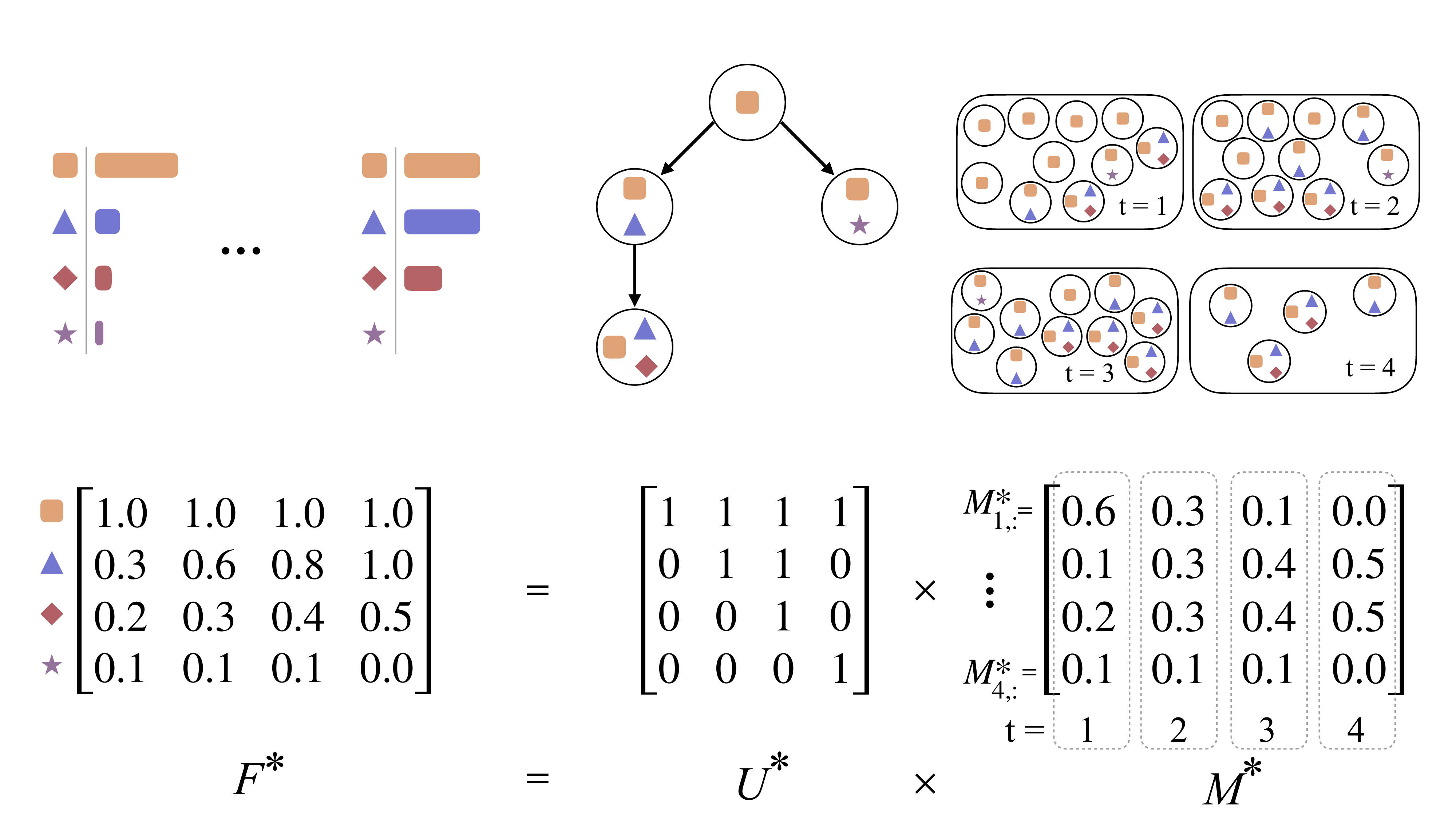}
    \caption{\rev{Quantities and relationships captured by the PPM model. Four mutant types evolve over four time samples. Mutant 1 carries mutation $\blacksquare$, mutant 2 carries mutations $\blacksquare$ and $\blacktriangle$, mutant 3 carries $\blacksquare$, $\blacktriangle$ and $\blacklozenge$ and mutant 4 carries $\blacksquare$ and $\bigstar$. Mutations are only accumulated over time, never lost, and the rooted tree illustrates which mutant types descend from which mutant types. We only observe mutation frequencies and from it we want to infer the other quantities. See the main text for more details.}}
    \label{eq:example_values_intro_paper_1}
\end{figure}

The Perfect Phylogeny Mixture (PPM) model \cite{hudson1983properties, kimura1969number} formalizes the infinite sites assumption by relating relative abundances of mutants $M^*$, their ancestral relationships $U^*$, and the observed frequencies of distinguishing mutations $F^*$.
In  Figure \ref{eq:example_values_intro_paper_1} we present an example of each of these quantities and now explain their interpretation.

Matrix $M^*$ describes how the relative abundances of mutant types $1, 2, 3$ and $4$ evolve over discrete sampling times \rev{$t = 1, \dots, 4$.}
The entry $M^*_{i,t}$ denotes the fraction of mutant type $i$ in the population at time $t$.
Matrix $U^*$ encodes the ancestral relationships among mutant types in a binary matrix, where $U_{i,j} = 1$ indicates that mutant $i$ is an ancestor of mutant $j$.
By definition, $U_{i,i} = 1$ for all $i$.
The matrix $U^*$ describes a rooted tree, with root mutant type $1$, which has children mutant types $4$ and $2$, and where mutant type $3$ is a child of mutant type $2$.
Throughout the paper, and to simplify language, we will sometimes refer to ``mutant type $i$'' simply as ``mutant $i$''.
Under the infinite sites assumption, from $U^*$ we identity unique mutations that accumulate and are never lost, and that distinguish child and parent.
In particular, mutant $1$ carries the symbolic root (null) mutation $1$, which is inherited by all its descendants.
Mutant $i$ acquires mutation $i$ and transmits it to all of its descendants. Because each mutant is uniquely identified by the mutation it acquires/transmits, $U^*$'s columns encode both ancestral relationships and mutation content.

At each point in time, the matrices $M^*$ and $U^*$ together determine the relative abundances of different mutations in the population.
For example, all mutants share mutation $1$, and therefore ${\bf F}^*_{1,:}$, the frequency of mutation $1$  across time, is constant and equal to $[1, 1, 1, 1, 1]$.
Mutation $2$ is shared by mutants $2$ and $3$, and thus its frequency over time is ${\bf F}^*_{2,:}  = {\bf M}^*_{2,:} + {\bf M}^*_{3,:} = [0.1, 0.3, 0.4, 0.5] + [0.2, 0.3, 0.4, 0.5] = [0.3, 0.6, 0.8, 1.0]$.
We can compute ${\bf F}^*_{3,:}$ and ${\bf F}^*_{4,:}$ by analogous reasoning. These operations are expressible in matrix form as
\begin{equation} \label{eq:first_time_F_UM_is_introduced_paper_1}
F^* = U^* M^*,
\end{equation}
which holds generally for any valid PPM instance.

The problem of inferring a phylogenetic tree from bulk data under the PPM model is to infer $U^*$ (and $M^*$) from $F^*$, or from a noisy or partially observed version of $F^*$.
In addition to the fact that solving this problem is NP-hard \cite{ bodlaender1992two, steel1992complexity, el2015reconstruction},
it frequently admits multiple distinct solutions, corresponding to different plausible evolutionary trajectories that explain the same data.
This multiplicity of solutions challenges biological interpretation—for instance, when designing treatment strategies targeting shared vulnerabilities among ancestral clones.
Throughout this paper, we refer to this multiplicity of solutions as \emph{degeneracy}. More degeneracy implies more solutions.

In the example of \eqref{eq:example_values_intro_paper_1}, degeneracy manifests as ambiguity in the placement of mutant type $4$ within the evolutionary tree.
One solution follows $U^*$, where mutant type $4$ is a child of mutant type $1$.
Alterative solutions place mutant $4$ as a child of mutant $2$ or mutant $3$.
For each of these three choices of $U$, a corresponding abundance matrix  $M$ exists that satisfies \eqref{eq:first_time_F_UM_is_introduced_paper_1}.
For any other choice of $U$, no meaningful $M$ satisfies \eqref{eq:first_time_F_UM_is_introduced_paper_1}. \rev{ See Figure \ref{fig:first_example_of_degeneracy} for an illustration of an alternative solution.}

\begin{figure}[h!]
\centering
\includegraphics[trim=1.2cm 1.cm 1.2cm 0.7cm, clip,width=0.9\columnwidth]{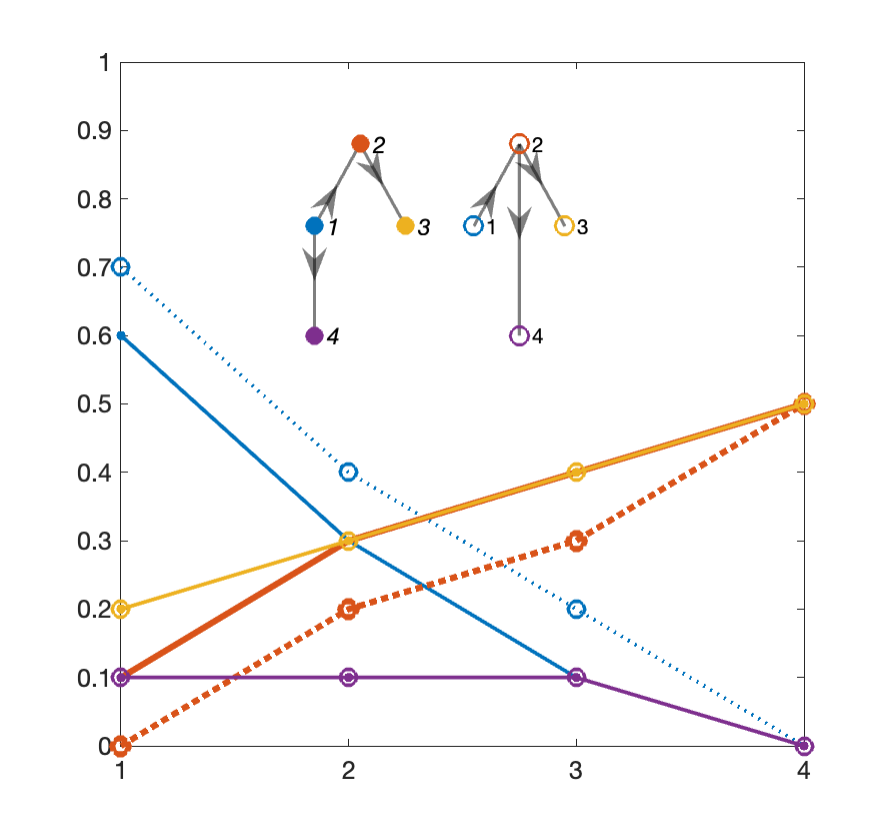}
\put(-238,20){\rotatebox{90}{Relative mutant abundances, $M^*$ and $M$}}
\put(-130,-5){Sample, $t$}
\put(-130,180){$U^*$}
\put(-80,180){$U$}
\put(-123,83){\color{MidnightBlue} ${\bf M}_{1,:}$}
\put(-178,85){\color{MidnightBlue} ${\bf M^*}_{1,:}$}
\put(-186,55){\color{Dandelion} \rotatebox{18}{ ${\bf M}_{3,:}$  ${\bf M}^*_{3,:}$}}
\put(-186,40){\color{Bittersweet} \rotatebox{0}{ ${\bf M}^*_{2,:}$ }}
\put(-126,40){\color{Bittersweet} \rotatebox{0}{ ${\bf M}_{2,:}$ }}
\put(-120,17){\color{DarkOrchid}${\bf M}_{4,:}$  ${\bf M}^*_{4,:}$}
\caption{\rev{At least two solutions, $U^*M^*$ and $UM$, explain $F^*$ in \eqref{eq:example_values_intro_paper_1}. The ancestries $U$ and $U^*$ are shown as a rooted labeled tree. The evolution of mutant abundances $M$ and $M^*$ are shown as a time series.
Solid lines/markers are for $U^*M^*$ and dotted lines/empty markers are for  $UM$.
The solution $UM$ can be excluded by either enforcing that children cannot be born before parents or  restricting the growth rate of $M$. This paper is a first theoretical study of the effect of such types of constraints on the number of solutions for the PPM model.}}
\label{fig:first_example_of_degeneracy}
\end{figure}

As the number of mutants increases or the number of sampling times decreases, the degree of degeneracy can grow rapidly.
Indeed, \cite{pradhan2018non} reports instances with as few as $13$ mutants and $5$ time samples that admit at least $380$ alternative solutions.
They also propose a numerical procedure to bound the number of solutions for a specific instance. This procedure has two steps: first, construct a directed acyclic graph (DAG) from specific and perfect observation data $F^*$,
where edges represent admissible parent-child relationships;
and second, count the spanning trees of this DAG using combinatorial techniques.

To reduce degeneracy, additional constraints may be imposed on the PPM model.
For instance, consider again the example in \eqref{eq:example_values_intro_paper_1} and the solution in which mutant $4$ is inferred as a child of mutant $2$.
The  corresponding ancestral matrix is
$U= [[1  ,   1 ,    1 ,    1];
     [0    , 1   ,  1   ,  1];
     [0   ,  0   ,  1  ,   0];
     [0   ,  0  ,   0   ,  1]]$,
and the associated abundance matrix is
$M = U^{-1} F^* = [[0.7,  0.4, 0.2, 0.0];$ $[
0.0   , 0.2  ,  0.3   , 0.5];$ $
[0.2  ,  0.3  ,  0.4  ,  0.5];$ $
[0.1  ,  0.1  ,  0.1     ,    0.0]]$, where rows are sub-arrays.
This solution exhibits two notable properties that can be inferred from $M$.
First, at time $t=1$, mutant $4$ is present while its infered parent, mutant $2$, has not yet appeared;
mutant $2$ only emerges at time $t = 2$.
Second, the mutant abundances in $M$ vary more rapidly over time than in the ground-truth solution $M^*$.
For example, defining $r(M) = \sum^4_{i=1} \sum^3_{t=1} |M_{i,t+1}-M_{i,t}|$, we obtain $r(M) = 1.6$ for this solution, compared to $r(M^*) = 1.4$.
\rev{In Figure \ref{fig:first_example_of_degeneracy} we see that $r(M^*) < r(M)$ because both mutant $1$ and $2$ change slower (by $0.1$) in the true solution.}

These observations motivate additional constraints on the PPM model that exclude such solutions (e.g., where mutant $4$ is inferred as a child of mutant $2$), and thereby reduce degeneracy.
The first class of constraints, known as \emph{longitudinal conditions} (LC), was introduced in \cite{myers2019calder} and  explicitly forbids a child mutant from \rev{ being born} before its parent.
The second class, which we term \emph{dynamic constraints} (DC), is introduced in this work and restricts the temporal variation of mutant abundances; for example, one may forbid solutions with $r(M) > 1.5$.
Such a constraint would exclude both the solution where mutant $4$ is a child of mutant $2$ and the solution where mutant $4$ is a child of mutant $3$.

Our goal is to theoretically study the degeneracy of \eqref{eq:first_time_F_UM_is_introduced_paper_1} and of two variants aimed at reducing it: one incorporating LC and the other incorporating DC. \rev{Knowing which constraints provably reduce degeneracy guides the design of future inference tools, which should  use the most effective constraints.}

\subsection{Paper organization and summary of main results}

In Section \ref{sec:notation_paper_1} we introduce the notation.
In Section \ref{sec:background_ppm_paper_1} we give background on the PPM model.
In Section \ref{sec:degeneracy_paper_1} we discuss degeneracy in the classical model,
and in Section \ref{sec:longitudional_conditions_definition_paper_1}, we review the LC recently proposed to tackle degeneracy.
 In Section \ref{sec:related_work_paper_1} we review other related work.
In Section \ref{sec:main_results_paper_1}, we present our main results, which constitute, to our knowledge, the first theoretical study of degeneracy under both LC and DC, and in Section \ref{sec:conclusion_paper_1}, we conclude the paper.
Numerical experiments, complete proofs and auxiliary results are provided in the appendices\footnote{All appendices, together with the lemmas, theorems, and equations they contain and that are referenced in this manuscript, appear in the Supplemental Material.}.
Our main theoretical contributions are:
\begin{itemize}[leftmargin=10pt]
\item In Section \ref{sec:redefinition_of_longitudional_conditions_paper_1}, we prove an equivalent reformulation of the LC introduced in \cite{myers2019calder} that is simpler, in that it involves only mutant abundances and ancestral relationships, and not mutation frequencies;
\item In Section \ref{sec:results_about_calder_inefficiency_paper_1}, we introduce an ensemble of problems for which we prove that LC do not reduce degeneracy;
\item In Section \ref{sec:effect_of_dynamic_restrictions_on_degeneracy_paper_1}, we introduce a second ensemble of problems for which we compute a lower bound on the degeneracy of the PPM model.
We then focus on small deviations from the ground truth phylogenetic tree and derive an upper bound on the degeneracy when DC are imposed on the evolutionary trajectories of mutants.
\end{itemize}

\section{Notation} \label{sec:notation_paper_1}

\par\noindent\textbf{Acronyms and global parameters.} PPM stands for Perfect Phylogeny Mixture; LC for Longitudinal Conditions; ELC for Extended Longitudinal Conditions; and DC for Dynamic Constraints. We denote by $q$ the number of mutant types (equivalently, mutations) and by $n$ the number of samples.

\par\noindent\textbf{Indices and indexing conventions.} We use $i$ to index mutant types, mutations, or equivalently nodes in a phylogenetic tree, and $t$ to index samples. We write $[q]\equiv\{1,\dots,q\}$ and $[n]\equiv\{1,\dots,n\}$. Unless stated otherwise, unordered index sets are assumed to be ordered increasingly when used for indexing.

\par\noindent\textbf{Matrices, vectors, and ground truth quantities.} Upper-case letters denote matrices and bold lower-case letters denote vectors. In particular, $F\in\mathbb{R}^{q\times n}$ is the matrix of mutation frequencies and $M\in\mathbb{R}^{q\times n}$ is the matrix of mutant abundances. Bold notation such as ${\bf F}_{i,:}$ and ${\bf F}_{:,t}$ refers to rows and columns of $F$. Ground truth quantities are marked with a superscript $^*$; for example, $F^*$, $M^*$, and $U^*$ denote the true mutation frequencies, abundances, and ancestral relationships.

\par\noindent\textbf{Tree and ancestry notation.} We consider directed labeled rooted trees (arborescences) with node set $[q]$ and root $r$. For a node $i$, $\partial i$ denotes the set of its children, and $\Delta i$ denotes the set consisting of $i$ together with all its descendants. The matrix $U\in\{0,1\}^{q\times q}$ encodes ancestral relationships: $U_{i,j}=1$ if and only if $i=j$ or $i$ is an ancestor of $j$. We also use a parent-operator matrix $T$.
We use $\partial^* i$ and $\Delta^* i$ to denote the corresponding sets in the ground truth tree encoded by $U^*$.

\par\noindent\textbf{Linear algebra conventions.} For a matrix $A$, $A^\top$ denotes its transpose and $A^{-1}$ its inverse when it exists. Inequalities such as $A\geq {\bf 0}$ are interpreted component-wise, e.g. $A \geq {\bf 0}$ means that $A_{i,j} \geq 0,\; \forall i,j$. The identity matrix is denoted by $I$, or by $I_n$ or $I_{n \times n}$ for an $n\times n$  identity matrix when its dimension needs to be specified. The vector ${\bf 1}$ denotes an all-ones vector, with ${\bf 1}_{1\times n}$ and ${\bf 1}_{n\times1}$ indicating its shape when needed.

\par\noindent\textbf{Sets and multisets.} Calligraphic letters denote sets or multisets. For a (multi)set $\mathcal{S}$, $|\mathcal{S}|$ denotes its cardinality, counting multiplicities if $\mathcal{S}$ is a multiset.
We use $\cap$ to denote (multi)set intersection and $\cup$ to denote union. Repeated elements are not removed when performing union involving multisets. We use $-$ to denote (multi)set different. Consider $\mathcal{C} = \mathcal{A}-\mathcal{B}$. If $\mathcal{A}$ has $x$ copies of element $e$ and $\mathcal{B}$ has $y\leq x$ copies of $e$, then $\mathcal{C}$ has $y-x$ copies of $e$. If $y > x$, $\mathcal{C}$ has no $e$. Given a multiset $\mathcal{S}$ and an array $M$, $M_{\mathcal{S}} \equiv \sum_{j\in \mathcal{S}} M_{j}$, where repeated elements are added multiple times in the sum.

\par\noindent\textbf{Subvectors and submatrices.} Given a vector ${\bf v}$ and an ordered index set $\mathcal{S}$, the vector ${\bf v}_{\mathcal{S}}$ consists of the entries of ${\bf v}$ indexed by $\mathcal{S}$ in that order. Given a matrix $A$ and ordered index sets $\mathcal{S}$ and $\mathcal{Q}$, $A_{\mathcal{S},\mathcal{Q}}$ denotes the corresponding submatrix. For $A\in\mathbb{R}^{q\times n}$, we write $A_{:,\mathcal{S}}=A_{[q],\mathcal{S}}$ and $A_{\mathcal{S},:}=A_{\mathcal{S},[n]}$.
If the order of an indexing set $\mathcal{S}$ is not specified, we assume its elements are ordered in increasing order.

\par\noindent\textbf{Probability and norms.}  $\mathbb{P}(E)$ denotes the probability of an event $E$, $\mathbb{E}(X)$ the expectation of a random variable $X$, and $\mathbb{P}(A\mid B)$ the conditional probability of $A$ given $B$. Unless specified otherwise, $|\cdot|$ and $|\cdot|_2$ denote the Euclidean norm for vectors and the Frobenius norm for matrices, while $|\cdot|_1$ denotes the vector $1$-norm or the $1$-norm of a vectorized matrix.

\par\noindent \textbf{Time-related definitions.} For each mutant type $i$, $t_i^{\min}$ and $t_i^{\max}$ denote the birth and death times as defined in \cite{myers2019calder}. The quantities ${t'}_i^{\min}$ and ${t'}_i^{\max}$ denote the corresponding times under our definitions (see Section~\ref{sec:redefinition_of_longitudional_conditions_paper_1}). Using the ground truth matrices $M^*$ and $U^*$, we write ${t^*}'^{\min}_i$ and ${t^*}'^{\max}_i$ for the birth and death times computed according to our definitions, and ${t^*}_i^{\min}$ and ${t^*}_i^{\max}$ for those computed by the definitions of \cite{myers2019calder}.

\section{Background}

\subsection{Inference using the PPM Model}\label{sec:background_ppm_paper_1}

Consider a sample with different types of mutants, possibly comprising multiple clones per type, where each type of mutant has a unique set of mutations that distinguishes it from other types.
The PPM model relates the following three quantities: the (relative) abundance of different mutations, the (relative) abundance of each mutant type, and the ancestral relationships among mutant types.
The PPM model assumes that the descendants of a mutant inherit all of its mutations and have additional mutations. Mutations are not lost, they only accumulate, and the same mutation does not appear in separate lineages. This is called the infinite sites assumption.

If $q$ is the number of mutant types in the sample, among these is a unique mutant type, the \emph{root mutant type}, from which all the other mutants types descend.
We assume that the ancestral relationship between mutants types is a rooted directed tree, sometimes also called an \emph{arborescence}.
For mathematical convenience, we assume that the root mutant type has
a mutation, the \emph{root mutation}, which is shared by all the mutants. We assume that the other mutants accumulate at most $q-1$ new mutations, hence, $q$ mutations in total. We note that a single ``new mutation''  may correspond to changes at multiple genome positions, a detail that we abstract away in this paper.
This allows us to label mutant types and mutations using the same labels. The mutant type $i$ is the mutant with the fewest mutations (i.e. the simplest mutant) that has mutation $i$. Without loss of generality, we use integers to describe mutations and mutant types,  and let $i \in [q] \equiv \{1,\dots,q\}$.

It follows that the fraction (or number) of clones in one sample that have a mutation $i$ is the sum of (a) the fraction (or number) of clones of the simplest mutant type that has mutation $i$ and (b) the fraction (or number) of all its descendants, which according to the PPM model are the only other clones that also have the mutation $i$.
We are interested in inferring
phylogenetic trees from $n$ different population samples, and these relationships hold for each sample.
If we distinguish samples using $t\in[n]$ for the $t$-th sample,
we
have
\begin{align}
F_{i,t} &=  M_{i,t} + \hspace{-0.5cm} \sum_{\substack{j \text{ descendant} \\ \text{of } i}}\hspace{-0.5cm} M_{j,t}  =M_{i,t} +  \sum_{j \neq i: U_{i,j}=1} M_{j,t}   \label{eq:simple_PPM_model_expanded_paper_1}  \\
  &= \sum^q_{j=1} U_{i,j} M_{j,t} = [U {M}]_{i,t} \label{eq:simple_PPM_model_paper_1}, \text{ where}
\end{align}

\begin{enumerate}[leftmargin=10pt]
    \item ${M} \in {\mathbb{R}_0^+}^{q\times n}$; and $M_{i,t}$ is either the fraction of mutant type $i$ in sample $t$, hence $\sum^q_{i=1} M_{i,t} = 1$, or $M_{i,t}$ is the number of clones of mutant type $i$ in sample $t$;
   \item ${F}\in {\mathbb{R}_0^+}^{q\times n}$ and $F_{i,t}$ is the fraction of mutation $i$ in the population in sample $t$, in which case $F_i \leq 1$, or it is the raw count of all mutations $i$ in that sample;
    \item $U \in \{0,1\}^{q\times q}$ is the ancestral matrix, with entries indexed by nodes of the arborescence, which represent mutant types (or equivalently mutations, since they share labels as explained).
    Entry $U_{i,j} = 1$ if and only if $i=j$ or node $j$ is an ancestor of node $i$;
    \begin{itemize}

        \item[--] It follows immediately from this definition that the $j$-th column of $U$ has a $1$ in row $i$ if and only if $i=j$ or if node $i$ is a child of node $j$, or a child of a child of $j$, etc. Hence
        $U = I + {T} + {T}^2 + \dots + {T}^{q-1}$, where ${T}$ is a matrix representation of the operator that takes children to parents in the ancestry tree.
        That is, ${T}$ satisfies ${T} {\bf e}_i = {\bf e}_j$ if $j$ is the parent of $i$, where ${\bf e}_i$ is the $i$-th canonical basis vector.
        \rev{Notice that $(I-T)U=(I-T)(I+T+T^2\dots+T^{q-1})=I-T^q$. Since a node cannot have more than $q-1$ ancestors $T^q=0$. Thus $(I-T)U = I$ and $U = (I-T)^{-1}$}. Ancestral relationships among mutant types, or mutations, can be represented using either $T$ or $U$.
        \item[--] For a fixed $U$ and a node $i$, we define $\Delta i$ as the set consisting of $i$ and all of its descendants, and $\partial i$ as the set of children of $i$;
    \end{itemize}
\end{enumerate}

In this paper, we study the multiplicity of solutions to the following inference problem.
Assume there exists an unknown triple $U^*$, $M^*$ and $F^*$ that satisfies \eqref{eq:simple_PPM_model_paper_1}.  We observe a clean, corrupted, or masked version of $F^*$, call it $F$ to distinguish it from $F^*$. From $F$, we seek to recover $U^*$ and $M^*$.

Model \eqref{eq:simple_PPM_model_paper_1} allows ancestral relationships $U$ with multiple roots, that is, a forest of directed rooted trees. To simplify the exposition, we make the following assumption.
\begin{ass}\label{ass:single_root_trees_assumption_paper_1}
Matrix $U^*$ has a unique root, $r^*$.
\end{ass}
Accordingly, when searching for an alternative $U$ that explains $F$, $U$ must encode a single-root directed tree.

If the $n$ samples are unrelated, we have $n$ separate inference problems.
We focus on the more interesting and useful setting in which one true  $U^*$ is shared across all samples, while the abundance of mutant type $i$ might vary across samples, i.e., $M^*_{i,t}$ might be different from $M^*_{i,t'}$ for $t \neq t'$.

The relationship between the observed $F$ and the truth $F^*$ leads to several inference formulations, we list a few. Assume that $M^*$ represent fractions rather than counts.
\begin{enumerate}[leftmargin=10pt]
    \item If $F=F^*$, i.e. $F^*$ is fully observed, we seek  factorizations of the form $F^* = UM$, subject to ${\bf 1}^{\top} M = {\bf 1}^{\top}$, $M \geq {\bf 0}$, and $U$ being a valid ancestral matrix;
    \item If $q' < q$ rows of $F^*$ are observed, we replace ${\bf 1}^{\top} M = {\bf 1}^{\top}$ by ${\bf 1}^{\top} M \leq {\bf 1}^{\top}$ above.
    During inference, both $F$ and $M$ have $q'$ rows (see Appendix \ref{appendix:PPM_partially_observed_paper_1} for more details);
    \item If $F$ is a corrupted version of $F^*$, we impose the constraint that ${F}$ be close to $UM$, for example by minimizing $\|{F} - U M\|_2$ subject to ${\bf 1}^{\top} M \leq {\bf 1}^{\top}$ and $M \geq {\bf 0}$, where $\|\cdot\|_2$ is the Frobenius matrix norm;
\end{enumerate}

These inference procedures can be interpreted as maximum likelihood, or maximum a posteriori estimation, under  appropriately defined posteriors and priors.
In all cases, the resulting inference problem is computationally hard \cite{ bodlaender1992two, steel1992complexity, el2015reconstruction},
and a myriad of heuristic methods exist.
Commonly used tools include PhyloSub \cite{jiao2014inferring}, PhyloWGS \cite{deshwar2015phylowgs}, AncesTree \cite{el2015reconstruction}, and BEAST \cite{suchard2018bayesian}, along with additional literature cited in \cite{kapli2020phylogenetic}.

\subsection{Degeneracy in the PPM model} \label{sec:degeneracy_paper_1}

Even when we observe the frequency of all mutations in the samples with perfect accuracy, one might find multiple solutions that explain the observations.

In a situation where we work with relative abundances, assume that we observe $F = F^*$ and we want to find a candidate ancestry matrix $U$ to explain it. Given $U$, the candidate $M$ is fixed and equals
$M = {U}^{-1} F^* = F^* - {T} F^*$, and must satisfy conditions ${\bf 1}^{\top} M = {\bf 1}^{\top}$
and $M \geq {\bf 0}$. We can study the degeneracy of this problem by counting how many ancestry matrices $U$ there are for which
\begin{align}
&U^{-1} F^*  \geq 0,\label{eq:first_cond_for_solution_paper_1}\\
&{\bf 1}^{\top} {U}^{-1} F^* = {\bf 1}^{\top}\label{eq:second_cond_for_solution_paper_1}.
\end{align}

Condition \eqref{eq:second_cond_for_solution_paper_1} can be simplified under the following assumption. Recall that $r^*$ is the unique root of $U^*$.
\begin{ass}\label{ass:assumption_no_root_node_is_zero_paper_1}
When doing inference, there is a sample $t$ for which $M^*_{r^*,t} > 0$.
\end{ass}

\begin{lem}\label{th:assumption_2_lemma_paper_1}
Fix $n$ and $q$. Let $F^* = U^* M^*$, where $U^*$ is an ancestry matrix with unique root $r^*$ (Assumption \ref{ass:single_root_trees_assumption_paper_1}) and $M^*$ is a mutant frequency matrix, i.e. $M^* \geq 0,\;{\bf 1}^{\top} M^* = {\bf 1}^{\top}$, satisfying  Assumption \ref{ass:assumption_no_root_node_is_zero_paper_1}.
Let $U$ be an ancestry matrix with a single root $r$ and $M = U^{-1}F^*$. Then
 ${\bf 1}^{\top} M = {\bf 1}^{\top}$ if and only if $r=r^*$.
\end{lem}
\begin{proof}
Let $U = (I - T)^{-1}$ have a unique root $r$ and let ${\bf e}_r$ be the $r$-th canonical basis vector in Euclidean space.
Let $t$ be one instant for which $M^*_{r^*,t} > 0$, which by Assumption \ref{ass:assumption_no_root_node_is_zero_paper_1} we know exists.
We can write ${\bf 1}^{\top} {\bf M}_{:,t} = {\bf 1}^{\top} U^{-1} U^* {\bf M}^*_{:,t}  = {\bf 1}^{\top} (I - T) U^* {\bf M}^*_{:,t}  = {\bf e}_r^{\top} U^* {\bf M}^*_{:,t}= {
\bf U}^*_{r,:} {\bf M}^*_{:,t}$.
If $r = r^*$ then ${\bf U}^*_{r^*,:}= {\bf 1}^{\top}$ and thus we can continue ${
\bf U}^*_{r,:} {\bf M}^*_{:,t}= {\bf 1}^{\top} {\bf M}^*_{:,t} = 1$.
If $r \neq r^*$, note that $U^*_{r,r^*}=0$, and therefore
${
\bf U}^*_{r,:} {\bf M}^*_{:,t}= U^*_{r,r^*} {M}^*_{r^*,t} + \sum_{v \neq r^*} U^*_{r,v} {M}^*_{v,t} \leq \sum_{v \neq r^*} {M}^*_{v,t} = 1 - {M}^*_{r^*,t} < 1$.
\end{proof}
\begin{rmk}\label{rmk:comment_on_removing_sum_condition_in_other_cases_paper_1}
The proof of Lemma \ref{th:assumption_2_lemma_paper_1} extends easily to the case where ${\bf 1}^{\top} M^* = {\bf 1}^{\top}$.
In particular, keeping the lemma's statement mostly unchanged, the following hold:
if $U$ and $U^*$ have unique roots $r$ and $r^*$ and  ${\bf 1}^{\top} {M}^* = {\bf C}^{\top}$, then ${\bf 1}^{\top} {M} = {\bf C}^{\top} \Leftrightarrow  r=r^*$;
if $U$ and $U^*$ have unique roots and ${\bf 1}^{\top} {M}^* \leq {\bf C}^{\top}$,  then ${\bf 1}^{\top} {M} \leq {\bf C}^{\top}$, both if $r=r^*$ or $r \neq r^*$.
\end{rmk}

Lemma \ref{th:assumption_2_lemma_paper_1} tells us that, if Assumptions \ref{ass:single_root_trees_assumption_paper_1} and  \ref{ass:assumption_no_root_node_is_zero_paper_1} hold,
when we study how many solutions $(U,M)$ explain $F^*$ in the scenario where ${\bf 1}^{\top} M^* = {\bf 1}^{\top}$, we can
avoid enforcing \eqref{eq:second_cond_for_solution_paper_1} during inference, and instead enforce \eqref{eq:first_cond_for_solution_paper_1} together with requiring that the root of $U$ satisfies $r=r^*$.

Condition \eqref{eq:first_cond_for_solution_paper_1} can be written as
$
F^*  \geq T F^*,
$
where $T$ is the tree matrix associated with $U$ as described in the third point in Section \ref{sec:background_ppm_paper_1}.
This condition was called the \emph{sum condition} by \cite{pradhan2018non}, where it was rewritten as,
\begin{equation} \label{eq:sum_condition_paper_1}
    {\bf F}^*_{i,:} \geq \sum_{j \in \partial i} {\bf F}^*_{j,:} \;\;\forall i = 1,\dots,q,
\end{equation}
where $ \partial i$ is the set of all children of $i$ in $U$. Equation \eqref{eq:sum_condition_paper_1} implies that for all $i$ and $j$ such that $j\in\partial i$,
\begin{equation}
\label{eq:subset_of_sum_condition_paper_1}
{\bf F}^*_{i,:} \geq {\bf F}^*_{j,:}.
\end{equation}
This is a necessary condition on  $U$ to explain $F^*$.

Condition \eqref{eq:subset_of_sum_condition_paper_1}
admits a graph-theoretic interpretation.
Assume that $F^*$ has no repeated columns. This implies that there is no sequence of distinct mutations $i_1,\dots,i_k$ such that $F_{i_1,:} \geq F_{i_2,:}\geq \dots \geq F_{i_k,:}\geq F_{i_1,:}$.
The directed graph $G_{F^*} = (V_{F^*},E_{F^*})$, whose vertices are the mutations and vertex $i$ connects to vertex $j$ if and only if ${\bf F}^*_{i,:} \geq {\bf F}^*_{j,:}$ is therefore a directed acyclic graph (DAG).
Any solution $U$ must be a
subgraph of $G_{F^*}$, since it must satisfy \eqref{eq:sum_condition_paper_1} and hence also \eqref{eq:subset_of_sum_condition_paper_1}.

Since the observed
$F^*$ can be explained by $U^*$, which is a rooted directed tree on all $q$ nodes, the DAG $G_{F^*}$ contains the true tree $U^*$ and is therefore \emph{weakly connected}, i.e., any node can be reached via the root $r^*$ of $U^*$.
Since $G_{F^*}$ is acyclic, $r^*$ is the unique node from which all nodes are reachable, and hence $G_{F^*}$ itself has a unique root (source) node.
Since any solution must be a subgraph of $G_{F^*}$ and include all $q$ nodes,
any solution $U$ that explains $F^*$ must be a
spanning directed tree of $G_{F^*}$ rooted at $r^*$.
Any
spanning tree of $G_{F^*}$ is not automatically a solution, but it is if it also satisfies \eqref{eq:sum_condition_paper_1}
\footnote{It is possible to generalize the discussion to the case where there are multiple mutants from which all others descend. That is, multiple possible roots. In that case, $G_{F^*}$ contains directed cycles and $F^*$ necessarily has repeated columns.}.

The authors of \cite{pradhan2018non} use this fact to upper bound the number of solutions for a given $F=F^*$.
Any spanning tree of $G_{F^*}$ is constructed by choosing, for each node $i\neq r^*$,  a single parent among the parents of $i$ in $G_{F^*}$.
Therefore, the number of spanning trees of $G_{F^*}$ is the product of the in-degrees of all nodes except $r^*$.
This gives an upper bound on the number of solutions, since not all spanning trees of $G_{F^*}$ satisfy \eqref{eq:sum_condition_paper_1}.

The same work \cite{pradhan2018non} also studies, numerically, how $q$ and $n$ affect the number of solutions, how the software tools PhyloWGS \cite{deshwar2015phylowgs} and CANOPY \cite{jiang2016assessing} handle \emph{degenerate problems} (i.e. problems with multiple solutions), and proposes several experimental avenues to reduce degeneracy.
In Appendix \ref{sec:numerical_paper_1}, we discuss some of these numerical results for completeness.

The PPM model is not time-aware. For example, permuting the $n$ columns of $F^*$
yields the same set of solutions.
The idea of adding time-related constraints to the PPM model is therefore natural, as it aims to capture real-life scenarios where samples are collected over time.
As a bonus side-effect, these additional constraints may also reduce degeneracy.
\subsection{Longitudinal conditions}\label{sec:longitudional_conditions_definition_paper_1}
The work of \cite{myers2019calder}  introduces time-related conditions that couple $M$ and $U$ beyond equation \eqref{eq:simple_PPM_model_paper_1} within the PPM model. These conditions are intuitive and aim to reduce degeneracy.
\begin{enumerate}
    \item Each sample $t$ gains a temporal meaning, namely, the $t$th sample in time;
    the model does not explicitly represent clock time elapsed between samples, although in many settings the time between consecutive samples is constant;
    \item When a mutant dies, it goes permanently extinct; \label{list:logitidional_conditions_english_2_paper_1}
    \item A mutant cannot be born after its parent dies, or before its ancestors are born. \label{list:logitidional_conditions_english_3_paper_1}
\end{enumerate}

The authors in \cite{myers2019calder} define the birth $t^{\min}_v$ time and the death $t^{\max}_v$ time of the mutants of type $v$ as in Definition \ref{def:birth_and_death_time_calder_def_paper_1} \footnote{\rev{Note that $M_{v,t}$ = 0 if and only if $F_{v,t}=\sum_{j\in \partial v} F_{j,t}$, so the stated definition of $t^{\max}_v$ matches that in \cite{myers2019calder}, although it looks different.}}.
\begin{defi}[Birth and death time] \label{def:birth_and_death_time_calder_def_paper_1}
For any mutant $v\in  [q]$,
\begin{align}
    &t^{\min}_v = \min \{t\in[n]: F_{v,t} > 0\} \label{eq:CALDER_t_min_def_paper_1},\\
    &t^{\max}_v = \min\{t\in[n] : t \geq t^{\min}_v\land M_{v,t} = 0\}\label{eq:CALDER_t_max_def_paper_1}
 \end{align}
\end{defi}
After these definitions, they require that the following conditions, called \emph{longitudinal conditions}, hold.
\begin{defi}[Longitudinal conditions (LC)]\label{def:longitudional_conditions_calder_paper_1}
For any mutants $v,w\in [q]$ and sample $t\in[n]$,
\begin{align}
    &t \geq t^{\max}_v \implies M_{v,t} = 0\label{eq:CALDER_extinction_def_paper_1},\\
    &{T}_{v,w} = 1 \implies t^{\min}_w \leq t^{\max}_v \label{eq:CALDER_cont_def_paper_1}.
\end{align}
\end{defi}
The authors in  \cite{myers2019calder}  empirically study the effect of the LC in Definition \ref{def:longitudional_conditions_calder_paper_1} in two ways.
First, they propose and benchmark a phylogenetic tool called CALDER that enforces the LC during inference. They generate simulated data
where the underlying ground-truth ancestry relations are known
and mutants evolve, respecting the LC.
They check whether the reconstructed trees are close to the ground truth or not.
They observe that CALDER outperforms the tools PhyloWGS \cite{deshwar2015phylowgs}
and CITUP \cite{malikic2015clonality}.
Namely, CALDER's median tree error is $0.269$, compared to  $0.297$ for PhyloWGS and $0.552$ for CITUP.
Second, again on simulated data, they apply a modified Gabow--Myers algorithm \cite{gabow1978finding} to compare the number of solutions of the PPM model when the LC are, or not, enforced, under the assumption of error-free sample observations.
In their experiments, most often using the LC yields the same number of solutions as not using them. Conditioned on cases where the LC reduce degeneracy, they observe an  average reduction in the number of solutions by about $30\%$ (see Appendix \ref{sec:numerical_paper_1} for more details).
Lacking in \cite{myers2019calder} are theoretical results about the reduction in degeneracy that the LC induce.

Note that unless
mutants are observed being born or dying, enforcing the LC cannot reduce the number of possible solutions in any way.
Furthermore, accurately computing birth and death times requires accurately estimating when sequences become zero or non-zero, which can be challenging when dealing with noisy data.
In these situations, the effect of the LC on degeneracy might be unclear, or even harmful, if misestimations end up excluding the correct solution.

The authors in \cite{myers2019calder} are aware of this problem. In particular, when there is uncertainty in $F$, and unless precautions are taken, one can easily get numerical solutions where $M > 0$ for all samples, in which case enforcing the LC is not helpful. The authors address this by relaxing the definition of $0$ in their code.
We later show that the efficacy of the LC goes beyond problems where $M>0$ for all samples, and that many instances where the birth and dead times are clearly and precisely observed, might not benefit, in terms of inference, from these longitudinal conditions being enforced.

\subsubsection{Discussion of Definitions \ref{def:birth_and_death_time_calder_def_paper_1} and \ref{def:longitudional_conditions_calder_paper_1}}

\label{sec:discussion_of_definitions_and_their_problems_paper_1}

Definitions \ref{def:birth_and_death_time_calder_def_paper_1} and \ref{def:longitudional_conditions_calder_paper_1} have two minor issues.
First, the quantities \eqref{eq:CALDER_t_min_def_paper_1} and \eqref{eq:CALDER_t_max_def_paper_1} might be undefined under the original definition.
Second, if \eqref{eq:CALDER_t_min_def_paper_1}--\eqref{eq:CALDER_cont_def_paper_1} are not jointly assumed, they do not have the meaning described in the three items above.
Let us discuss each of these problems.

The crux of the first problem is that observations are limited to $t \in [n]$, but the system might have been evolving before $t=1$ and might keep evolving after $t=n$, and certainly evolves between observation times $t$ and $t+1$.

This limited window can lead to undefined birth or death times.
Consider a scenario with only two mutants, $v$ and $w$, where $v$ is the father of $w$.
\begin{itemize}[leftmargin=*]
    \item If ${\bf M}_{v,:} = {\bf 1}$ and ${\bf M}_{w,:} = {\bf 0}$, then $v$ never died during the observation period, and $w$ was never observed being born or dying during the observation period. In this case, $t^{\min}_w$, $t^{\max}_w$ and $t^{\max}_v$ are undefined;
    \item If ${\bf M}_{v,:} = {\bf 0}$ and ${\bf M}_{w,:} = {\bf 1}$, then $v$ was born, produced $w$, and became extinct before $t=1$, and at $t=1$ only $w$ remains, dominating the population. In this case,  $t^{\max}_w$ is undefined;
    \item If ${\bf M}_{w,:} = [0,0,1,1]$ and ${\bf M}_{v,:} = {\bf 0}$ then $t^{\max}_w$ is undefined.
\end{itemize}

Beyond the problem of $t^{\max}$ or $t^{\min}$ not being defined, in the last two examples we have
$t^{\min}_v=t^{\max}_v$ ($=1$ or $=3$, respectively) despite mutant $v$ never being observed. If the clock time  between two consecutive samples is large, this can perhaps be accepted and attributed to a sampling problem: mutant $v$ was born and died between two samples being taken.
However, mathematically, Definition \ref{def:birth_and_death_time_calder_def_paper_1}  allows $t^{\min}_v=t^{\max}_v$ even if the time between samples is very small.

Let us set aside limitations coming from the observation window or discrete sampling, in other words, let us assume that ``existence'' and ``observation of existence'' are one and the same thing.
Regarding the second problem, several issues arise. For example, the definition of $t^{\min}_v$, by itself, is not equivalent to the first instant when the mutant $v$ comes into existence, but rather to the first instant when either it or any of its descendants comes into existence.
Similarly, $t^{\max}_v$ alone is not necessarily equivalent to the first instant when mutant $v$ goes extinct, but to the first instant \emph{after $t^{\min}_v$} when mutant $v$ is observed to be dead.
For example, without assuming \eqref{eq:CALDER_extinction_def_paper_1}--\eqref{eq:CALDER_cont_def_paper_1}, a child $w$ of mutant $v$ can be born before $v$, and hence $t^{\min}_v$ can be smaller than the first time when $v$ appears in the population.
This also makes $t^{\max}_v = t^{\min}_v$, and hence $t^{\max}_v$ can be smaller than the last time when $v$ appeared in the population.
Furthermore, some of the conditions described in the three items above do not follow from the conditions \eqref{eq:CALDER_extinction_def_paper_1}--\eqref{eq:CALDER_cont_def_paper_1}, but rather from the definitions \eqref{eq:CALDER_t_min_def_paper_1}--\eqref{eq:CALDER_t_max_def_paper_1} themselves.
For example, the fact that a mutant $w$ cannot be born before its ancestor $v$ is born follows directly from the definition of $t^{\min}_v$ in \eqref{eq:CALDER_t_min_def_paper_1}.
Indeed, assuming both $t^{\min}_v$ and $t^{\min}_w$ are defined, the PPM model requires that $F_{v,t} \geq  F_{w,t}$ (see equation \ref{eq:sum_condition_paper_1}).

\section{Related work}
\label{sec:related_work_paper_1}

\rev{
While we have already discussed prior work advancing the theory of degeneracy in the PPM model~\cite{pradhan2018non}, here we consider previous work which studies degeneracy, and topics closely related to it, in other evolutionary models. Although these models differ substantially from the PPM model, they provide useful context and points of comparison.}

\rev{
Each of these models specifies a stochastic process for how mutations occur along the edges of an evolutionary tree, whose nodes correspond to \emph{taxa}---distinct biological groupings such as species, families, or classes. Each taxon is described by a sequence over a finite alphabet, where each position is a \emph{character} and its value a \emph{character state}.}

\rev{
These models capture different types of mutational events, including character substitutions, insertions and deletions, or large-scale rearrangements, most of which are not considered by the PPM model\footnote{The PPM model considers only restricted substitutions where a character mutates at most once along the tree.}.
For these models, taxa at the leaves of the tree are observed.
From these observations one seeks to infer the model's parameters, i.e. the mutational process details and the tree topology,  commonly via maximum likelihood estimation, with Bayesian approaches used to quantify uncertainty over parameters.}

\rev{
We are interested in \emph{degeneracy}: the non-uniqueness of the parameters consistent with a given distribution of observations.}

\rev{A model is \emph{identifiable} when its parameters are uniquely determined by that distribution, and \emph{generically identifiable} when this holds outside a measure-zero set of parameters.
}

\rev{
An estimator is \emph{consistent} when its output, including the estimated tree, converges in probability to the truth as the sample size grows; consistency thus means degeneracy vanishes with more data, while non-identifiability precludes consistency.}

\rev{
For a mixture model such as the PPM, a \emph{non-identifiable mixture} is one in which distinct weights and component parameters yield the same overall distribution for a \emph{fixed} number of components (possible even when each component is individually identifiable) whereas a \emph{non-identifiable mixture distribution} is the stronger case, admitting fundamentally different representations, possibly with different numbers of components or structures.
}

\rev{
The \emph{substitution models} cited below are group-based, time-reversible Markov models of per-site character evolution along tree edges, differing in the constraints imposed on substitution rates: Cavender-Farris-Neyman (CFN)~\cite{neyman1971molecular, farris1973probability, cavender1978taxonomy}, Jukes-Cantor (JC69)~\cite{jukes1969evolution}, and the Kimura models (K80, K81)~\cite{kimura1981estimation,kimura1980simple}, together with the broader equivariant strand-symmetric (SMM) and GMM models.
}

\rev{
\emph{Coalescent models}~\cite{marjoram2006modern} instead trace ancestry backwards in time through the merging of lineages, while phylogenetic networks generalize trees by permitting reticulation events, enlarging the model class to recover identifiability where tree-based models cannot.}

\rev{
Further details can be found in Appendix~\ref{app:details_on_related_work}.
}
\subsection{Related results}
The problem of non-uniqueness of solutions in the PPM model~\cite{el2016inferring, malikic2015clonality} and the implications of degeneracy~\cite{qi2019implications} are widely recognized.~\cite{qi2019implications} shows that non-uniqueness is a widespread phenomenon in this setting and that degeneracy is exacerbated by increasing the number of mutations and counteracted by increasing the number of samples, which helps reduce degeneracy by revealing the branching structure of additional mutations.~\cite{qi2019implications} also show that experimental techniques such as long read sequencing and single cell sampling can reduce the size of the \textit{solution space}, and in turn, help with degeneracy, although these experimental techniques do not address the root of the problem, they  lead to a combinatorial reduction in the number of arborescences in general.

Phylogenetic mixture models are useful for analyzing heterogeneous evolution. Prior work~\cite{vstefankovivc2007phylogeny, casanellas2012space, allman2006identifiability, hollering2021identifiability}, study identifiability and consistency under various group-based phylogenetic substitution models. In particular,~\cite{vstefankovivc2007phylogeny} shows that the maximum likelihood tree topology provably differs from the generating tree topology for several  mutation models (JC, K3, CFN, K2) when two trees from the same generating topology with arbitrarily small perturbations to the same transition matrix are mixed and the likelihood is maximized over non-mixture distributions.
For any evolution model whose transition matrices are parameterized by multi-linear polynomials~\cite{vstefankovivc2007phylogeny} also proves that one of the following must hold: either (1) there exists a linear test (a separating hyperplane) that can be used to identify the topology, or (2) there exists mixtures that are fundamentally non-identifiable. ~\cite{vstefankovivc2007phylogeny} also shows that the CFN model has a non-identifiable mixture, the JC and K2 models have no ambiguous mixtures, and the K3 model has a non-identifiable mixture distribution.

Other work~\cite{casanellas2012space} builds upon previous work characterizing identifiability ~\cite{vstefankovivc2007phylogeny, allman2006identifiability} by establishing a non-identifiability an upper bound $h_0(q)$ on the number of mixture components that can be used for equivariant models to be identifiable. For the equivariant models JC69, K80, K81, SSM, GMM with $q$ taxa, they show that h-mixtures are not generically identifiable for $h \geq h_0(q)$. They provide upper bounds that are exponential in the number of taxa for each of the aforementioned models. A recent paper studies identifiability in the network setting~\cite{pardi2015reconstructible}.

Other prior work studies degeneracy in the context of coalescent models.
Previously,~\cite{li2011inference} uses these models to infer a  population size history (i.e. size over time) from the complete diploid genome sequence of a present day human.
For a single population in the same problem setting,~\cite{kim2015can} provides a provable information theoretic lower bound on the number of samples needed be able to distinguish between two population histories. In particular,~\cite{kim2015can} proves a that an exponential number of samples is required to correctly infer a single population history. Along a similar line of work,~\cite{kim2019many} generalizes this analysis to the multiple subpopulation setting with known coalescence times. In this setting, they show that the number of samples is exponential in the number of subpopulations even for recent history.

\section{Main results} \label{sec:main_results_paper_1}
In this section, we present our main results.
Our overarching contribution is to show that a rigorous treatment of degeneracy under the PPM model is possible.
Unlike prior work, which has focused on specific observations $F$ (e.g., \cite{pradhan2018non}), our theorems hold over ensembles of problems. The key idea is to construct these ensembles so that they are simultaneously amenable to analytical treatment while still covering a broad range of realistic problem instances.
Our specific contributions are, first, showing that under perfect observations and  general conditions, imposing the longitudinal conditions (LC) does not reduce degeneracy; and
second, introducing novel conditions that, also under perfect observations, do provably reduce degeneracy.
The theoretical results for these new conditions are more restricted, both because the ensemble of problems considered is more specific and because the degeneracy analysis is restricted to counting ancestries $U \neq U^*$ that are close to $U^*$.
We reiterate that throughout this section we operate under Assumption \ref{ass:single_root_trees_assumption_paper_1}.
No additional assumptions, including Assumption \ref{ass:assumption_no_root_node_is_zero_paper_1}, should be presumed unless explicitly stated or proven to follow from other assumptions.

Prior to presenting these major contributions, we introduce two preliminary ones. The first is a refinement of Definition \ref{def:birth_and_death_time_calder_def_paper_1} to address the issues discussed in Section \ref{sec:discussion_of_definitions_and_their_problems_paper_1}.
This refinement is necessary to rigorously prove degeneracy results for our ensembles of problems without excluding valid problem instances solely due to undefined birth or death times.
The second is to provide a new representation for the LC, and prove that our new definitions are equivalent to Definition \ref{def:longitudional_conditions_calder_paper_1}.
These reformulations are valuable in themselves, and also enable a more direct and transparent analysis of degeneracy,
thereby simplifying the subsequent proofs.

\subsubsection{Birth and death time edge cases} \label{sec:redefine_birth_death_time_original_paper_1}

We extend Definition \ref{def:birth_and_death_time_calder_def_paper_1} so that $t^{\min}_v$ and $t^{\max}_v$ are always well-defined while preserving their intended interpretation.
There is no unique way to achieve this (see Appendix \ref{appendix:alternative_definition_paper_1}).

\begin{defi}[Extended birth and death time] \label{def:extended_birth_and_death_time_calder_def_paper_1}
For any mutant $v\in  [q]$,
\begin{align}
    &t^{\min}_v = \min \{t\in[n]: F_{v,t} > 0\} ,\\
    &t^{\max}_v = \min\{t\in[n] : t \geq t^{\min}_v\land M_{v,t} = 0\},
 \end{align}
 and furthermore,
 \begin{itemize}
\item if $t^{\min}_v$ is undefined, we set $t^{\min}_v = 1$;
\item if $t^{\max}_v$ is undefined,
we set $t^{\max}_v = M_1$, where $M_1 > n$ is a fixed constant.
\end{itemize}
\end{defi}

For convenience, and specifically to reduce the number of special cases we need to consider in our proofs, we extend the LC to exclude solutions in which a mutant is born and dies at the same observation time.
\begin{defi}[Extended longitudinal conditions (ELC)] \label{defi:extended_longitudional_conditions_paper_1}
For any mutants $v,w\in [q]$ and any sample $t\in[n]$,
\begin{align}
    &t \geq t^{\max}_v \implies M_{v,t} = 0,\\
    &\rev{T}_{v,w} = 1 \implies t^{\min}_w \leq t^{\max}_v, \\
    &t^{\min}_v < t^{\max}_v \label{eq:extra_condition_to_add_to_longitudional_conditions_paper_1}.
\end{align}
\end{defi}
In Definition \ref{defi:extended_longitudional_conditions_paper_1}, birth and death times are defined according to Definition \ref{def:extended_birth_and_death_time_calder_def_paper_1}.
Our main results are based on Definition \ref{defi:extended_longitudional_conditions_paper_1}, i.e., the ELC, rather than
Definition \ref{def:longitudional_conditions_calder_paper_1}.

Under Definition \ref{def:extended_birth_and_death_time_calder_def_paper_1},
adding the condition $t^{\min}_v < t^{\max}_v$ is equivalent to imposing
\begin{equation}\label{eq:extra_long_constraint_no_zero_M_paper_1}
{\bf M}_{v,:} \neq {\bf 0},
\end{equation}
as proved in Appendix \ref{app:equivalent_of_extending_calder_long_conditions_in_one_or_in_another_way_paper_1} (Lemma \ref{th:equivalent_of_extending_calder_long_conditions_in_one_or_in_another_way_paper_1}).
Consequently, any solution satisfying the ELC also satisfies ${\bf M}_{v,:} \neq {\bf 0} \, \forall v$.

In the context of recovering the ground truth $U^*$ and $M^*$ from an observed $F$, with or without assuming that the ELC hold, our theoretical results will invoke Assumption \ref{ass:assumption_no_root_node_is_zero_paper_1}.
This assumption plays a role analogous to the non-zero constraint above but applies only to the root mutant,  avoiding the need to enforce ${\bf 1}^{\top} M = {\bf 1}^{\top}$ during inference.

\subsection{Redefinition of birth and death time and of LC} \label{sec:redefinition_of_longitudional_conditions_paper_1}
We now propose new definitions of birth and death times of a mutant $v$ that depend explicitly on transitions of $M_{v,t}$ between zero and non-zero values.
Unlike Definitions \ref{def:birth_and_death_time_calder_def_paper_1} and \ref{def:extended_birth_and_death_time_calder_def_paper_1}, in which computing $t^{\min}_v$ requires checking whether $F_{v,t}>0$ and computing $t^{\max}_v$ depends on $t^{\min}_v$, our formulation makes the dependence on $M$ explicit. This is crucial for our probabilistic analysis of degeneracy.
In particular, it simplifies our analysis because the ensemble of problems considered later is defined via a probability distribution over $M$.
We then re-express the LC introduced in \cite{myers2019calder} so that, together with the new birth and death time definitions, they retain exactly the same meaning as in \cite{myers2019calder}.
This equivalence is formalized in Lemma \ref{th:equivalent_CALDER_conditions_paper_1}.

\begin{defi}[New definition for birth and death time]\label{defi:new_definition_of_birth_and_death_time_paper_1}
Define
    \begin{align}
    &t'^{\min}_v = \max \{t\in[n]: M_{v,t-1} = 0 \land M_{v,t} > 0 \} \label{eq:our_t_min_def_paper_1},\\
    &t'^{\max}_v = \min\{t\in[n] : M_{v,t-1} > 0 \land M_{v,t} = 0 \}\label{eq:our_t_max_def_paper_1},
\end{align}
and extend these definitions with the following special cases:
\begin{itemize}
    \item if ${t'}^{\min}_v$ is undefined, then  ${t'}^{\min}_v = 1$;
    \item if ${t'}^{\max}_v$ is undefined, then
     (i) if $M_{v,n} = 0$ we set $t'^{\max}_v = 1$ and
     (ii) if $M_{v,n} > 0$ we set $t'^{\max}_v = M_1$, where $M_1 > n$ is a constant.
\end{itemize}
\end{defi}

\begin{defi}[Extended longitudinal conditions expressed using new birth and death time]\label{defi:new_definition_of_LC_paper_1}
For all mutants $v$ and $w$,
\begin{align}
&   t'^{\min}_v < t'^{\max}_v,\label{eq:longitudional_condition_bento_1_paper_1}\\
&\rev{T}_{v,w} =1 \implies t'^{\min}_w \in [ t'^{\min}_v ,  t'^{\max}_v]. \label{eq:longitudional_condition_bento_2_paper_1}
 \end{align}
\end{defi}
In Definition \ref{defi:new_definition_of_LC_paper_1}, birth and death times are defined according to Definition \ref{defi:new_definition_of_birth_and_death_time_paper_1}.

\begin{lem} \label{th:equivalent_CALDER_conditions_paper_1}
The conditions in Definitions \ref{defi:extended_longitudional_conditions_paper_1} and  \ref{defi:new_definition_of_LC_paper_1} are equivalent. Furthermore, if these conditions hold, then Definitions \ref{def:extended_birth_and_death_time_calder_def_paper_1} and \ref{defi:new_definition_of_birth_and_death_time_paper_1} are equivalent, i.e. $t'^{\min}_v=t^{\min}_v$ and $t'^{\max}_v=t^{\max}_v$.
\end{lem}

\begin{rmk}
If mutant $v$ is never observed, and hence ${\bf M}_{v,:} = {\bf 0}$, then $t^{\min}_v = t'^{\min}_v = t^{\max}_v = t'^{\max}_v = 1$.
However, both the ELC in Definition \ref{defi:extended_longitudional_conditions_paper_1} and our re-expression of them in Definition \ref{defi:new_definition_of_LC_paper_1} prevent solutions of this kind.
This is a design choice, that we made to simplify the redefinition of birth and death times and the ELC, as well as the proof of Lemma \ref{th:equivalent_CALDER_conditions_paper_1}, which establishes the equivalence between the new and original definitions.
It is possible to avoid excluding these outcomes, but this requires detailing, for example, how to interpret ${\bf M}_{v,:} = {\bf 0}$. It could mean that $v$ is not yet born within the observation window, or that $v$ was born and died before any observation was made, or that $v$ was born and died between sampling periods during the observation window. Such interpretations must be made in conjunction with what happens to other values of $M$, particularly those corresponding to descendants of $v$.
\end{rmk}

\begin{rmk}
According to the ELC, expressed using either Definition \ref{defi:extended_longitudional_conditions_paper_1} or  Definition \ref{defi:new_definition_of_LC_paper_1}, it is possible that, for example, between time $t = 3$ and time $t=4$, a parent mutant $v$ transitions from $M_{v,3} = 1$ to $M_{v,4} = 0$, while a child mutant $w$ takes over the population, transitioning from $M_{w,3} = 0$ to $M_{w,4} = 1$. In this case, $t^{\max}_v = t'^{\max}_v = t^{\min}_w = t'^{\min}_w$.
There is nothing special about $ t = 3$ and $t=4$ in this example. If observation had instead begun at $t = 4$, then this time point would correspond to $ = 1$, yielding a situation in which, ${\bf M}_{v,:} = {\bf 0}$ and ${\bf M}_{w,:} = {\bf 1}$.
This outcome, with ${\bf M}_{v,:} = {\bf 0}$, is excluded by the  definition of the ELC.
\end{rmk}

\begin{rmk}
Lemma \ref{th:equivalent_CALDER_conditions_paper_1} and all associated definitions hold even if $M$ does not represent mutant frequencies that sum to $1$, and even if the entries of $M$ sum to a value strictly less than $1$. Everything remains valid, e.g., if $M$ represents absolute counts.
\end{rmk}

The birth and death times have been redefined in Definition \ref{defi:new_definition_of_birth_and_death_time_paper_1} explicitly in terms of $M$.
However, the redefinition of the ELC in Definition \ref{defi:new_definition_of_LC_paper_1}, although expressed using these redefined birth and death times, does not yet make its dependence on $M$ explicit.
We address this issue with Lemma \ref{th:quasi_convex_long_cond_M_no_eps_paper_1} below.

Let $\mathcal{L}$ be the set of mutant abundances $M$ that satisfy the ELC (Definition \ref{defi:extended_longitudional_conditions_paper_1} or Definition \ref{defi:new_definition_of_LC_paper_1}), and have  non-negative, bounded entries. The matrix $M$ may represent either mutant frequencies or absolute mutant counts.
\begin{lem}\label{th:quasi_convex_long_cond_M_no_eps_paper_1}
$M \in \mathcal{L}$ if and only if $M$ simultaneously satisfies
\begin{align}
&{M}\geq 0, \quad\quad \label{eq:non_negative_M_each_entry_paper_1}\\
&{\bf M}_{v,:} \neq {\bf 0} \;, \forall v,\label{eq:non_zero_M_paper_1}\\
& M_{v,t} = 0 \text{ if } M_{v,t-1} = 0 \land \sum^{t-1}_{k=1} M_{v,k} > 0  \;, \forall v, \; \forall t \geq 2, \label{eq:right_form_M_condition_paper_1}\\
& M_{v,t} + M_{v,t-1} > 0   \text{ if }  M_{w,t-1} = 0 \land { M}_{w,t} >0 , \label{eq:right_birth relative_to_father_birth_paper_1}\\
& \hspace{3.5cm}\forall v,w: \mathcal{T}_{v,w}=1,\; \forall t = 1,\dots,n, \nonumber
\end{align}
where, for $t = 1$, the last condition should be interpreted as $M_{v,1} > 0$ if $M_{w,1} > 0$.
\end{lem}
For a proof, see Appendix \ref{appd:proof_of_quasi_convex_long_cond_M_no_eps_paper_1}.
\subsection{Degeneracy under the ELC} \label{sec:results_about_calder_inefficiency_paper_1}

Our inference problem in this section is to recover $U^*$ and $M^*$ from the observation $F^* = U^* M^*$; that is, we assume noiseless observations and no masking.
We study whether solving this problem while imposing that any recovered $M$ must satisfy the ELC reduces the number of alternative solutions, compared to the case where these conditions are not imposed.
Our results are not stated for a single, albeit arbitrary, instance, but rather for an ensemble of problems described by a distribution of possible $M^*$ and $U^*$ pairs.
This contrasts with the procedure described in \cite{pradhan2018non}, which bounds the number of solutions of the vanilla PPM model for one specific problem.

Our ensemble, described in Assumption \ref{ass:continuity_of_M_for_calder_conditions_degeneracy_computations_paper_1} below, was chosen because it is a natural and analytically tractable choice for which the theoretical study of degeneracy is feasible.

We use the notation that, for any multiset $\mathcal{S}$, $M^*_{\mathcal{S},t} \equiv \sum_{j\in \mathcal{S}} M^*_{j,t}$, where repeated elements are included multiple times in the sum \footnote{This notation is different from ${\bf M}^*_{\mathcal{S},t}$, with bold ${\bf M^*}$, which would represent a vector with entries ${M}^*_{v,t}, v\in \mathcal{S}$, where repeated elements in $\mathcal{S}$ create repeated entries in ${\bf M}^*_{\mathcal{S},t}$.}.

\begin{ass}\label{ass:continuity_of_M_for_calder_conditions_degeneracy_computations_paper_1}
The sequence of observation times $t$ is fixed, finite, and independent of $U^*$ and $M^*$. Both birth and death times take values exclusively from this set of observation times.
Conditioned on $U^*$, $M^*$ is random, non-negative, and satisfies the ELC with probability one. Furthermore, conditioned on $U^*$, the distribution of $M^*$ is such that the probability that the birth time of a child mutant exactly coincides with the death time of its parent is zero.
In addition,  conditioned on $U^*$,
given any two disjoint multisets  of mutants $\mathcal{S}_1$ and $\mathcal{S}_2$ (multisets allow duplicates),
where each multiset contains at least one mutant alive at time $t$,
the probability that the sum of abundances (frequencies or counts) of the mutants in these two multisets (with repeated elements counted repeatedly) is exactly equal is zero.
Mathematically, if $\mathbb{P}(  M^*_{\mathcal{S}_1,t}, M^*_{\mathcal{S}_2,t}>0\mid U^*) > 0$,
then
$$\mathbb{P}( M^*_{\mathcal{S}_1,t} =  M^*_{\mathcal{S}_2,t}\mid   M^*_{\mathcal{S}_1,t}, M^*_{\mathcal{S}_2,t}>0,U^*) = 0.$$
Note that $M^*_{\mathcal{S}_1,t}>0$ means that there $\exists j\in\mathcal{S}_1: M^*_{j,t}>0$, and the same interpretation applies when $\mathcal{S}_1$ is replaced by $\mathcal{S}_2$.
\end{ass}

\subsubsection{Discussion of Assumption \ref{ass:continuity_of_M_for_calder_conditions_degeneracy_computations_paper_1}} \label{sec:discussion_of_assumptions_for_longitudional_conditions_paper_1}

The requirement that $M^*$ satisfies the ELC arises from the fact that
(a) in this paper we assume that the real systems of interest are those in which the LC hold by nature, e.g., systems in which a mutant cannot be born after its parent has died, and
(b) we aim to determine whether imposing the LC reduces the number of non-true solutions found;
therefore, we assume that the true solution is among those found and is not excluded by imposing the LC,
and (c) we avoid the LC because their are almost equal to the ELC but the later simplify our proofs.

For example, if we assume that the temporal sampling is sufficiently fine, and that $M^*$ is derived from discrete samples of a continuous process, where birth and death times are discrete approximations of continuous random variables, the requirement that the birth and death times of a child and its corresponding parent do not exactly coincide is very mild, hence choosing LC or ELC is indifferent.

Since $M^*$ satisfies the ELC, and since the constraints imposed on $M^*$ depend on phylogenetic relationships among mutants, which are encoded in $U^*$, the distribution of $M^*$
depends on $U^*$.
Hence, all distributional statements conditioned on $U^*$.

When the number of reads covering a position in the genome is high, because of instrument random error, and if $M^*$  is a rescaled version of this high number of counts, e.g. a frequencies, it follows that when a mutant is alive, i.e., at any time after its birth and before its death, its abundance is well modeled by an absolutely continuous random variable.
Hence, we assume that the probability of $M_{j,t}$ taking any specific value, conditioned on being non-zero, is zero.
The reason for requiring absolute continuity only when conditioned on $M_{j,t}$ being non-zero is that
at many time points there exists a non-zero probability that a mutant is not yet born or is already extinct, in which case $M_{j,t}$ is exactly zero.
    The last assumption in Assumption \ref{ass:continuity_of_M_for_calder_conditions_degeneracy_computations_paper_1} extends this rationale from the distribution of a single mutant's abundance to the joint distribution of abundances of multiple mutants.
This assumption requires conditioning on the presence of at least one mutant in $\mathcal{S}_1$ and at least one mutant in $\mathcal{S}_2$ being alive.
However, as Lemma \ref{th:ass_1_with_all_non_zero_and_ass_1_with_a_few_non_zero_paper_1} in Appendix \ref{app:last_condition_forces_all_mutants_to_be_alive_is_the_same_as_not_doing_it_paper_1} shows, the assumption could also be stated  by conditioning on all mutants in both $\mathcal{S}_1$ and $\mathcal{S}_2$ being alive.

To be as minimally restrictive as possible,  Assumption \ref{ass:continuity_of_M_for_calder_conditions_degeneracy_computations_paper_1} requires only that the two sums do not coincide. Alternatively, we could have adopted the more restrictive assumption of absolute continuity of the joint density of all living mutants.

\rev{We are now in a position to address the central question regarding the LCs: does enforcing them reduce the number of solutions? The LC were introduced in~\cite{myers2019calder} with the goal of reducing degeneracy by excluding solutions where, for instance, a child mutant is born before its parent. Having shown in Section~\ref{sec:redefinition_of_longitudional_conditions_paper_1} that the ELC are equivalent to the original LC, Theorem \ref{th:expected_number_of_solutions_without_long_cond_and_with_is_the_same_paper_1} below now answers  this question, concretely, the LC do not reduce the average degeneracy.}

\begin{thm}\label{th:expected_number_of_solutions_without_long_cond_and_with_is_the_same_paper_1}
Let $n$ and $q$ be fixed. Consider any joint probability distribution over $U^*$ and $M^*$  that   satisfies Assumption \ref{ass:continuity_of_M_for_calder_conditions_degeneracy_computations_paper_1}.
Define $F^* = U^* M^*$. The expected number of solutions that explain $F^*$ according to the PPM model is the same whether or not the ELC are imposed.
\end{thm}

\begin{proof}[Proof of Theorem \ref{th:expected_number_of_solutions_without_long_cond_and_with_is_the_same_paper_1}]
A solution under the PPM model is any valid ancestry matrix $U$\footnote{Recall that by Assumption \ref{ass:single_root_trees_assumption_paper_1}, which we assume throughout the paper, $U^*$ has a unique root and hence that any  alternative explanation $U$ must also have a unique root.} for which $M = U^{-1} F^*$ is a valid abundance matrix.
Let $E_1(U,F^*)$ denote the event that $U$ explains $F^*$ according to the PPM model, and let $P_1(U,U^*) \equiv \mathbb{P}  ( E_1(U, U^* M^*) \mid U^*,U )$ be the probability that $U$ explains the data $F^* = U^* M^*$, where $U^*$ is a fixed ancestry matrix and $M^*$ is random conditioned on $U^*$.
The expected number of solutions for a random problem is
\begin{align*}
&\mathbb{E}_{U^*,M^*} \hspace{-0.1cm}\left(\sum_{U} \mathbb{I}( E_1(U,F^*) ) \right) \hspace{-0.cm}=\hspace{-0.cm} \sum_{U} \mathbb{E}_{U^*} \mathbb{E}_{M^*\mid U^*} ( \mathbb{I}( E_1(U,U^*M^*) ) )\\
&=  \sum_{U} \mathbb{E}_{U^*} ( \mathbb{P}( E_1(U,U^*M^*)\mid U^*,U))= \sum_{U} \mathbb{E}_{U^*} ( P_1(U,U^*) ).
\end{align*}

Let $E_2(U,F^*)$ be the event that $U$ explains $F^*$ according to the PPM model and that the ELC hold, which by Lemma \ref{th:quasi_convex_long_cond_M_no_eps_paper_1} is the same as  $M = U^{-1} F^*$  satisfying conditions\eqref{eq:non_negative_M_each_entry_paper_1}-\eqref{eq:right_birth relative_to_father_birth_paper_1}.
Let $P_2(U,U^*) \equiv \mathbb{P}  ( E_2(U, U^* M^*) \mid U^*,U )$ be the probability that $U$ explains the data $F^* = U^* M^*$ and the ELC hold conditioned on $U^*$, where $U^*$ is a fixed ancestry matrix and $M^*$ is random.
The expected number of solutions in this case is
\begin{align*}
&\mathbb{E}_{U^*,M^*} \hspace{-0.1cm}\left( \sum_{U} \mathbb{I}( E_2(U,F^*) ) \right)= \sum_{U} \mathbb{E}_{U^*} \mathbb{E}_{M^*\mid U^*}  (\mathbb{I}( E_2(U,U^*M^*) ) )\\
&=  \sum_{U} \mathbb{E}_{U^*} ( \mathbb{P}( E_2(U,U^*M^*)\mid U^*,U)) = \sum_{U} \mathbb{E}_{U^*} ( P_2(U,U^*) ).
\end{align*}

By Theorem \ref{th:probability_of_solution_without_long_cond_and_with_is_the_same_paper_1} below, we have $P_1(U,U^*) = P_2(U,U^*)$.
Hence, the expected number of solutions is identical in both cases.
\end{proof}

\rev{Theorem   \ref{th:expected_number_of_solutions_without_long_cond_and_with_is_the_same_paper_1} proves that enforcing the ELC does not affect the average degeneracy. Theorem \ref{th:probability_of_solution_without_long_cond_and_with_is_the_same_paper_1} below goes far beyond this and shows that, for any particular solution, the chance that it is accepted by a model enforcing the ELC, or not, is the same.}
\begin{thm} \label{th:probability_of_solution_without_long_cond_and_with_is_the_same_paper_1}
Let $n$ and $q$ be fixed. Let $U^*$ and $U$ be fixed ancestry matrices with a unique root.
Let $M^*$ satisfy Assumption \ref{ass:continuity_of_M_for_calder_conditions_degeneracy_computations_paper_1}, and let $F^* = U^* M^*$.
The probability (over $M^*$, for fixed $n,q,U^*$ and $U$) that $U$ explains $F^*$ according to the PPM model is the same whether or not the ELC are enforced.
\end{thm}
\begin{rmk} \label{rmk:proof_of_equal_prob_holds_for_other_constaints}
In Section \ref{sec:background_ppm_paper_1}, we explained that during inference one might enforce the conditions ${\bf 1}^{\top} M = {\bf 1}^{\top}$, or ${\bf 1}^{\top} M \leq {\bf 1}^{\top}$.
More generally, if $M$ encodes for example absolute counts, we might choose to enforce the conditions ${\bf 1}^{\top} M = {\bf C}^{\top}$, or ${\bf 1}^{\top} M \leq {\bf C}^{\top}$.
The proof of Theorem \ref{th:probability_of_solution_without_long_cond_and_with_is_the_same_paper_1} is provided for the case ${\bf 1}^{\top} M = {\bf 1}^{\top}$, but it also holds if ${\bf 1}^{\top} M = {\bf C}^{\top}$, or ${\bf 1}^{\top} M \leq {\bf C}^{\top}$ are enforced.
Indeed, the bird’s-eye view of the proof is as follows. First we show that these conditions can be ignored when proving a key sufficient equivalence relationship. Second, we complete a calculation that does not involve these conditions.
Regardless of whether $M^*$ and $M$ satisfy
${\bf 1}^{\top} M \leq {C}^{\top}$ or ${\bf 1}^{\top} M = {\bf C}^{\top}$, Lemma \ref{th:assumption_2_lemma_paper_1} still holds (cf. Remark \ref{rmk:comment_on_removing_sum_condition_in_other_cases_paper_1}), and hence the proof's first step remains valid.
The second step (proving \eqref{eq:equivalent_step_to_prob_equality_without_sum_condition_paper_1}) also remains valid because it does not involve any of these conditions.
Furthermore, the main step of the proof of Theorem \ref{th:expected_number_of_solutions_without_long_cond_and_with_is_the_same_paper_1} relies on Theorem \ref{th:probability_of_solution_without_long_cond_and_with_is_the_same_paper_1}, and hence it holds regardless of whether $M^*$ and $M$ satisfy
${\bf 1}^{\top} M \leq {C}^{\top}$ or ${\bf 1}^{\top} M = {\bf C}^{\top}$.
\end{rmk}

The proof of Theorem \ref{th:probability_of_solution_without_long_cond_and_with_is_the_same_paper_1} is concise but relies on several auxiliary lemmas, whose statements and proofs appear in Appendix \ref{app:auxiliary_results_for_proving_prob_equiv_calder_paper_1}.
Below, we briefly summarize their results.
The proof is given for inference with the constraint ${\bf 1}^{\top} M = {\bf 1}^{\top}$, but as noted in Remark \ref{rmk:proof_of_equal_prob_holds_for_other_constaints}, it readily extends to other constraints.
\begin{itemize}
\item Lemma \ref{th:M_not_zero_paper_1} shows that Assumption \ref{ass:assumption_no_root_node_is_zero_paper_1} holds with probability one, and hence in the proof we can avoid enforcing that solutions must satisfy ${\bf 1}^{\top} M = {\bf 1}^{\top}$;
\item Lemma \ref{th:M_infered_no_zero_when_Mstart_in_calder_is_non_negative_paper_1} proves that enforcing $M\geq0$ implies \eqref{eq:non_zero_M_paper_1};
\item Lemma \ref{th:new_M_statistifes_death_longitudional_condition_paper_1} proves that enforcing $M\geq0$ implies \eqref{eq:right_form_M_condition_paper_1};
\item Lemma \ref{th:proof_that_M_star_satisfies_t_min_child_t_max_father_relationship_implies_M_also_does_paper_1} proves that enforcing $M\geq0$  implies \eqref{eq:right_birth relative_to_father_birth_paper_1}.
\end{itemize}

\begin{proof}[Proof of Theorem \ref{th:probability_of_solution_without_long_cond_and_with_is_the_same_paper_1}]

An ancestry matrix $U$ explains $F^* = U^* M^*$ according to the PPM model if $M = U^{-1} F^*$ satisfies $M \geq 0$ and ${\bf 1}^{\top} M = {\bf 1}^{\top}$.
To enforce the ELC, and by Lemma \ref{th:quasi_convex_long_cond_M_no_eps_paper_1}, we must further require that \eqref{eq:non_zero_M_paper_1}-\eqref{eq:right_birth relative_to_father_birth_paper_1} hold.
Therefore, to complete the proof, we need to show that
\begin{align}
\mathbb{P}&(M \geq 0 \land {\bf 1}^{\top} M = {\bf 1}^{\top}\mid U^*,U)\nonumber\\
&= \mathbb{P}( M \geq 0 \land {\bf 1}^{\top} M = {\bf 1}^{\top} \land \text{ \eqref{eq:non_zero_M_paper_1}-\eqref{eq:right_birth relative_to_father_birth_paper_1} hold}\mid  U^*,U),\label{th:full_eq_that_long_cond_are_irrelevant_paper_1}
\end{align}
where $M = U^{-1} F^* = U^{-1} U^* M^*$ is a random variable.

Let $r$ and $r^*$ denote the roots of $U$ and $U^*$, respectively.
By Lemma \ref{th:M_not_zero_paper_1}, which follows almost directly from Assumption \ref{ass:continuity_of_M_for_calder_conditions_degeneracy_computations_paper_1}, we have that Assumption \ref{ass:assumption_no_root_node_is_zero_paper_1} holds with probability $1$.
Therefore,
\begin{align*}
\mathbb{P}&( M \geq 0 \land {\bf 1}^{\top} M = {\bf 1}^{\top}\mid U,U^*) = \\ &= \mathbb{P}(  M \geq 0 \land {\bf 1}^{\top} M = {\bf 1}^{\top} \land \text{ Assumption \ref{ass:assumption_no_root_node_is_zero_paper_1} holds}\mid U,U^*) \\
&= \mathbb{I}(r = r^*)\ \mathbb{P}(  M \geq 0  \land \text{ Assumption \ref{ass:assumption_no_root_node_is_zero_paper_1} holds}\mid U,U^*)\\
&=\mathbb{I}(r = r^*)\ \mathbb{P}( M \geq 0\mid U,U^*),
\end{align*}
where we used Lemma \ref{th:assumption_2_lemma_paper_1} in the second equality.
Similarly,
\begin{align*}
    \mathbb{P}( M \geq 0& \land {\bf 1}^{\top} M = {\bf 1}^{\top} \land \text{ \eqref{eq:non_zero_M_paper_1}-\eqref{eq:right_birth relative_to_father_birth_paper_1} hold}\mid U,U^*)\\
    &= \mathbb{I}(r = r^*)\mathbb{P}( M \geq 0 \land \text{ \eqref{eq:non_zero_M_paper_1}-\eqref{eq:right_birth relative_to_father_birth_paper_1} hold}\mid U,U^*).
\end{align*}

Therefore, to show \eqref{th:full_eq_that_long_cond_are_irrelevant_paper_1}, it is sufficient to prove that for $U$ and $U^*$ such that $r = r^*$, we have
\begin{equation}\label{eq:equivalent_step_to_prob_equality_without_sum_condition_paper_1}
\mathbb{P}(M \geq 0\mid U^*,U) = \mathbb{P}( M \geq 0 \hspace{-0.1cm}\land\hspace{-0.1cm} \text{ \eqref{eq:non_zero_M_paper_1}-\eqref{eq:right_birth relative_to_father_birth_paper_1} hold}\mid U^*,U).
\end{equation}
By Lemmas \ref{th:M_infered_no_zero_when_Mstart_in_calder_is_non_negative_paper_1}, \ref{th:new_M_statistifes_death_longitudional_condition_paper_1} and
\ref{th:proof_that_M_star_satisfies_t_min_child_t_max_father_relationship_implies_M_also_does_paper_1}, we obtain
\begin{align}
\mathbb{P}&( M \geq 0 \land  \text{ \eqref{eq:non_zero_M_paper_1}-\eqref{eq:right_birth relative_to_father_birth_paper_1} are not all satisfied} \mid U^*,U)\nonumber\\
&\leq \mathbb{P}( M \geq 0 \land  \text{ \eqref{eq:non_zero_M_paper_1} is not satisfied} \mid U^*,U) \nonumber\\
&\quad+ \mathbb{P}( M \geq 0 \land  \text{ \eqref{eq:right_form_M_condition_paper_1} is not satisfied} \mid U^*,U) \nonumber\\
&\quad+ \mathbb{P}( M \geq 0 \land  \text{ \eqref{eq:right_birth relative_to_father_birth_paper_1} is not satisfied} \mid U^*,U) \nonumber = 0.
\end{align}
Hence, the proof is complete.
\end{proof}
\vspace{-10pt}
\subsection{Degeneracy under the DC}\label{sec:effect_of_dynamic_restrictions_on_degeneracy_paper_1}
We study the effect of imposing DC has on the degeneracy of the PPM model.
Our setting remains the same: we observe $F=F^*=U^* M^*$,
and from this, we seek to infer $U^*$ and $M^*$.
Our DC differ from the LC in Definition \ref{def:longitudional_conditions_calder_paper_1}, introduced by \cite{myers2019calder}.
Unlike their LC, which involve explicit birth and death times, our constraints impose limits on the temporal rate of change of mutant abundance in a population.
In particular,
we require that any solution $M$ satisfies a constraint on the quantity
$$r(M) \equiv d(1,{\bf M}_{:,1})+\sum^n_{t=2} d(t,{\bf M}_{:,t},{\bf M}_{:,t-1}),$$
where $d$ is a given function\footnote{This dynamic constraint is a discrete analog of the action integral $\int^n_0 L(t,{\bf M}_{:,t},\frac{{\rm d}{\bf M}_{:,t}}{{\rm d} t}) {\rm d} t$, where $L$ is the Lagrangian.}. In particular,
impose the natural upper bound $r(M) \leq r(M^*)$.
Our results focus on the choice
$$d(t,{\bf M}_{:,t},{\bf M}_{:,t-1}) = \| {\bf M}_{:,t} - {\bf M}_{:,t-1}  \|^2,$$
which implies that $r(M) = \|MD\|^2$, where $D$ is the linear operator that computes differences between consecutive columns of $M$.
For brevity, we define $\dot{M} \equiv M D$ and $\dot{{\bf M}}_{i,:} \equiv {{\bf M}}_{i,:} D$. Similarly,  $\dot{M}^* \equiv M^* D$ and $\dot{{\bf M}}^*_{i,:} \equiv {{\bf M}}^*_{i,:}$.
To analyze degeneracy under this DC, we count the number of valid mutant abundance matrices $M$ and ancestry matrices $U$, where $U$ has the same unique root as $U^*$, that satisfy
\begin{align}\label{eq:degeneracy_count_with_dynamics_paper_1}
&r(M) \leq r(M^*) \land  M \geq 0 \Leftrightarrow\;\| \dot{M}\| \leq \| \dot{M}^*\| \land  M \geq 0 \\
&\Leftrightarrow\; \| U^{-1} U^* \dot{M}^*\| \leq \| \dot{M}^*\| \land U^{-1} U^* M^* \geq 0. \label{eq:degeneracy_count_with_dynamics_last_line_paper_1}
\end{align}
If Assumption \ref{ass:assumption_no_root_node_is_zero_paper_1} holds, then by Lemma \ref{th:assumption_2_lemma_paper_1} and Remark \ref{rmk:comment_on_removing_sum_condition_in_other_cases_paper_1}, restricting attention to ancestry matrices $U$ with the same unique root as $U^*$ ensures that our results hold regardless of any additional constraints imposed during inference.
In particular, the results hold whether we enforce ${\bf 1}^{\top} M = {\bf 1}^{\top} M^*$ or ${\bf 1}^{\top} M \leq {\bf 1}^{\top} M^*$, and whether or not the model uses relative abundances, i.e., whether ${\bf 1}^{\top} M^* = {\bf 1}^{\top}$.

The number of solutions to \eqref{eq:degeneracy_count_with_dynamics_paper_1} is compared with the number of solutions to
\begin{equation} \label{eq:degeneracy_count_with_NO_dynamics_paper_1}
M \geq 0\; \Leftrightarrow\;  U^{-1} U^* M^* \geq 0,
\end{equation}
which represents the degeneracy of the PPM model in the absence of DC (cf. equation \eqref{eq:first_cond_for_solution_paper_1} in Section \ref{sec:degeneracy_paper_1}).
In this paper, this comparison is performed while restricting $U$ to differ from $U^*$ by the displacement of a single leaf, as formalized in Definition \ref{def:one_left_deviation_paper_1}.
We then study degeneracy for the ensemble of problems described in Assumption \ref{ass:U_star_and_M_star_indep_and_partially_uniform_paper_1}.

\begin{defi}\label{def:one_left_deviation_paper_1}
    Let $\mathcal{D}(U^*)$ denote the set of all ancestry matrices $U$ that differ from $U^*$ by reassigning one leaf of $U^*$ to a different parent in $U$. Define $\mathcal{D}^+(U^*) \equiv \{U^*\} \cup \mathcal{D}(U^*)$ .
\end{defi}
\begin{rmk}
All trees in $\mathcal{D}(U^*)$ and $\mathcal{D}^+(U^*)$ are directed rooted trees with the same root as $U^*$, since the leaf reassignment described in Definition \ref{def:one_left_deviation_paper_1} does not alter the root.
\end{rmk}

\begin{ass}\label{ass:U_star_and_M_star_indep_and_partially_uniform_paper_1}
Fix the number of samples $n$ and the number of mutants $q$.
The ancestry matrix $U^*$ is uniformly sampled from the set of all directed labeled rooted trees with $q > 2$ nodes.
The abundance matrix $M^*$ is non-negative, independent of $U^*$, and satisfies the following properties:
(a) the probability distribution of each mutant's abundance vector ${\bf M}^*_{i,:}$ is identical across all mutants $i$, and
(b) no mutant is absent with positive probability, i.e. $\mathbb{P}({\bf M}^*_{i,:} = {\bf 0}) = 0$.
\end{ass}
\begin{rmk}
If this assumption holds, we can assume that Assumption \ref{ass:assumption_no_root_node_is_zero_paper_1} holds and that, without loss of generality, the root of $U^*$ is $r^* = 1$.
By Lemma \ref{th:assumption_2_lemma_paper_1} and Remark \ref{rmk:comment_on_removing_sum_condition_in_other_cases_paper_1}, this also implies that we do not need to enforce
${\bf 1}^{\top} M =   {\bf 1}^{\top} M^*$ or ${\bf 1}^{\top} M \leq   {\bf 1}^{\top} M^*$ during inference.
\end{rmk}

Before we state and prove our main result for this section, we present  intermediate results, whose proofs are provided in Appendix \ref{app:definition_of_increments_paper_1}. Below, $F(q) = (q-1)(q-2)(1-1/q)^{q-2}$.

\begin{lem}\label{th:analytical_bound_on_degeneracy_middle_paper_1}
    Let $U^*$ and $M^*$ satisfy Assumption \ref{ass:U_star_and_M_star_indep_and_partially_uniform_paper_1}.
    Let $E$ and $E'$ denote the expected number of trees in $\mathcal{D}^+(U^*)$ that satisfy \eqref{eq:degeneracy_count_with_NO_dynamics_paper_1} and  \eqref{eq:degeneracy_count_with_dynamics_last_line_paper_1}, respectively.
    We have that
    \begin{align}
    &E = 1 + F(q) \mathbb{P}({{\bf M}^*}_{1,:} \leq {{\bf M}^*}_{2,:} ),\label{eq:degeneracy_without_regu_for_one_leaf_perturb_paper_1}\\
    &E'\leq 1 + F(q)
    \mathbb{P}(  {({\dot{{\bf M}}^*}_{1,:})}^{\top}(\dot{{\bf M}}^*_{1,:} - 2\dot{{\bf M}}^*_{2,:}) \leq 0).\label{eq:degeneracy_with_regu_for_one_leaf_perturb_paper_1}
    \end{align}

    When DC are enforced, the expected number of degenerate solutions in $\mathcal{D}(U^*)$ is reduced by at least the factor
\begin{equation}\label{eq:ration_of_degeneracy_decrease_while_dynamic_constraints_general}
    {\mathbb{P}({{\bf M}^*}_{1,:} \leq {{\bf M}^*}_{2,:} )}\Big /{\mathbb{P}(  {({\dot{{\bf M}}^*}_{1,:})}^{\top}(\dot{{\bf M}}^*_{1,:} - 2\dot{{\bf M}}^*_{2,:}) \leq 0)}.
    \end{equation}
\end{lem}

Even if our bound \eqref{eq:ration_of_degeneracy_decrease_while_dynamic_constraints_general} is greater than $1$,  enforcing \eqref{eq:degeneracy_count_with_dynamics_last_line_paper_1} instead of \eqref{eq:degeneracy_count_with_NO_dynamics_paper_1} during inference always beings the number of solutions down. Our intuition is that, under Assumption \ref{ass:U_star_and_M_star_indep_and_partially_uniform_paper_1},the ratio \eqref{eq:ration_of_degeneracy_decrease_while_dynamic_constraints_general} will be small.
In particular, we expect that the instantaneous rates of change in abundance for any two mutants are not strongly correlated, in the sense that $({\dot{{\bf M}}^*}_{1,:})^{\top} \dot{{\bf M}}^*_{2,:}$ is typically much smaller than $\|{\dot{{\bf M}}^*}_{1,:}\|^2$.
Hence, the denominator of \eqref{eq:ration_of_degeneracy_decrease_while_dynamic_constraints_general} is significantly smaller than its numerator, implying a substantial reduction in degeneracy.

To obtain a concrete result from Lemma \ref{th:analytical_bound_on_degeneracy_middle_paper_1}, we must specify a model for ${\bf M}^*$.
For analytical tractability, we consider a continuous-time formulation in this section.
Although time is often treated as discrete -- since experimental measurements are typically taken at discrete time points and discrete formulations simplify algorithmic development and inference -- biological processes evolve continuously in time.
The PPM model remains valid under continuous-time.
With this in mind, we introduce Assumption \ref{ass:brownian_motion_paper_1} where mutants evolve in continuous-time.
We use ${\bf M}^*_{1,:}$ and ${\bf M}^*_{2,:}$ to represent the trajectories of the processes for mutant types $1$ and $2$ over the time interval $[0,n]$, where $n$ now is a real positive number that measures the amount of sample data we have.
We denote by ${M}_{1,t=0}$ and ${M}_{2,t=0}$ the initial samples (we sometimes omit ``$t=$'' for brevity). Before, when in discrete time, we used $t=1$ for the first sample, $t=2$ for the second sample, and so on.
Assumption \ref{ass:brownian_motion_paper_1} is not claimed to describe the only or most realistic biological model; however,
it is a reasonable
choice that allows us to study degeneracy theoretically.
\begin{ass}\label{ass:brownian_motion_paper_1}
Fix the number of samples $n$ and mutants $q$.
The ancestry matrix $U^*$ is uniformly sampled from the set of all directed labeled rooted trees with $q > 2$ nodes.
Define the stochastic processes ${\bf \tilde{M}}_{i,:}$, $i \in [q]$, as i.i.d. Brownian motions, independent of $U^*$.
The starting point of each $\tilde{M}_{i,:}$ is $\tilde{M}_{i,t=0}$. The set of all starting points are i.i.d. uniformly distributed on $[0,1]$.
Conditioned on the starting point, each $\tilde{M}_{i,t}$ has variance $t$.
Define each process ${\bf M}^*_{i,:}$ as the distribution of ${\bf \tilde{M}}_{i,:}$, conditioned on the event: $\tilde{M}_{i,t} \geq 0$ for all $t\in[0,n]$.
Note that under this assumption $M^*_{i,t=0} = \tilde{M}_{i,t=0}$ (in distribution), and Assumption \ref{ass:U_star_and_M_star_indep_and_partially_uniform_paper_1} holds.
\end{ass}
\vspace{-6pt}
Let us now discuss how Lemma \ref{th:analytical_bound_on_degeneracy_middle_paper_1} applies under the continuity assumption in Assumption \ref{ass:brownian_motion_paper_1}.
The probability of the event ${\bf M}^*_{1,:} \leq {\bf M}^*_{2,:}$ (cf. \eqref{eq:degeneracy_without_regu_for_one_leaf_perturb_paper_1}) is now
$\mathbb{P}({ M}^*_{1,t} \leq { M}^*_{2,t} \forall t \in [0,n])$.
The probability of the event ${({\dot{{\bf M}}^*}_{1,:})}^{\top}(\dot{{\bf M}}^*_{1,:} - 2\dot{{\bf M}}^*_{2,:}) \leq 0$ is now
\vspace{-12pt}
\begin{align*}
&\lim_{\substack{ N \to \infty \\ \delta =n/N }} \mathbb{P}\Big( \sum^N_{i = 1}  (M^*_{1,\delta(i+1)} -M^*_{1,\delta i})\\
&\qquad\quad\times( ({M^*_{1,\delta(i+1)} -M^*_{1,\delta i}}) - 2({M^*_{2,\delta(i+1)} -M^*_{2,\delta i}}))\leq 0\Big)\\
&=\lim_{\substack{ N \to \infty \\ \delta =n/N }} \mathbb{P}\bigg(n \frac{1}{N}\sum^N_{i = 1}  \frac{M^*_{1,\delta(i+1)} -M^*_{1,\delta i}}{\sqrt{\delta}}\\
&\qquad\quad\times\bigg( \frac{M^*_{1,\delta(i+1)} -M^*_{1,\delta i}}{\sqrt{\delta}} - 2\frac{M^*_{2,\delta(i+1)} -M^*_{2,\delta i}}{\sqrt{\delta}}\bigg)\leq 0\bigg)\\
& = \lim_{\substack{ N \to \infty \\ \delta =n/N }} \mathbb{P}\Big(\frac{1}{N}\sum^N_{i = 1}  Z_{1,\delta i}\left( Z_{1,\delta i} - 2Z_{2,\delta i}\right)\leq 0\Big),\numberthis\label{eq:equiv_form_for_dyn_constraints_cont_time_paper_1}
\end{align*}
where $Z_{1,t} \hspace{-0.05cm}=\hspace{-0.05cm}  {(M^*_{1,t + \delta} -M^*_{1,t})}/{\sqrt{\delta}}$, $Z_{2,t} \hspace{-0.05cm}=\hspace{-0.05cm}  {(M^*_{2,t + \delta} -M^*_{2,t})}/{\sqrt{\delta}}$.

The following results holds, whose proof is in Appendix \ref{app:definition_of_increments_paper_1}.
\begin{lem} \label{th:bound_on_prob_M1_less_M2_under_brownian_motion_paper_1}
If Assumption \ref{ass:brownian_motion_paper_1}
holds and $n>0$ then
\begin{align}\label{eq:shrinking_with_n}
   \mathbb{P}({\bf M}^*_{1,:} \leq {\bf M}^*_{2,:} ) =
{1}/{36 n} +O \left({1}/{n^2}
\right).
\end{align}
\end{lem}

\begin{rmk}
Lemma \ref{th:analytical_bound_on_degeneracy_middle_paper_1} and Lemma \ref{th:bound_on_prob_M1_less_M2_under_brownian_motion_paper_1} taken together
imply that the PPM model degeneracy within the solution set $\mathcal{D}(U^*)$ is bounded by $O(q^2/n)$.
\end{rmk}

\begin{lem}\label{th:ergodic_expression_goes_to_one_paper_1}
If Assumption \ref{ass:brownian_motion_paper_1}
holds and $n>0$ then
\begin{align*}
&\mathbb{P}\left(
{\left({\dot{{\bf M}}^*}_{1,:}\right)}^{\top}\left(\dot{{\bf M}}^*_{1,:} - 2\dot{{\bf M}}^*_{2,:}\right)  \leq 0
\right) \\
&\quad\equiv \lim_{\substack{ N \to \infty \\ \delta =n/N }} \mathbb{P}\left(\frac{1}{N}\sum^N_{i = 1}  Z_{1,\delta i}\left( Z_{1,\delta i} - 2Z_{2,\delta i}\right)\leq 0\right)= 0.
\end{align*}
\end{lem}

\begin{rmk} \label{rmk:dynamic_constraints_remove_degeneracy_completely_paper_1}
Lemma \ref{th:ergodic_expression_goes_to_one_paper_1} and Lemma \ref{th:analytical_bound_on_degeneracy_middle_paper_1} taken together
imply that adding DC to the PPM model eliminates degeneracy within the solution set $\mathcal{D}(U^*)$.
\end{rmk}

\begin{proof}[Proof of Lemma \ref{th:ergodic_expression_goes_to_one_paper_1}]
Let $X = \frac{1}{N}\sum^N_{i = 1}  Z_{1,\delta i}\left( Z_{1,\delta i} - 2Z_{2,\delta i}\right)$ and $\tilde{X} = \frac{1}{N}\sum^N_{i = 1}  \tilde{Z}_{1,\delta i}\left( \tilde{Z}_{1,\delta i} - 2\tilde{Z}_{2,\delta i}\right)$, where $Z_{i,t}$ is defined previously and $\tilde{Z}_{i,t} =  \frac{\tilde{M}_{i,t + \delta} -\tilde{M}_{i,t}}{\sqrt{\delta}}$ for $i\in\{1,2\}$.

We have that
\begin{align*}
&\mathbb{P}(X \leq 0) = \mathbb{E}(\mathbb{P}(X \leq 0 \mid {\bf M}^*_{:,0})) \\
&=  \mathbb{E}(\mathbb{P}(\tilde{X} \leq 0 \land {\tilde{M} \geq 0} \mid {\bf M}^*_{:,0}) / \mathbb{P}(\tilde{M} \geq 0\mid {\bf M}^*_{:,0})) \\
&\leq \mathbb{E}(\mathbb{P}(\tilde{X} \leq 0  \mid {\bf M}^*_{:,0}) / \mathbb{P}(\tilde{M} \geq 0\mid {\bf M}^*_{:,0}) \mathbb{I}_{{\bf M}^*_{:,0} > C}) + \mathbb{P}({\bf M}^*_{:,0} \leq C)\\
&= \mathbb{E}(\mathbb{P}(\tilde{X} \leq 0  \mid {\bf M}^*_{:,0}) / \mathbb{P}(\tilde{M} \geq 0\mid {\bf M}^*_{:,0}) \mathbb{I}_{{\bf M}^*_{:,0} > C}) + C,
\end{align*}
for any $0\leq C\leq 1$, and where ${\bf M}^*_{:,0} > C$ (resp. $\leq C$) means ${M}^*_{i,0} > C \, \forall i\in[q]$ (resp. $\leq C$).

Conditioned on fixed initial abundances ${\bf M}^*_{:,t=0}$, and for any finite $\delta>0$, the variables $\tilde{Z}_{1,\delta i}$ and $\tilde{Z}_{2,\delta i}$ are Brownian motion increments, are independent and identically distributed across $i$, and follow $\mathcal{N}(0,1)$.
Hence,
\[\mathbb{E}(\tilde{X}) = \mathbb{E}\left[\tilde{Z}_{1,n}\left(\tilde{Z}_{1,n} - 2\tilde{Z}_{2,n}\right)\right] = 1 \text{ and }\]
\[\mathbb{V}(\tilde{X}) = (1/N^2) \sum^N_{i=1}\mathbb{V}\left[\tilde{Z}_{1,\delta i}\left(\tilde{Z}_{1,\delta i} - 2\tilde{Z}_{2,\delta i}\right)\right] = 6/N.\]
A simple application of Chebyshev's inequality implies that
\[\mathbb{P}(\tilde{X} \leq 0 \mid {\bf M}^*_{:,0}) \leq 6/N = 6\delta /n.\]
We also have that $\mathbb{P}(\tilde{M} \geq 0\mid {\bf M}^*_{:,0}) = \text{Erf}(m_0/\sqrt{n})$, which is increasing in $m_0$. This is well known result in the literature from the distribution of the hitting time of a standard Brownian on a boundary, see e.g. \cite{borodin2012handbook}.

Hence $\mathbb{P}(X \leq 0) \leq 6\delta /(n\text{Erf}(C/\sqrt{n})) + C$.
Choosing $C = \sqrt{\delta}$ and letting $\delta \to 0$ yields the desired result.
\end{proof}

Using Lemmas \ref{th:analytical_bound_on_degeneracy_middle_paper_1}, \ref{th:ergodic_expression_goes_to_one_paper_1}, and \ref{th:bound_on_prob_M1_less_M2_under_brownian_motion_paper_1},
we prove the following main result.
\begin{thm}\label{th:main_theorem_for_our_dynamic_constraints_paper_1}
 Let $U^*$ and $M^*$ satisfy Assumption  \ref{ass:brownian_motion_paper_1}.
Let $E$ and $E'$ denote the expected number of trees in $\mathcal{D}^+(U^*)$ that satisfy \eqref{eq:degeneracy_count_with_NO_dynamics_paper_1} and  \eqref{eq:degeneracy_count_with_dynamics_last_line_paper_1}, respectively.
For large $n$, the ratio $E/E'$ satisfies
\begin{equation}\label{eq:ratio_of_degeneracy_paper_1}
{E}/{E'}\geq {(q-1)(q-2)}/({36 e n}).
\end{equation}
\end{thm}

\begin{proof}[Proof of Theorem \ref{th:main_theorem_for_our_dynamic_constraints_paper_1}]
The expected number of trees in  $\mathcal{D}^+(U^*)$ satisfying \eqref{eq:degeneracy_count_with_dynamics_last_line_paper_1} is bounded by the expected number of trees in  $\mathcal{D}^+(U^*)$ that satisfy only the first term in \eqref{eq:degeneracy_count_with_dynamics_last_line_paper_1}.
By equation \eqref{eq:degeneracy_with_regu_for_one_leaf_perturb_paper_1} in Lemma \ref{th:analytical_bound_on_degeneracy_middle_paper_1}
and by Lemma
\ref{th:ergodic_expression_goes_to_one_paper_1}
, this expectation is upper bounded by
$$1 + (q-2)(q-1)(1-1/q)^{q-2} \times 0 = 1.$$

Meanwhile, using \eqref{eq:degeneracy_without_regu_for_one_leaf_perturb_paper_1} in Lemma \ref{th:analytical_bound_on_degeneracy_middle_paper_1}, and Lemma \ref{th:bound_on_prob_M1_less_M2_under_brownian_motion_paper_1},
the expected number of trees in  $\mathcal{D}^+(U^*)$ satisfying \eqref{eq:degeneracy_count_with_NO_dynamics_paper_1} is
$$1 + (q-2)(q-1)(1-1/q)^{q-2} \times (1/(36n) + O(1/n^2)),$$
which for large $n$ is lower bounded by \eqref{eq:ratio_of_degeneracy_paper_1}.
\end{proof}

\begin{rmk}
\rev{
The uniform prior on $U^*$ gives equal importance to all ground-truth ancestries in the analysis. However, the bound in Theorem 3 depends on the prior only through $\mathbb{E} \left(L(U^*)\right)$, the expected number of leaves of $U^*$. In particular, for any prior distribution $\mathbb{P}(U^*)$ (independent of $M^*$) for which $\mathbb{E}(L(U^*)) < Cq$, a bound of the same order holds. More generally, for any distribution, the bound holds with $q-1$ replaced by $C \cdot \mathbb{E}(L(U^*))$.
}
\end{rmk}

\section{Conclusion and future work} \label{sec:conclusion_paper_1}
We have proved that the longitudinal conditions (LC) introduced in \cite{myers2019calder} are ineffective at reducing degeneracy in the PPM model, in a perfect-observation setting (POS) and under some mild technical assumptions. We have also proved that in a POS, and under more stringent assumptions than for our first result, a different type of dynamic constraint (DC) is very effective at reducing degeneracy. We did an exhaustive numerical enumeration of solutions to extend our claim about the efficacy of the DC beyond these assumptions, but still in a POS.
Although we are not the first to study the degeneracy of the PPM, or its degeneracy under the LC, to the best of our knowledge we are the first to do a theoretical study for probabilistic ensembles of problems.

We recommend a few directions in which to extend our results in the future.
First, in Assumption \ref{ass:continuity_of_M_for_calder_conditions_degeneracy_computations_paper_1} the probability that a parent and child are born at the time could become some probability $p_1 \in [0,1]$ instead of $0$, and the probability that two sets of disjoints mutants have exactly the same total abundance could become some probability $p_2 \in [0,1]$ instead of $0$.
Afterwards, one could adapt the proof of
Theorem \ref{th:expected_number_of_solutions_without_long_cond_and_with_is_the_same_paper_1} such the difference between the expected number of solutions with and without the LC would be a function of $p_1$ and $p_2$.

Second, and following a similar reasoning, the probabilistic assumptions in Assumption \ref{ass:U_star_and_M_star_indep_and_partially_uniform_paper_1} and \ref{ass:brownian_motion_paper_1} could be relaxed such that there would be a controlled correlation between the evolution of different mutants, and a controlled correlation between the evolution of the same mutant over different points in time.
Afterwards, one could adapt the proof of
Lemma \ref{th:analytical_bound_on_degeneracy_middle_paper_1} and Theorem \ref{th:main_theorem_for_our_dynamic_constraints_paper_1} such that the bound on the expected number of solutions with and without DC would reflect these changes.

Third, our results regarding the effectiveness of the DC are restricted to alternative ancestry trees that are very similar to the true tree. Although these trees are the ones that most likely will confuse existing tools, and hence their study is very relevant, one could use heavier combinatorial machinery to investigate the situation where all possible trees are considered, both similar and dissimilar to the ground truth tree.

Finally, one could study how DC and LC compare with the classical PPM model in a setting where the observations are corrupted or masked. One could try to study this theoretically, and/or numerically. For this later goal, one would have to develop a practical tool that could infer phylogenetic trees for mutation frequency data under the DC.

\printbibliography[notcategory=appendixonly,title={References}]

\newpage
\begin{appendices}

\section{Numerical results} \label{sec:numerical_paper_1}

In Theorem \ref{th:main_theorem_for_our_dynamic_constraints_paper_1} we analyze the effect of DC on degeneracy under the condition of small deviations from the true ancestry.
This contrasts with Theorem \ref{th:expected_number_of_solutions_without_long_cond_and_with_is_the_same_paper_1} which does not impose such restrictions.
In this section, we numerically investigate degeneracy both with and without these additional constraints, allowing all possible trees to be considered, and compare the resulting counts with our theoretical predictions.
\begin{figure*}[h!]
    \centering
    \begin{picture}(500,190)
    \setlength{\unitlength}{1pt}
    \includegraphics[width=0.33\textwidth]{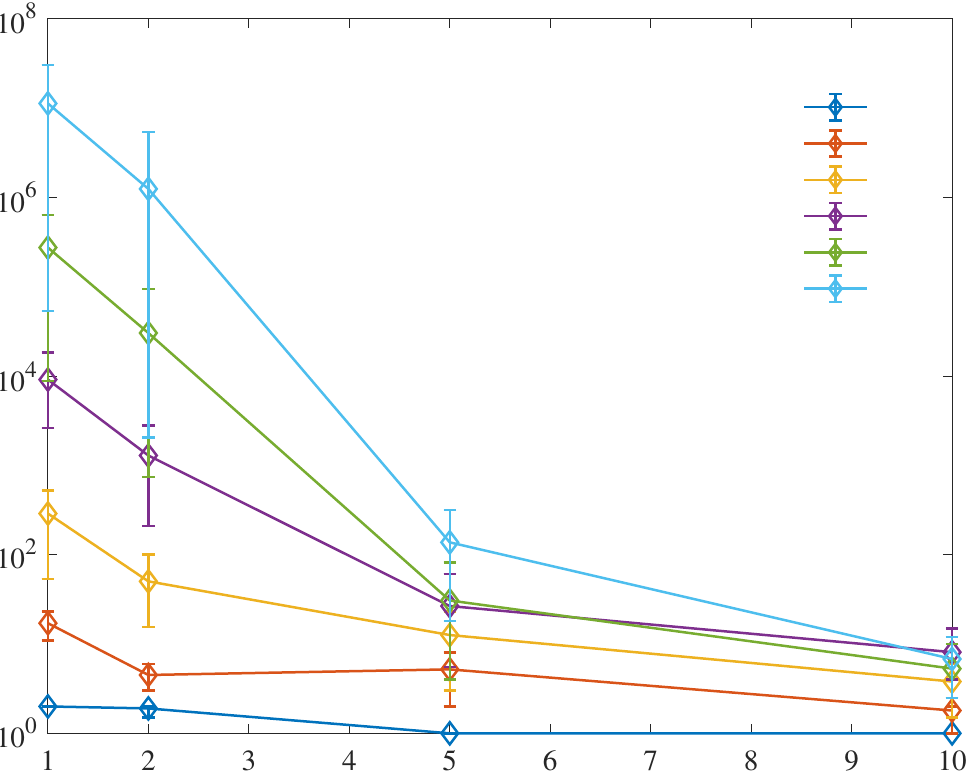}
    \put(-85, -5){\makebox(0,0)[c]{\fontsize{9}{11}\selectfont $n$ samples}}
    \put(-178,65){\rotatebox{90}{\makebox(0,0)[c]{\fontsize{9}{11}\selectfont Average \# of solutions }}}
    \put(-45, 117){\makebox(0,0)[c]{\fontsize{7}{9}\selectfont $q=3$}}
        \put(-45, 110){\makebox(0,0)[c]{\fontsize{7}{9}\selectfont $q=5$}}
    \put(-45, 104){\makebox(0,0)[c]{\fontsize{7}{9}\selectfont $q=7$}}
    \put(-45, 97){\makebox(0,0)[c]{\fontsize{7}{9}\selectfont $q=9$}}
    \put(-43, 91){\makebox(0,0)[c]{\fontsize{7}{9}\selectfont $q=11$}}
    \put(-43, 84){\makebox(0,0)[c]{\fontsize{7}{9}\selectfont $q=13$}}
    \includegraphics[width=0.33\textwidth]{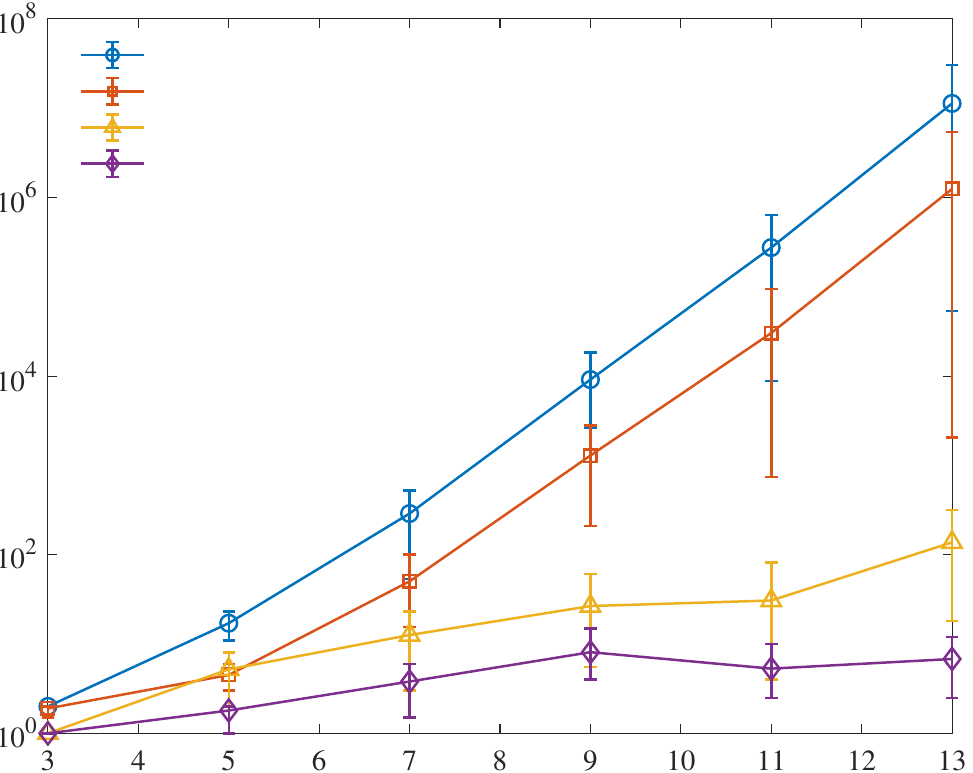}
    \put(-95, -5){\makebox(0,0)[c]{\fontsize{9}{11}\selectfont $q$ mutations}}
    \put(-133, 127){\makebox(0,0)[c]{\fontsize{7}{9}\selectfont $n=1$}}
        \put(-133, 120){\makebox(0,0)[c]{\fontsize{7}{9}\selectfont $n=2$}}
    \put(-133, 113){\makebox(0,0)[c]{\fontsize{7}{9}\selectfont $n=5$}}
    \put(-131, 107){\makebox(0,0)[c]{\fontsize{7}{9}\selectfont $n=10$}}
\includegraphics[width=0.33\textwidth]{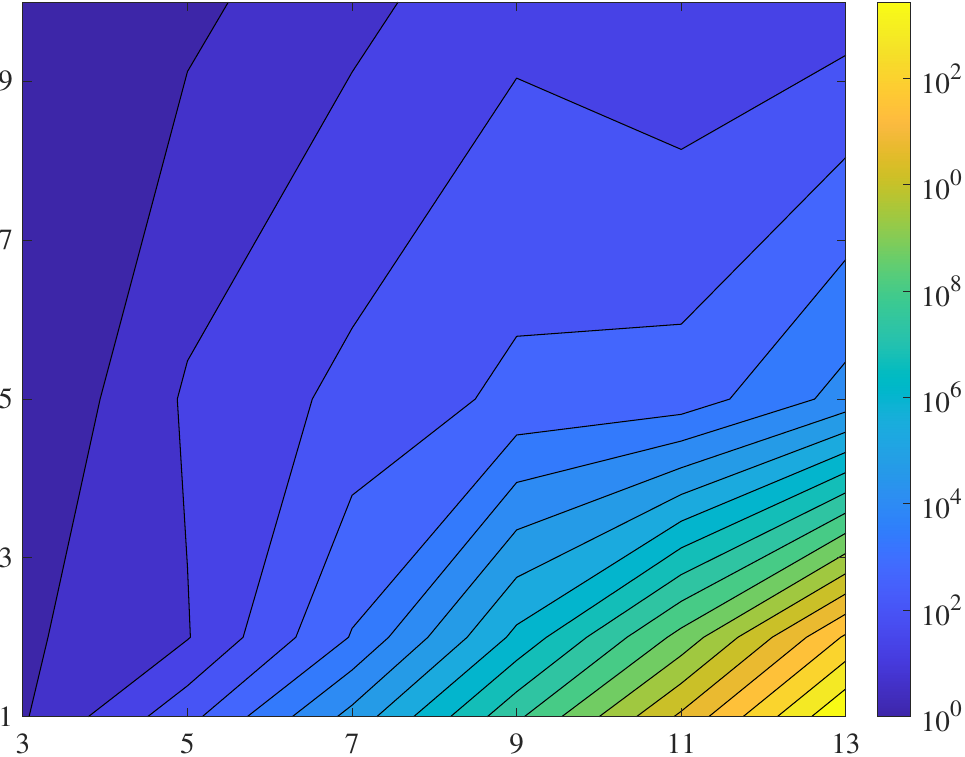}
    \put(-95, -5){\makebox(0,0)[c]{\fontsize{9}{11}\selectfont $q$ mutations}}
    \put(-160,65){\rotatebox{90}{\makebox(0,0)[c]{\fontsize{9}{11}\selectfont \color{white} $n$ samples}}}
    \end{picture}
    \vspace{0.3cm}
    \caption{Empirical results obtained from~\cite{pradhan2018non}, reorganized for presentation clarity and convenience. All plots track the (approximate) average number of solutions under uniform sampling from the solution space of random problems as we vary either the number of samples $n$ (left plot), the number of mutations $q$ (center plot), or both (right plot).
    Each point represents the mean of $10$ different random problems, i.e. random $F^*$. Error bars on the first two plots cover the $10$-th and $90$-th percentiles. On the right-most plot, a color map is used to indicate the (approximate) average number of solutions changes with both $q$ and $n$.}
    \label{fig:empirical_results}
\end{figure*}

\subsection{Prior empirical results on the degeneracy of the PPM model}

Prior work~\cite{pradhan2018non} employed a tumor evolution simulator~\cite{el2018inferring} to generate synthetic data
$F^*$, observed without noise or masking, with the goal of counting the number of possible explanations.
In this simulator, $U^*$ and $M^*$ are jointly generated, and $F^*$ is derived through a process that models the sequencing procedure.
To count solutions, an MCMC rejection sampling approach is used, which provides uniform sampling but is exponentially inefficient.

We include these results in Figure~\ref{fig:empirical_results} for convenience and completeness.
The average number of solutions appears to decrease sub-exponentially with the number of samples $n$ (Figure~\ref{fig:empirical_results}, left) and to increase approximately exponentially with the number of mutant types $q$ (Figure~\ref{fig:empirical_results}, center). However, due to the heuristic nature of the analysis and the limited range of parameters explored, the precise rates of decrease and growth remain uncertain from this analysis.

Prior work \cite{myers2019calder} examines the effect of the LC on degeneracy. It employs the same simulator as ~\cite{pradhan2018non} but generates a distinct set of synthetic data.
Specifically, the study begins with $10$ distinct $U^*$ matrices with sizes ranging from $q=4$ to $q=13$.
For each $U^*$, they vary the number of samples $n$ to be generated from $2$ to $9$, and for every fixed combination of $n$, $q$ and $U^*$, the authors generate $10^6$ matrices $M^*$ and $F^*$. They assume that $F^*$ is observed perfectly, without  corruption or masking.
This setup yields $8\times10^7$ random instances. For each of these instances,  the number of solutions, both with and without
the LC, is counted exhaustively using a variation of \cite{gabow1978finding}.
The results show that for all but $423{,}328$ instances, the LC do not reduce degeneracy; that is, for approximately $99.5\%$ of the problems, imposing the LC makes no difference.
This observation is consistent with our Theorem \ref{th:expected_number_of_solutions_without_long_cond_and_with_is_the_same_paper_1}.

The small fraction ($0.5\%$) of instances for which the LC do affect degeneracy does not contradict our theory,
but rather reflects the fact that the problem generator used in \cite{myers2019calder} does not perfectly satisfy Assumption \ref{ass:continuity_of_M_for_calder_conditions_degeneracy_computations_paper_1}.
In particular, the small number of samples ($n$ ranging from $2$ to $9$), and the limited sequencing depth -- with the maximum number of counts appearing to be $200$, based on the publicly available repository cited in \cite{myers2019calder} -- are inconsistent with the rationale behind Assumption \ref{ass:continuity_of_M_for_calder_conditions_degeneracy_computations_paper_1} (Section \ref{sec:discussion_of_assumptions_for_longitudional_conditions_paper_1}),  which invokes dense time sampling and deep sequencing.

With few samples and limited counts, it becomes possible -- even if rare -- to find, among the $10^6$ mutant abundance trajectories generated for each tree, cases where either a parent death coincides exactly with a child birth, or where nonzero mutant-group abundances align exactly.
Finally, note that our ELC additionally require that no mutant remain always dead, a constraint not enforced in  \cite{myers2019calder} and potentially violated in their simulations.

\subsection{New numerical results on the degeneracy of the PPM model with(out) extra constraints}\label{sec:our_new_numerical_results_paper_1}

We conduct an independent brute-force exhaustive enumeration of the number of solutions for different random problems.
This independent computation allows us check the reproducibility of previous work and to compensate for some of its limitations. Namely, the heuristic nature of the counts in \cite{pradhan2018non} and the fact that \cite{myers2019calder} only explored $10$ different $U^*$'s. More importantly however, it allows us to test the effect of the DC on degeneracy beyond the assumptions required for our theory to hold.

We begin by fixing the number of samples to $n=10$ and letting the number of mutant types be $q \in \{4,\dots,8$\}.
We do not use the simulator from \cite{pradhan2018non} but instead
generate $100$ random $M^*$ matrices using a process designed to be consistent with Assumptions \ref{ass:U_star_and_M_star_indep_and_partially_uniform_paper_1} and \ref{ass:brownian_motion_paper_1}, under which our theoretical results hold. The details of this process are provided in Appendix \ref{appendix:details_to_generate_M_paper_1}.

Next, we iterate over all possible $U^*$, i.e., over the set of all directed labeled rooted trees on $q$ nodes.
For each $U^*$ we consider all possible $100$ trajectories, $M^*$.
For every problem instance $(U^*,M^*)$,
we go over all possible $U$ and count how often \eqref{eq:degeneracy_count_with_NO_dynamics_paper_1} holds, as well as how often the different variants of \eqref{eq:degeneracy_count_with_dynamics_last_line_paper_1} are satisfied.
We then compute, for each $U^*$, the average number of solutions found, and report the average of these quantities across all $U^*$ in Figure \ref{fig:relax_constraint_degeneracy_paper_1}.

We perform an exact count up to $q=8$ nodes. This computation required about $42$ hours on a single Intel(R) Core(TM) i9-14900KS processor, which provides $32$ (virtual) cores. The process is not memory-intensive.
For $q=8$, there are $8^6 = 262{,}144$ rooted trees \cite{cayley1878theorem,aigner1999proofs}, where we take node $1$ as the root without loss of generality. The total number of conditions evaluated to compute this average is
$$100\times(8^6)^2 \approx 6.8\times 10^{12}.$$
For  $q=9$, this number would increase to approximately $2.2\times 10^{15}$, and the same computation would require
roughly $566$ days on the same hardware, which is computationally prohibitive.

The variants of \eqref{eq:degeneracy_count_with_dynamics_last_line_paper_1} that we consider are
(a) using either the L1 or L2 norm for $\|\cdot\|$, and
(b) using $\dot{M}$  in \eqref{eq:degeneracy_count_with_dynamics_last_line_paper_1} or replacing it with $M$, which results in the constraint $\|M\| \leq \|M^*\| \land M\geq 0$.

Solution counts without any DC, i.e.,  the degeneracy for the classical PPM model, are labeled ``Unconst.''.
These counts are analogous to  those reported in  \cite{pradhan2018non},
but here are obtained using a different random problem generator and an exact exhaustive count.
We include the results from \cite{pradhan2018non} in Figure \ref{fig:relax_constraint_degeneracy_paper_1} (left plot, green dotted line) to illustrate the difference between an exact count and an approximate count (albeit for different random problems).
\begin{figure*}[h!]
\begin{center}
\includegraphics[trim={1.8cm 0.0cm 1.6cm 0.0cm},clip,width=0.3\linewidth]{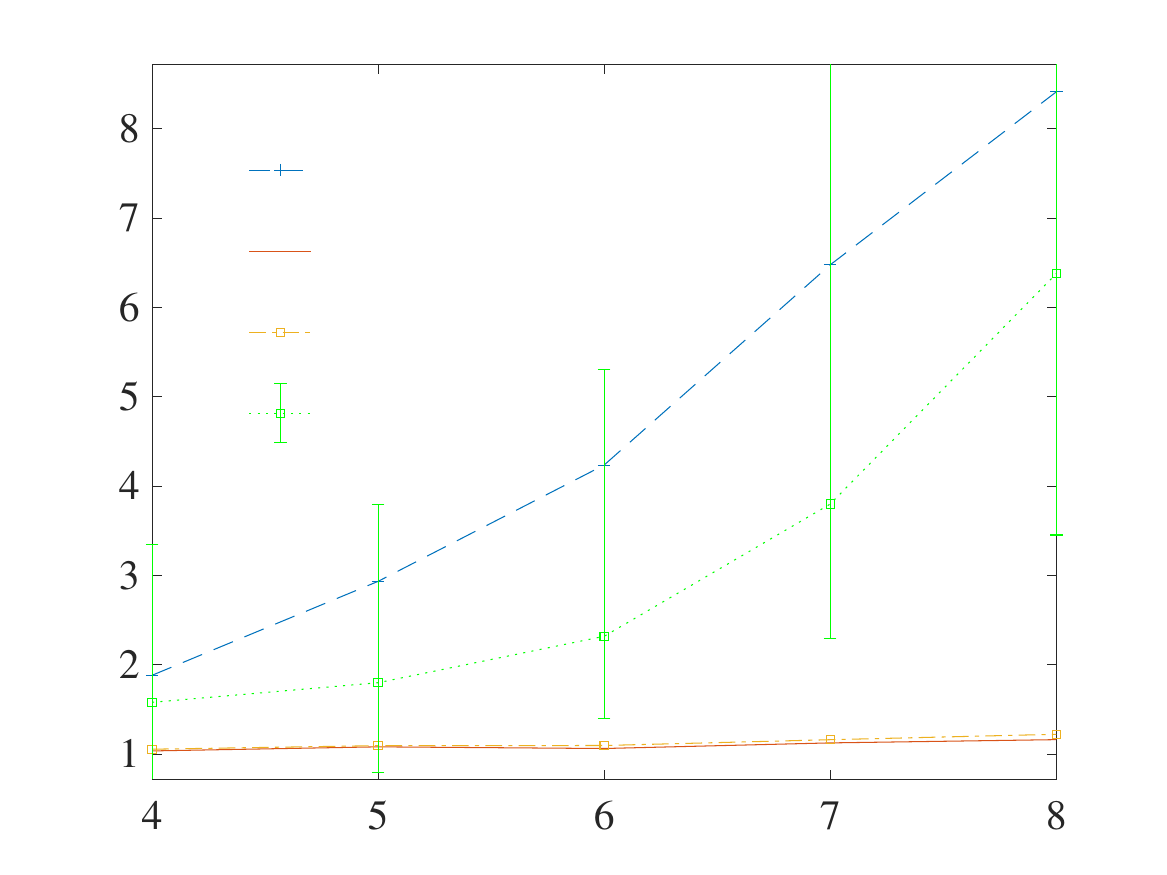}
\put(-165,65){\rotatebox[]{90}{\parbox{5cm}{\centering \fontsize{9}{11}\selectfont {Average \# of solutions per problem}}}}
\put(-118,111){\parbox{3cm}{\fontsize{7}{9}\selectfont Unconst.}}
\put(-118,98){\parbox{3cm}{\fontsize{7}{9}\selectfont L2 \& $\dot{M}$ in \eqref{eq:degeneracy_count_with_dynamics_last_line_paper_1}}}
\put(-118,85){\parbox{3cm}{\fontsize{7}{9}\selectfont L1 \& $\dot{M}$ in \eqref{eq:degeneracy_count_with_dynamics_last_line_paper_1}}}
\put(-118,72){\parbox{3cm}{\fontsize{7}{9}\selectfont \cite{pradhan2018non}}}
\put(-93,-5){\parbox{3cm}{\fontsize{9}{11}\selectfont {$q$ mutations}}}
\quad\quad\;
\includegraphics[trim={1.8cm 0.0cm 1.6cm 0.0cm},clip,width=0.3\linewidth]{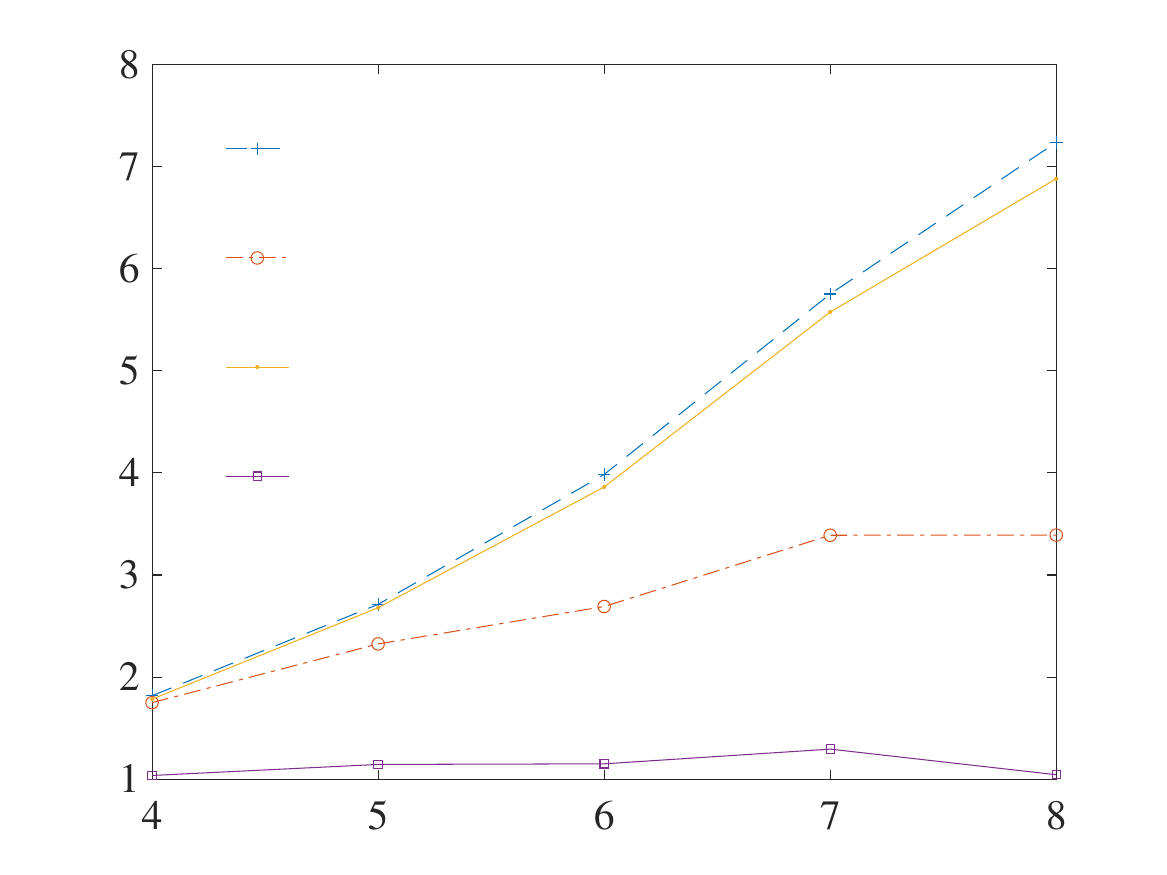}
\put(-173,65){\rotatebox[]{90}{\parbox{3cm}{\centering \fontsize{9}{11}\selectfont {Fold reduction in average \# of solutions}}}}
\put(-120,114){\parbox{3cm}{\fontsize{7}{9}\selectfont L2 \& $\dot{M}$ in \eqref{eq:degeneracy_count_with_dynamics_last_line_paper_1}}}
\put(-120,96){\parbox{3cm}{\fontsize{7}{9}\selectfont L2 \& $M$ in \eqref{eq:degeneracy_count_with_dynamics_last_line_paper_1}}}
\put(-120,79){\parbox{3cm}{\fontsize{7}{9}\selectfont L1 \& $\dot{M}$ in \eqref{eq:degeneracy_count_with_dynamics_last_line_paper_1}}}
\put(-120,62){\parbox{3cm}{\fontsize{7}{9}\selectfont L1 \& $M$ in \eqref{eq:degeneracy_count_with_dynamics_last_line_paper_1}}}
\put(-93,-5){\parbox{3cm}{\fontsize{9}{11}\selectfont {$q$ mutations}}}
\quad\;
\includegraphics[trim={1.8cm 0.0cm 1.6cm 0.0cm},clip,width=0.3\linewidth]{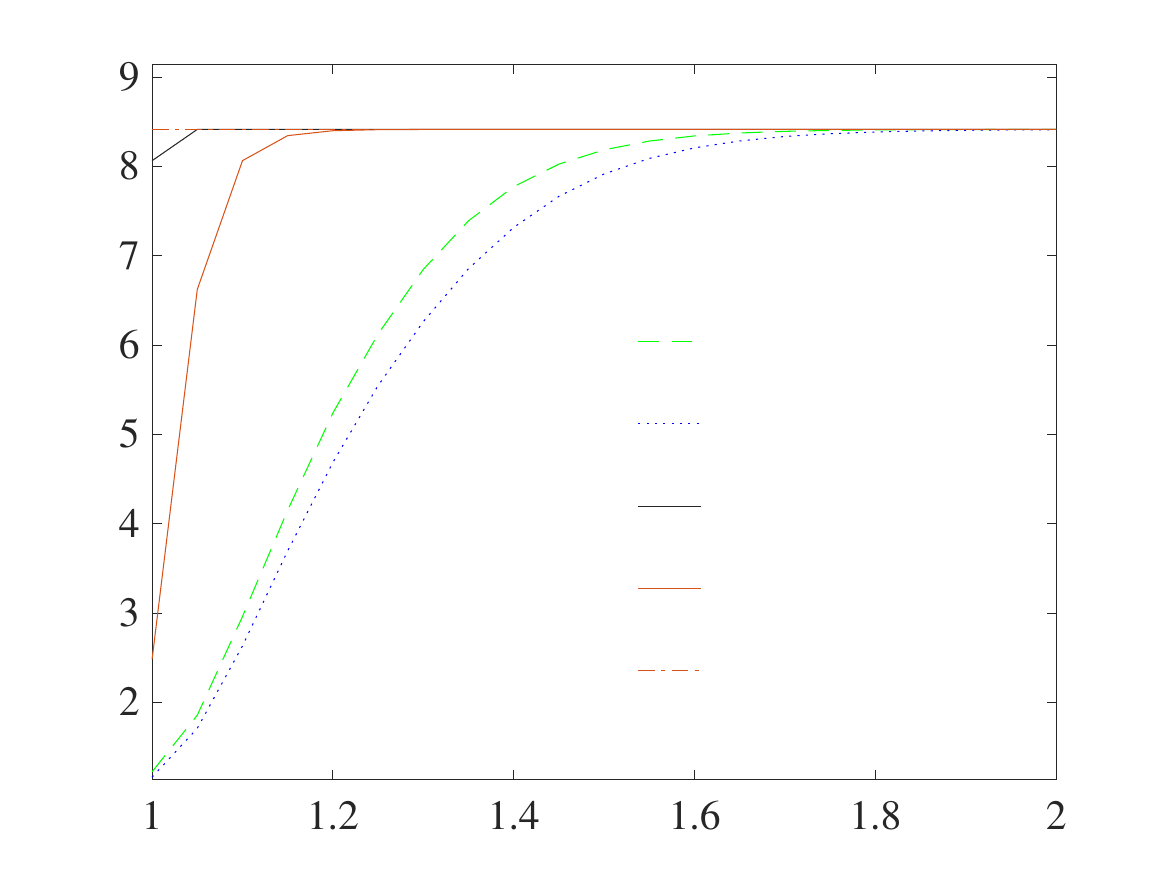}
\put(-165,65){\rotatebox[]{90}{\parbox{5cm}{\centering \fontsize{9}{11}\selectfont {Average \# of solutions per problem}}}}
\put(-110,-5){\parbox{4cm}{\fontsize{9}{11}\selectfont {relaxation parameter $\gamma$}}}
\put(-53,84){\parbox{3cm}{\fontsize{7}{9}\selectfont L1 \& $\dot{M}$ in \eqref{eq:degeneracy_count_with_dynamics_last_line_paper_1}}}
\put(-53,70){\parbox{3cm}{\fontsize{7}{9}\selectfont L2 \& $\dot{M}$ in \eqref{eq:degeneracy_count_with_dynamics_last_line_paper_1}}}
\put(-53,56){\parbox{3cm}{\fontsize{7}{9}\selectfont L1 \& $M$ in \eqref{eq:degeneracy_count_with_dynamics_last_line_paper_1}}}
\put(-53,43){\parbox{3cm}{\fontsize{7}{9}\selectfont L2 \& $M$ in \eqref{eq:degeneracy_count_with_dynamics_last_line_paper_1}}}
\put(-53,31){\parbox{3cm}{\fontsize{7}{9}\selectfont Unconst.}}
\end{center}
\caption{
Average number of solutions versus the number of mutations $q$. The average is taken over all possible true $U^*$'s, over $100$ different random ground-truth $M^*$'s, and over all possible alternative explanations $U$ (including $U=U^*$).
(Left plot) Comparison of the average number of solutions without extra constraints (``Unconst.'') with different variants of the DC \eqref{eq:degeneracy_count_with_dynamics_last_line_paper_1}.
(Center plot) Ratio of the average number of solutions with and without DC the different variants of the DC \eqref{eq:degeneracy_count_with_dynamics_last_line_paper_1}.
(Right plot) Effect of relaxing the DC in \eqref{eq:degeneracy_count_with_dynamics_last_line_paper_1} on the average number of solutions.
Specifically, the condition $r(M) \leq r(M^*)$ is replaced by $r(M) \leq C = \gamma r(M^*)$ for  different values of $\gamma$.
}
\label{fig:relax_constraint_degeneracy_paper_1}
\end{figure*}

Let us focus on Figure \ref{fig:relax_constraint_degeneracy_paper_1} (left).
Comparing the blue dashed line (``Unconst.'') and the green dotted line from \cite{pradhan2018non}, we see that heuristic counts either underestimate values or can give ambiguous results when confidence intervals are taken into account.
Comparing the blue dashed line (``Unconst.'') with either the solid or dashed-dotted orange line (L1 or L2 with $\dot{M}$), we see that our DC almost eliminate any ambiguity (the average number of solutions is $\approx 1$).
This observation is consistent with Remark \ref{rmk:dynamic_constraints_remove_degeneracy_completely_paper_1}.
The specific norm used is more or less irrelevant.
The average number of solutions without DC is $8.41$,  and it drops to $1.16$ with an L1 constraint on $\dot{M}$, corresponding to a  $7.24\times$ reduction.
For an L2 constraint on $\dot{M}$ the reduction is $6.88 \times$.

Now we focus on Figure \ref{fig:relax_constraint_degeneracy_paper_1} (center). We observe that DC that penalize the rate of change reduce degeneracy more than those that penalize the abundances themselves, i.e.
constraining $\|\dot{M}\|$ works better than constraining $\|{M}\|$.

Finally, we focus on Figure \ref{fig:relax_constraint_degeneracy_paper_1} (right).
Here we modify the DC from $r(M) \leq r(M^*)$ to $r(M) \leq C = \gamma r(M^*)$ and see the effect that choosing $\gamma \geq 1$ has on degeneracy. We see that for a DC that penalizes the rate of change, the reduction in degeneracy is less sensitive to the exact choice of upper bound $C$ than for a DC that penalizes the abundances themselves.

In Figure \ref{fig:degeneracy_with_and_without_time_constraints_paper_1} we show the distribution of the number of solutions across all problems for the largest setting we tested, i.e., $100\times8^6$ problems for $q=8$, both with and without DC.
DC reduce degeneracy not only on average, as already observed in Figure \ref{fig:relax_constraint_degeneracy_paper_1}, but also substantially shorten the tail of the distribution of the number of solutions compared to the unconstrained case.
\begin{figure}[h!]
\begin{center}
\includegraphics[trim={0.17cm 0.7cm 1.5cm 0.8cm},clip,width=0.95\linewidth]{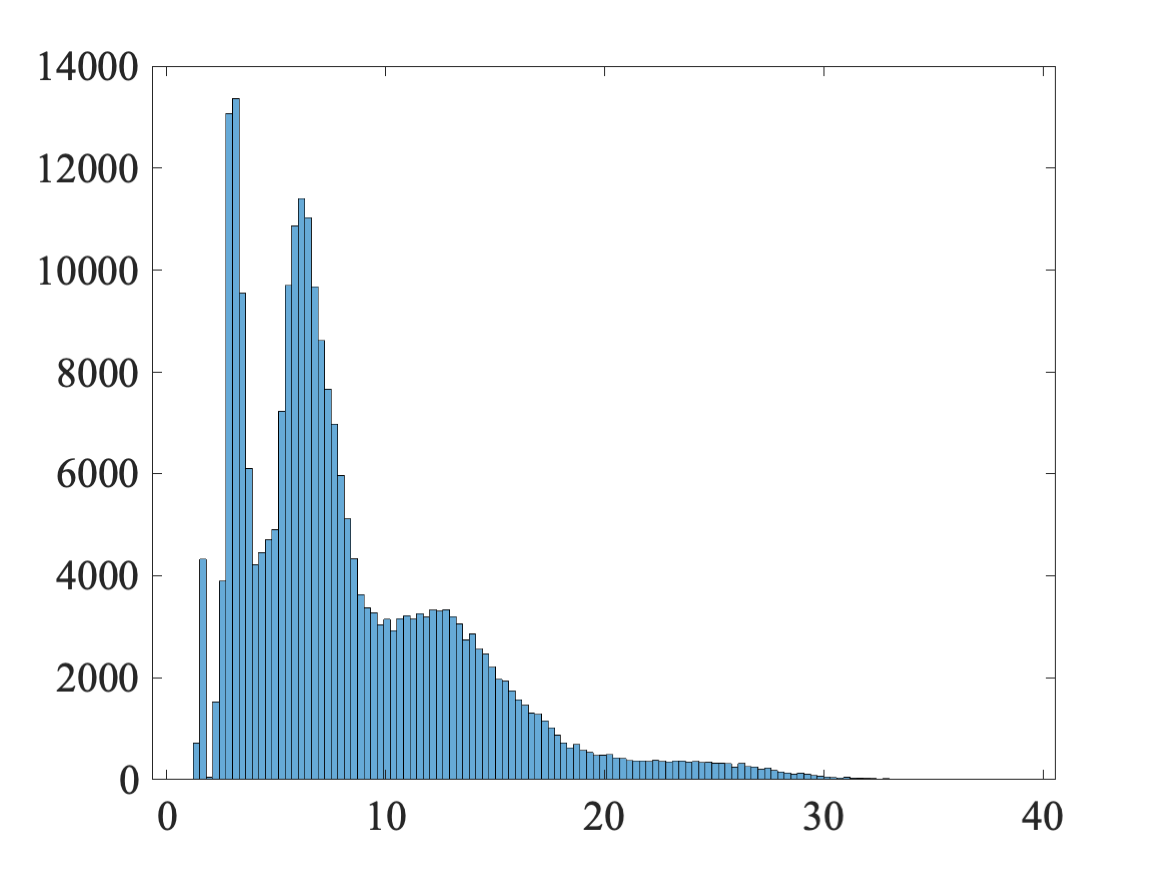}
\put(-245,90){\rotatebox[]{90}{\parbox{3.5cm}{\centering \fontsize{9}{11}\selectfont {\# of problems}}}}
\put(-155,125){\parbox{4.8cm}{\fontsize{9}{11}\selectfont {\centering Distribution of number of solutions without DC}}}
\put(-132,-8){\parbox{3cm}{\fontsize{9}{11}\selectfont {\# of solutions}}}
\vspace{0.5cm}
\includegraphics[trim={0.17cm 0.7cm 1.5cm 0.8cm},clip,width=0.95\linewidth]{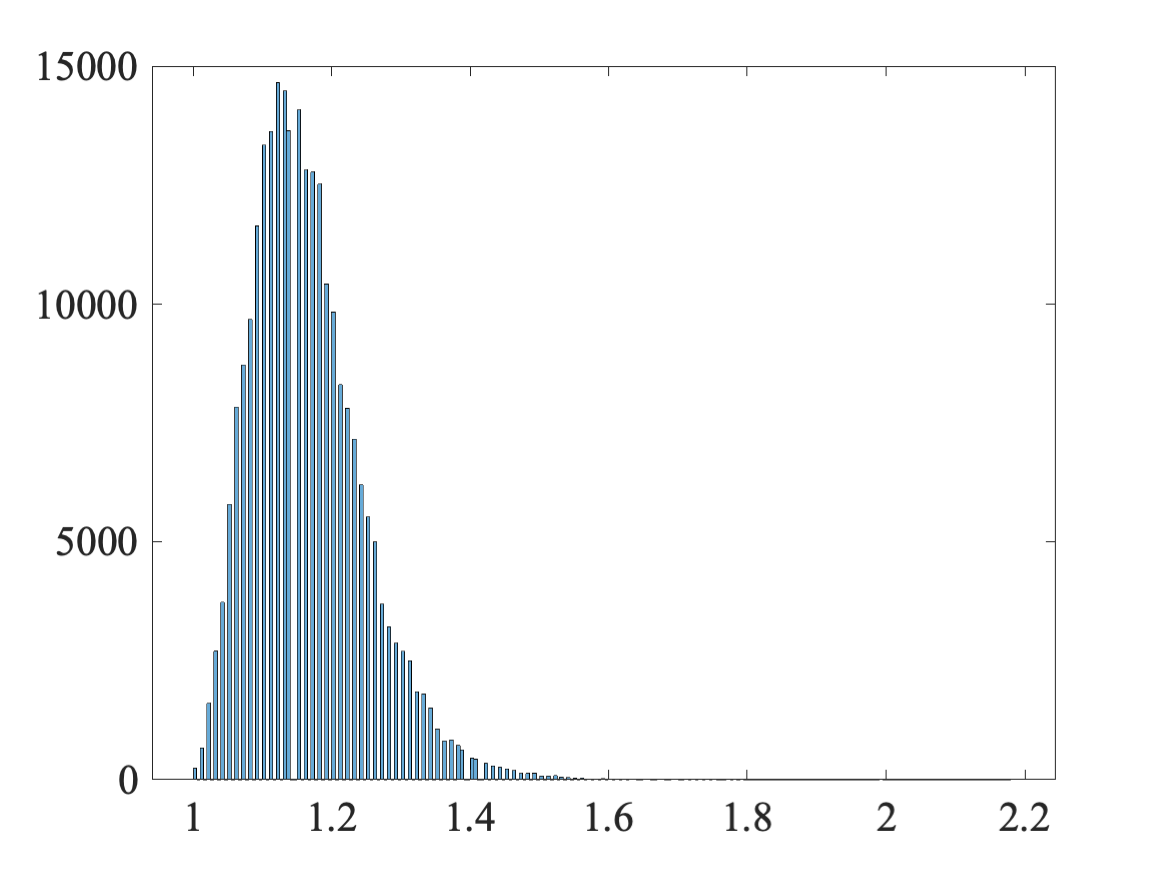}
\put(-245,90){\rotatebox[]{90}{\parbox{3.5cm}{\centering \fontsize{9}{11}\selectfont {\# of problems}}}}
\put(-155,125){\parbox{4.8cm}{\fontsize{9}{11}\selectfont {\centering Distribution of number of solutions with DC}}}
\put(-132,-8){\parbox{3cm}{\fontsize{9}{11}\selectfont {\# of solutions}}}
\end{center}
\caption{Distribution of the number of solutions per problem with and without a DC on the L2 norm of $\dot{M}$. The total number of problems is $100$ different $M^*$ times all possible $8^6$ different ancestries $U^*$ for $q=8$ and $n=10$.}
\label{fig:degeneracy_with_and_without_time_constraints_paper_1}
\end{figure}

\subsection{Details on generating $M^*$}
\label{appendix:details_to_generate_M_paper_1}

We initialize the trajectories by setting $\tilde{{\bf M}}_{:,1} =  {\bf 1}/n$ and recursively compute
$$\tilde{{\bf M}}_{:,t} = \tilde{{\bf M}}_{:,t-1}  + \alpha {\bf Z}_{:,t},$$
where $\{Z_{i,t}\}$ are i.i.d. random variables  following a normal distribution $\mathcal{N}(0,1)$.

For Figure \ref{fig:degeneracy_with_and_without_time_constraints_paper_1} we set $\alpha = 0.05$.
To obtain $M^*$ we apply the matrix $P = I - {\bf 1}{\bf 1}^{\top}/n$ to each
$\tilde{{\bf M}}_{:,t}$ so that ${\bf 1}^{\top} M^*_{:,t} = {\bf 1}^{\top}$, and discard any trajectory that does not satisfy the non-negativity constraint $M^* \geq 0$.
A typical example is shown in Figure \ref{fig:trajectory_paper_1}.
\begin{figure}[h!]
\hfill
\includegraphics[trim={1.1cm 0.7cm 1.5cm 0.8cm},
clip,
width=0.9\columnwidth]{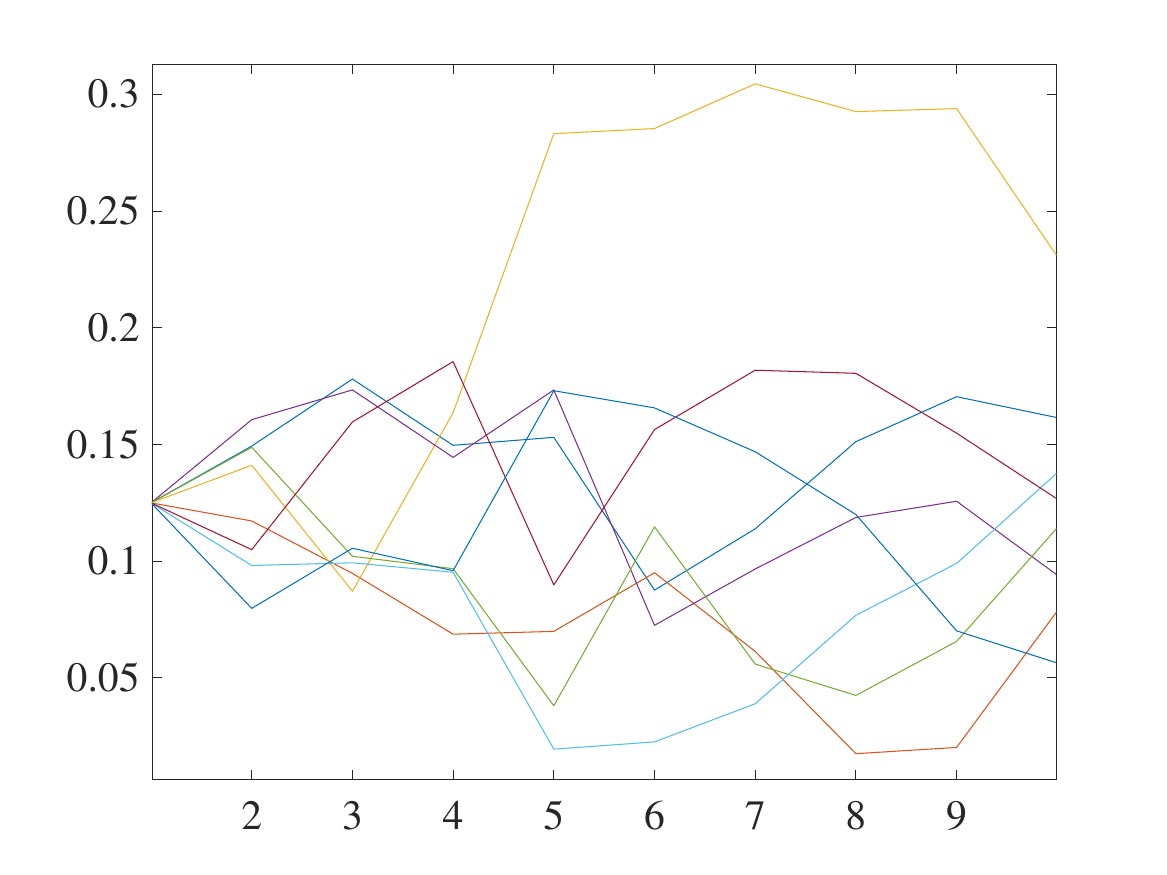}
\put(-245,20){\rotatebox{90}{\parbox{5cm}{\centering \fontsize{9}{11}\selectfont {Mutant abundance frequency, ${M}^*$}}}}
\put(-125,-10){\parbox{3cm}{\fontsize{9}{11}\selectfont {Sample $t$ }}}
\vspace{0cm}
\caption{Typical generated set of trajectories $M^*$.}
\label{fig:trajectory_paper_1}
\end{figure}
The procedure used to generate $M^*$ nearly satisfies Assumption \ref{ass:U_star_and_M_star_indep_and_partially_uniform_paper_1} and Assumption \ref{ass:brownian_motion_paper_1}.
Specifically, $U^*$ and $M^*$ are generated independently,
and it is  impossible for row ${\bf M}^*_{i,:}$ to become identically zero.
The evolution of each ${\bf M}^*_{i,:}$ follows a Brownian motion that is accepted if it remains non-negative and rejected otherwise.
However, we subsequently adjust the trajectories to enforce the constraint ${\bf 1}^{\top} M^* = {\bf 1}^{\top}$, and the initial condition is not drawn uniformly at random but instead fixed as ${\bf M}^*_{:,1} = {\bf 1}/n$.
These modifications break both the independence condition in Assumption \ref{ass:U_star_and_M_star_indep_and_partially_uniform_paper_1} and the distributional assumptions on $\tilde{M}$ in Assumption \ref{ass:brownian_motion_paper_1}.
These differences are small and are unlikely to affect the empirical results we present.

Because of the way we generate $M^*$, Assumption \ref{ass:assumption_no_root_node_is_zero_paper_1} is always satisfied; that is, there exists at least one time index $t$ such that $M^*_{r^*,t} > 0$.
By Lemma \ref{th:assumption_2_lemma_paper_1}, this implies that, during inference, we do not need to explicitly enforce the constraint ${\bf 1}^{\top} M = {\bf 1}^{\top}$ and instead only need to verify whether $U^{-1} U^* M^* \geq 0$ and whether $\|MD\| \leq \|M^*D\|$.

\newpage
\section{Background on Related Work}
\label{app:details_on_related_work}

\subsection{Degeneracy related concepts}

Several terms related to degeneracy are discussed in the literature in  this context.

\begin{itemize}[leftmargin=*]
\item  \textit{Identifiable:}
An evolutionary model whose stochastic process defines a injective map from the parameters to the probability distribution of the observed samples.
Although identifiability implies the existence of a set-theoretic inverse from distributions to parameters, and hence there is no degeneracy in this set-theoretic sense, such an inverse need not be measurable or realizable as the limit of any estimator based on finite, or even an infinite sequence of samples.
Hence, identifiability does not imply that there is a \emph{consistent} estimator (see 3rd bullet), and there might be degeneracy in this sense. These cases however are pathological, which is not the case for the evolutionary models in the literature.
\item \textit{Generically identifiable}: Identifiable except on a set of parameters of measure zero.
\item \textit{Consistent:}
    A estimator for which the output parameters, including the estimated  tree, converge (in probability) to the true parameters, including the true tree, as the number of samples goes to infinity. For a consistent estimator, degeneracy disappears as we get more observed data.
    Non indentifiability implies non consistency of all estimators.
\item  \textit{Compatible trees:} When there is degeneracy, not all trees are equally bad as alternative explanations of the truth. A tree $T'$ is \emph{compatible} with tree $T$ if $T'$ can be reduced to $T$ by merging groups of internal nodes.

\item \textit{Non-identifiable mixture.}
A mixture model, such as the PPM model, where multiple choices of mixture weights and component parameters  produce the same overall distribution, for the same fixed number of mixtures.
This can occur even when each component model is individually identifiable.

\item \textit{Non-identifiable mixture distribution.}
Stronger concept than non-identifiable mixtures. A mixture models allows for fundamentally different mixture representations, possibly with different numbers of components or different structures.

\end{itemize}

Phylogenetic networks generalize trees by allowing reticulation events such as recombination or hybridization.
They are often introduced as alternative generative models that remain identifiable in settings where tree-based models are not.
Networks provide a way to resolve or avoid tree-level degeneracy by expanding the model class rather than attempting to select among incompatible trees.
\vspace{-.4cm}
\subsection{Other models}
We now discuss some well known models.
\paragraph{Substitution models} \label{sec:review_of_some_models_paper_1}
The Cavender-Farris-Neyman (CFN) model~\cite{neyman1971molecular, farris1973probability, cavender1978taxonomy} describes the evolution of a fixed-length list of binary characters via a substitution model defined by a two-state, time-reversible Markov process.
In the CFN,  each taxon (typically a species) is represented by a sequence of observed binary characters, which model DNA bases. In the CFN model, each character is either $0$ or $1$.

Sequences are first aligned using a multiple sequence alignment method (MSA) so that for every position, $i$, the $i$-th character across all taxa is assumed to have descended from the same character in the shared ancestor.
These aligned positions are commonly referred to as "sites."
Afterwards, and in the case of the CFN, the bases are binarized.

In the CFN, all sites evolve independently of each other and in the same way, and the evolution, ie., substitutions, happens along the edges of the tree. Unlike in the PPM model, the same position can mutate multiple times.
The phylogenetic tree is parameterized by a tree topology and edge parameters, which are $2$ by $2$ transition matrices that control the expected number of substitutions per site along an edge in the tree.

During inference, the interior nodes of the tree correspond to latent (unobserved) states, or distributions over states, while the leaf nodes correspond to observations (taxa with observed character states, eg., the actual data). We encourage the reader to see \cite{semple2003phylogenetics} for a more comprehensive reference.

The Jukes--Cantor model (JC) \cite{jukes1969evolution}, uses  a simple time-reversible, continuous-time Markov chain with four states (C, T, A, G), all of which are handled symmetrically, analogous to CFN.
The Kimura two-parameter model (K2) \cite{kimura1980simple} is a more accurate model that accounts for two different types of substitutions:  ``transitions'' and ``transversions''.

Transitions are substitutions occurring within either \rev{purines} or   \rev{pyrimidines}, whereas transversions are mutations occurring between a \rev{purine} and a \rev{pyrimidine}. In K2, transitions occur with a higher frequency than transversions. The Kimura three--parameter model (K3) \cite{kimura1981estimation} is a model that takes into account how many hydrogen bonds are changed by certain mutations.

\paragraph{Coalescent models}
Coalescent models~\cite{marjoram2006modern} are some of the most popular population genetics models. Rather than stochastically modeling substitutions forward in time, such as group-based substitution models, these models describe the evolutionary process by going backwards in time and finding points at which certain pairs of fragments merge together.

\section{Inference under the PPM model when not observing all the mutants. Discussed in Section \ref{sec:background_ppm_paper_1}.} \label{appendix:PPM_partially_observed_paper_1}

We start with the PPM model where all the mutants are observed and assume, without loss of generality, that there is only one sample, i.e., $n =1$.
We have $F^* = U^* M^*$, ${\bf 1}^{\top} M^* = {\bf 1}^{\top}$, and $M^* \geq {\bf 0}$, where $U^*$ encodes ancestral relationships in the rooted directed tree ${T^*}$.

Let $\mathcal{O}$ be the subset of observed mutants,
indexed in an arbitrary but fixed order, and $|\mathcal{O}| = q' < q$.
We observe $F = F^*_{\mathcal{O}}$ and want to infer $U^*_{ \mathcal{O} , \mathcal{O}}$.
Note that the submatrix $U^*_{ \mathcal{O}, \mathcal{O}}$ represents ancestral relationships that define a forest, i.e., a collection of directed rooted trees, where the roots of its trees correspond to distinct nodes in ${T^*}$.

Let $\mathcal{N} = [q] - \mathcal{O}$ be the set of non-observed mutants, again indexed in an arbitrary but fixed order. We can write
$$F^*_\mathcal{O} = U^*_{ \mathcal{O} , \mathcal{O}} M^*_{\mathcal{O}} + U^*_{ \mathcal{O} , \mathcal{N}} M^*_{\mathcal{N}} = U^*_{ \mathcal{O} , \mathcal{O}} (M^*_{\mathcal{O}} + (U^*_{ \mathcal{O} , \mathcal{O}})^{-1} U^*_{ \mathcal{O} , \mathcal{N}} M^*_{\mathcal{N}}),$$
where, since $n=1$, quantities like $M^*_{\mathcal{O}}$ actually denote $M^*_{\mathcal{O},t=1}$.
These quantities may therefore be represented either using matrix notation (for a single-column matrix, in non-bold font) or using vector notation (in bold font).

Define $\tilde{M}^* \equiv M^*_{\mathcal{O}} + (U^*_{ \mathcal{O} , \mathcal{O}})^{-1}U^*_{ \mathcal{O} , \mathcal{N}} M^*_{\mathcal{N}}$.

We will show that:
\begin{enumerate}
    \item For any $U^*$, we have $\tilde{M}^* \geq { 0}$ and ${\bf 1}^{\top} \tilde{M}^* \leq {\bf 1}^{\top}$;
    \item If $\mathcal{O}$ contains the root of ${T^*}$, then
${\bf 1}^{\top} \tilde{M}^* = {\bf 1}^{\top}$;
\item  When $\mathcal{O}$ is not empty and
does not contain the root of ${T}^*$, and if we view $\tilde{M}^*$ as a function of $M^*$ (for any fixed $U^*$), then the image of
$$\{M \in \mathbb{R}^{q\times n}: {\bf 1}^{\top} M = {\bf 1}^{\top} \land M \geq { 0} \}$$
under the mapping $\tilde{M}^*(\cdot)$ contains the set
$$\{ \tilde{M} \in \mathbb{R}^{q'\times n}: \tilde{M} \geq { 0} \land {\bf 1}^{\top} \tilde{M}\leq {\bf 1}^{\top} \}.$$
\end{enumerate}
Because of these facts, given $F\in \mathbb{R}^{q'\times n}$,  we should find
$\tilde{U}\in \mathbb{R}^{q'\times q'}$ and $\tilde{M}\in \mathbb{R}^{q'\times n}$ such that
$F= \tilde{U} \tilde{M}$, $\tilde{M} \geq { 0}$,  ${\bf 1}^{\top} \tilde{M}\leq {\bf 1}^{\top}$, and $\tilde{U}$ is a set of rooted directed trees.

Indeed, by the first point we should not enforce anything less restrictive than what $\tilde{M}^*$ itself satisfies, which implies
$\tilde{M} \geq { 0}$ and ${\bf 1}^{\top} \tilde{M}\leq {\bf 1}^{\top}$.
By the second point, we should not enforce ${\bf 1}^{\top} \tilde{M} =  {\bf 1}^{\top}$, since \emph{a priori} we do not know whether $\mathcal{O}$ contains the root of ${T^*}$.
By the third point, we should not enforce any constraint more restrictive than $\tilde{M} \geq {\bf 0} \land {\bf 1}^{\top} \tilde{M}\leq {\bf 1}^{\top}$, because any such $\tilde{M}$
can be realized by a valid pair $(U^*,M^*)$ under the PPM model with partially observed mutants.

\subsection{Proofs}

With a proper re-indexing $U^*_{ \mathcal{O}, \mathcal{O}}$ has block diagonal form, and each block represents the ancestral relationships of a tree in this forest \footnote{This block diagonal form justifies the existence of $(U^*_{ \mathcal{O} , \mathcal{O}})^{-1}$.}.
Let us call the $i$-th tree in the forest ${T^*}^{(i)}$, and denote its set of nodes by $\mathcal{O}^{(i)}$.
With a slight abuse of notation, let ${T^*}^{(i)}$ also denote the matrix representation of the operator that maps each node in the tree ${T^*}^{(i)}$ to its parent\footnote{Note that ${T^*}^{(i)} \neq {T^*}_{\mathcal{O}^{(i)},\mathcal{O}^{(i)}}$. The former matrix always represents a directed rooted tree, while the latter matrix could be, e.g., an all zero matrix.}.
We define ${U^*}^{(i)} \equiv U^*_{ \mathcal{O}^{(i)} , \mathcal{O}^{(i)}}$ and, by point~3 in Section \ref{sec:background_ppm_paper_1}, we have
$$({U^*}^{(i)})^{-1} = I - {T^*}^{(i)}.$$
Note that $(U^*_{ \mathcal{O}, \mathcal{O}})^{-1}$  has block diagonal form, with its $i$-th block equal to $({U^*}^{(i)})^{-1}$.

\indent \textbf{1)}
First, we show that ${\bf 1}^{\top} \tilde{M}^*\leq {\bf 1}^{\top}$.
Let $\tilde{U}^* \equiv U^*_{ \mathcal{O}, \mathcal{O}}$.
Working with one block at a time, we see that the vector ${\bf v}^{\top} \equiv {\bf 1}^{\top} (\tilde{U}^*)^{-1}$ is a binary vector with $1$'s at the indices corresponding to the roots of each tree ${T^*}^{(i)}$, and zeros everywhere else.
By the definition of ${T^*}^{(i)}$ as the diagonal blocks of $U^*_{\mathcal{O},\mathcal{O}}$, the root of each tree lies on a distinct lineage of ${T^*}$.
Hence, each node in $\mathcal{N}$ has at most one of these roots as an ancestor in $T^*$.
It follows that the vector
${\bf w}^{\top} \equiv {\bf v}^{\top} U^*_{ \mathcal{O} , \mathcal{N}}$
cannot have entries larger than $1$, and therefore is also a binary vector.

If the root of ${T^*}$ belongs to $\mathcal{N}$, then the vector ${\bf w}^{\top}$
must contain at least one zero entry.
In this case,
\begin{align*}
{\bf 1}^{\top}\tilde{M}^* &= {\bf 1}^{\top} M^*_{\mathcal{O}} + {\bf 1}^{\top}(\tilde{U}^*)^{-1}U^*_{ \mathcal{O} , \mathcal{N}} M^*_{\mathcal{N}} \\
&= {\bf 1}^{\top} M^*_{\mathcal{O}} + {\bf w}^{\top} M^*_{\mathcal{N}}\\
&
\leq {\bf 1}^{\top} M^*_{\mathcal{O}} + {\bf 1}^{\top} M^*_{\mathcal{N}} \ =\ {\bf 1}^{\top}.
\end{align*}
If the root of ${T^*}$ belongs to $\mathcal{O}$, then the forest consists of a single tree whose root coincides with the root of ${T^*}$;
denote this root by $r^*$.
In this case, ${\bf v}^{\top}$ has exactly one nonzero entry, equal to $1$, at the index corresponding to $r^*$.
Since every mutant in $\mathcal{N}$ is a descendant of $r^*$, we have
${\bf w} = {\bf 1}$
and therefore
\begin{align*}
{\bf 1}^{\top}\tilde{M}^* &= {\bf 1}^{\top} M^*_{\mathcal{O}} + {\bf 1}^{\top}(\tilde{U}^*)^{-1} U^*_{ \mathcal{O} , \mathcal{N}} M^*_{\mathcal{N}}\\
&= {\bf 1}^{\top} M^*_{\mathcal{O}} + {\bf w}^{\top} M^*_{\mathcal{N}}\\
& =
{\bf 1}^{\top} M^*_{\mathcal{O}} + {\bf 1}^{\top} M^*_{\mathcal{N}} \ =\ {\bf 1}^{\top}.
\end{align*}

\indent \textbf{2)}
Next, we show that $\tilde{M}^* \geq 0$. Since $M^* \geq 0$, it suffices to show that $(\tilde{U}^*)^{-1}U^*_{ \mathcal{O} , \mathcal{N}} $ has no negative entries. In fact, it suffices to  show that non-negativity holds for each subset of nodes $\mathcal{O}^{(i)}$. Mathematically, this means showing that
$$({U^*}^{(i)})^{-1} U^*_{ \mathcal{O}^{(i)} , v}  \geq  {\bf 0}$$
 for any $v \in \mathcal{N}$.

 To prove this, first observe that
$$({{U}^*}^{(i)})^{-1} U^*_{ \mathcal{O}^{(i)} ,v}  = U^*_{ \mathcal{O}^{(i)} ,v}  - {{T^*}^{(i)}} U^*_{ \mathcal{O}^{(i)} , v}.$$
Second, recall the overloaded meaning of ${{T^*}^{(i)}}$: it denotes both a tree on the nodes $\mathcal{O}^{(i)}$, namely, the $i$-th tree in the forest, and the corresponding matrix operator mapping nodes to their parents.
These two representations uniquely determine each other.
Third, recall the meaning of $U^*_{ \mathcal{O}^{(i)} , v}$: this column contains a $1$ in the rows associated with the set of all the ancestors of $v$ (ancestors in the sense of the hierarchy defined by ${T^*}$)  that are nodes in the tree ${{T^*}^{(i)}}$. Overloading notation, let us identify this set with $U^*_{ \mathcal{O}^{(i)} , v}$.
Among the ancestors in
$U^*_{ \mathcal{O}^{(i)} , v}$, all but one of them, namely the root of the tree ${{T^*}^{(i)}}$, have their parent also in $U^*_{ \mathcal{O}^{(i)} , v}$.
Therefore, ${{T^*}^{(i)}} \;U^*_{ \mathcal{O}^{(i)} , v}$ equals $U^*_{ \mathcal{O}^{(i)} , v}$ except in the entry corresponding to the root of ${{T^*}^{(i)}}$,
which is $0$ in the former but $1$ in the latter.
Hence, $U^*_{ \mathcal{O}^{(i)} ,v}  - {{T^*}^{(i)}} U^*_{ \mathcal{O}^{(i)},v}$ is a non-negative column.

\indent \textbf{3)} Finally, we prove the third statement on our list.
First, notice that if $\mathcal{O}$ is not empty, then we can set $M^*_\mathcal{N} = {\bf 0}$, and simply by varying $M^*_\mathcal{O}$ we can cover
$$\{ \tilde{M}\in\mathbb{R}^{q'\times n}: \tilde{M} \geq { 0} \land {\bf 1}^{\top} \tilde{M} =  {\bf 1}^{\top} \}.$$
Second, notice that if the root $r^*$ of ${T^*}$ is in $\mathcal{N}$ and we choose $M^*$ such that $M^*_{r^*,:} = {\bf 1}$ and all the other components are zero, then
$U^*_{\mathcal{O},\mathcal{N}} M^*_\mathcal{N} = {\bf 0}$, since no node in $\mathcal{O}$ is an ancestor of the root $r^* \in  \mathcal{N}$.
Hence $\tilde{M}^* = {\bf 0}$. Since $\tilde{M}^*(\cdot)$ is a linear function, the image of the convex set $\{ {M}\in\mathbb{R}^{q' \times n}: {M} \geq {\bf 0} \land {\bf 1}^{\top} {M} =  {\bf 1}^{\top} \}$ must also be convex.
Since this image covers $\{ \tilde{M}\in\mathbb{R}^{q' \times n}: \tilde{M} \geq {\bf 0} \land {\bf 1}^{\top} \tilde{M} =  {\bf 1}^{\top} \}$ and also covers the point ${\bf 0}$, it must cover the set $\{ \tilde{M}\in\mathbb{R}^{q \times n}: \tilde{M} \geq {\bf 0} \land {\bf 1}^{\top} \tilde{M} \leq  {\bf 1}^{\top} \}$.

\section{Alternative to the redefinition of birth and death times in Section \ref{sec:redefine_birth_death_time_original_paper_1}} \label{appendix:alternative_definition_paper_1}

In Section \ref{sec:redefine_birth_death_time_original_paper_1}, we extend the definition of birth and death time from \cite{myers2019calder} in order to cover edge cases not addressed in their original definition.
We also add an extra condition to their longitudinal conditions (LC), requiring that the birth time be strictly smaller than the death time. These edge cases are discussed in Section \ref{sec:discussion_of_definitions_and_their_problems_paper_1}.

Furthermore, in Section \ref{sec:redefinition_of_longitudional_conditions_paper_1} we define new birth and death times, as well as new LC, and we prove in Lemma \ref{th:equivalent_CALDER_conditions_paper_1} that both our definitions and those of \cite{myers2019calder} are equivalent.

In this section, we present another possible redefinition for birth and death times.
In particular, we can define $t^{\min}_v$ and  $t^{\max}_v$ as in \eqref{eq:CALDER_t_min_def_paper_1} and \eqref{eq:CALDER_t_max_def_paper_1}, and add that
\begin{itemize}
\item if $t^{\min}_v$ from \eqref{eq:CALDER_t_min_def_paper_1} is undefined, then set $t^{\min}_v = M_1$,
\item once $t^{\min}_v$ is defined, if $t^{\max}_v$ as in \eqref{eq:CALDER_t_max_def_paper_1} (but using this newly defined $t^{\min}_v$) is undefined, then set $t^{\max}_v = M_2$,
\end{itemize}
where $M_1$ and $M_2$ are constants satisfying $M_2 > M_1 > n$.
This is different from Definition \ref{def:extended_birth_and_death_time_calder_def_paper_1} in the main paper, where we set $M_1 = 1$ and $M_2 > n$.

We define $t'^{\min}_v$ and  $t'^{\max}_v$ as in \eqref{eq:our_t_min_def_paper_1} and \eqref{eq:our_t_max_def_paper_1} and add that
\begin{itemize}
\item if $t'^{\min}_v$ as in \eqref{eq:our_t_min_def_paper_1} is not defined, then: if $M_{v,1} >0$ then $t'^{\min}_v =1$, and if $M_{v,1} =0$ then $t'^{\min}_v =M_1$;
\item if $t'^{\max}_v$ as in \eqref{eq:our_t_max_def_paper_1} is not defined, then $t'^{\max}_v = M_2$;
\end{itemize}
where $M_2 > M_1 > n$ are the same constants as in the above paragraph.

Lemma \ref{th:equivalent_CALDER_conditions_paper_1} also holds if it uses these new definitions, which we restate here as Lemma \ref{th:equivalent_CALDER_conditions_for_alternative_definitions_paper_1}.
\begin{lem} \label{th:equivalent_CALDER_conditions_for_alternative_definitions_paper_1}
    Define $t'^{\min}_v$ and $t'^{\max}_v$ as in \eqref{eq:our_t_min_def_paper_1} and \eqref{eq:our_t_max_def_paper_1} but extended according to Appendix \ref{appendix:alternative_definition_paper_1}.
    Define $t^{\min}_v$ and $t^{\max}_v$ as in \eqref{eq:CALDER_t_min_def_paper_1} and \eqref{eq:CALDER_t_max_def_paper_1} but extended according to Appendix \ref{appendix:alternative_definition_paper_1}.
The ELC in Definition \ref{defi:extended_longitudional_conditions_paper_1} are equivalent to the ELC in Definition \ref{defi:new_definition_of_LC_paper_1}
and if they hold, then $t'^{\min}_v=t^{\min}_v$ and $t'^{\max}_v=t^{\max}_v$.
\end{lem}

We omit the proof of Lemma \ref{th:equivalent_CALDER_conditions_for_alternative_definitions_paper_1}, as it is very similar to the proof of Lemma \ref{th:equivalent_CALDER_conditions_paper_1}.

\begin{rmk}
The equivalence between $t^{\min}_v < t^{\max}_v$ in  \eqref{eq:extra_condition_to_add_to_longitudional_conditions_paper_1}
and ${\bf M}_v \neq {\bf 0}$ in   \eqref{eq:extra_long_constraint_no_zero_M_paper_1}, which is proved in Lemma \ref{th:equivalent_of_extending_calder_long_conditions_in_one_or_in_another_way_paper_1} in Appendix \ref{app:equivalent_of_extending_calder_long_conditions_in_one_or_in_another_way_paper_1}, does not hold if we use the definitions in Appendix \ref{appendix:alternative_definition_paper_1}, and hence the ELC no longer exclude scenarios where a mutant is observed dead throughout the observation window, i.e. ${\bf M}_v = {\bf 0}$.
For example, if we have two mutants $v$ and $w$, and $w$ is a child of $v$, then the ELC in Definition \ref{defi:extended_longitudional_conditions_paper_1} using the redefined $t^{\min}_v$ and $t^{\max}_v$ in Appendix \ref{appendix:alternative_definition_paper_1} exclude the situation where ${\bf M}_{v,:}= {\bf 0}$ and ${\bf M}_{w,:}= {\bf 1}$, but they do not exclude the situation where ${\bf M}_{v,:}= {\bf 1}$ and ${\bf M}_{w,:}= {\bf 0}$.

In the first situation, $t'^{\min}_w=t^{\min}_w=t^{\min}_v=t^{\max}_v=1$, $t'^{\max}_w=t^{\max}_w=t'^{\max}_v=M_2$, and $t'^{\min}_v=M_1$.
The ELC in Definition \ref{defi:extended_longitudional_conditions_paper_1} do not hold because \eqref{eq:extra_condition_to_add_to_longitudional_conditions_paper_1} does not hold (i.e., $1=t^{\min}_v \nless t^{\max}_v=1$) and the ELC in Definition \ref{defi:new_definition_of_LC_paper_1} do not hold because \eqref{eq:longitudional_condition_bento_2_paper_1} does not hold (i.e., $1=t'^{\min}_w \notin [t'^{\min}_v,t'^{\max}_v] = [M_1,M_2]$).
In the second situation, $t'^{\min}_v=t^{\min}_v=1$, $t'^{\min}_w=t^{\min}_w=M_1$,
and $t'^{\max}_v=t^{\max}_v=t'^{\max}_w=t^{\max}_w=M_2$. Both the ELC in Definitions \ref{defi:extended_longitudional_conditions_paper_1} and \ref{defi:new_definition_of_LC_paper_1} hold. In this case, the interpretation of ${\bf M}_{v,:}= {\bf 1}$ and ${\bf M}_{w,:}= {\bf 0}$ is that mutant $v$ was born before the observation window started, while mutant $w$ is born only after the observation window.

Contrarily, we exclude both situations by either using Definition \ref{def:extended_birth_and_death_time_calder_def_paper_1} and the ELC in Definition \ref{defi:extended_longitudional_conditions_paper_1} , or Definition \ref{defi:new_definition_of_birth_and_death_time_paper_1} and the ELC in Definition \ref{defi:new_definition_of_LC_paper_1}.
In this case, for the first situation  $t^{\min}_v=t'^{\min}_v=t^{\max}_v=t'^{\max}_v = t^{\min}_w=t'^{\min}_w=1$ and $t^{\max}_w=t'^{\max}_w=M_1$. This case is excluded because \eqref{eq:extra_condition_to_add_to_longitudional_conditions_paper_1} and \eqref{eq:longitudional_condition_bento_1_paper_1} are violated by mutant $v$. For the second situation $t^{\min}_w=t'^{\min}_w=t^{\max}_w=t'^{\max}_w = t^{\min}_v=t'^{\min}_v=1$ and $t^{\max}_v=t'^{\max}_v=M_1$. This case is excluded because \eqref{eq:extra_condition_to_add_to_longitudional_conditions_paper_1} and \eqref{eq:longitudional_condition_bento_1_paper_1} are violated by mutant $w$.
\end{rmk}

\section{Equivalence between extended longitudinal conditions using either \eqref{eq:extra_long_constraint_no_zero_M_paper_1} or \eqref{eq:extra_condition_to_add_to_longitudional_conditions_paper_1}}\label{app:equivalent_of_extending_calder_long_conditions_in_one_or_in_another_way_paper_1}

\begin{lem}\label{th:equivalent_of_extending_calder_long_conditions_in_one_or_in_another_way_paper_1}
If we use Definition \ref{def:extended_birth_and_death_time_calder_def_paper_1} for birth and death times, then condition \eqref{eq:extra_long_constraint_no_zero_M_paper_1} holds if and only if condition \eqref{eq:extra_condition_to_add_to_longitudional_conditions_paper_1} holds.
\end{lem}
\begin{proof}

{\rev{If \eqref{eq:extra_long_constraint_no_zero_M_paper_1} does not hold, then ${\bf M}_{v,:}  = {\bf 0}$ and hence, by the definition of $t^{\max}_v$ alone we have that $t^{\max}_v = t^{\min}_v$, and so  \eqref{eq:extra_condition_to_add_to_longitudional_conditions_paper_1} does not hold.
}}

If \eqref{eq:extra_condition_to_add_to_longitudional_conditions_paper_1} does not hold, then we consider all possible scenarios in which $t^{\min}_v = t^{\max}_v$.
If $t^{\min}_v$ and $t^{\max}_v$ are both well defined without needing to use their extended definitions,
then $t^{\min}_v = t^{\max}_v$ simultaneously requires $M_{v,t^{\min}_v} > 0$ and $M_{v,t^{\min}_v} = 0$, which is not possible.
If $t^{\max}_v$ is well defined without needing to use its extended definitions but $t^{\min}_v$ requires its extended definition,
then it must be that $t^{\min}_v = 1$ and ${\bf M}_{v,:}  = 0$.
If $t^{\min}_v$ is well defined without needing to use its extended definition but $t^{\max}_v$ requires its extended definition,
then $t^{\max}_v = M_1 > n \geq t^{\min}_v$, which contradicts the assumption that $t^{\min}_v = t^{\max}_v$.

Finally, where defining both $t^{\min}_v$ and $t^{\max}_v$ requires their extended definitions, then we have $t^{\min}_v = 1 < t^{\max}_v = M_1$. This contradicts the assumption that $t^{\min}_v = t^{\max}_v$, and therefore this case cannot occur.
This completes the proof.
\end{proof}

\section{Proof of Lemma \ref{th:equivalent_CALDER_conditions_paper_1}} \label{appendix:proof_of_calder_equivalence_lemma_paper_1}

Throughout this proof, we say that $t'^{\min}_v$ (resp. $t'^{\max}_v, t^{\min}_v$ or $t^{\max}_v$) \emph{follows an exception} if, in order to define it, we need to resort to the special cases in its definition, i.e., we cannot directly apply \eqref{eq:CALDER_t_min_def_paper_1} (resp. \eqref{eq:CALDER_t_max_def_paper_1}, \eqref{eq:our_t_min_def_paper_1} or \eqref{eq:our_t_max_def_paper_1}).

We prove Lemma \ref{th:equivalent_CALDER_conditions_paper_1} by proving Lemmas \ref{th:logitudional_equivalence_lemma_1_paper_1} and \ref{th:logitudional_equivalence_lemma_2_paper_1} below.

\begin{lem}\label{th:logitudional_equivalence_lemma_1_paper_1}
If conditions   \eqref{eq:CALDER_extinction_def_paper_1} and \eqref{eq:longitudional_condition_bento_1_paper_1} hold, then for any mutant $v$, $t'^{\min}_v = t^{\min}_v$ and $t'^{\max}_v = t^{\max}_v$.
\end{lem}

\begin{proof}\hspace{1 mm}We split the proof into the following cases.

{\it Proof that} $t^{\min}_v \leq t'^{\min}_v$: We consider three cases.
If neither quantity follows an exception, and since for any $v,t$ we have $M_{v,t} > 0 \implies F_{v,t} > 0$ (cf. equation \eqref{eq:simple_PPM_model_expanded_paper_1}), it follows directly from their definitions that $t^{\min}_v \leq t'^{\min}_v$.
If $t'^{\min}_v$ follows an exception but $t^{\min}_v$ does not, then the exception cannot be because $M_{v,t} = 0\ \forall t$ (otherwise $t^{\min}_v$ would also follow an exception),
and therefore it must be that $M_{v,1} > 0$ and $t'^{\min}_v = t^{\min}_v=1$.
Finally, if  $t^{\min}_v$ follows an exception, then $M_{v,t} = 0 \ \forall t$, and hence $t'^{\min}_v$ also follows an exception. This implies that $t'^{\min}_v = t^{\min}_v=1$.

{\it Proof that} $t^{\min}_v \geq t'^{\min}_v$: Assume, for the sake of contradiction, that there exists a mutant $v$ such that $t^{\min}_v < t'^{\min}_v$. We consider the following two  possible cases.
\newline
If  $t^{\min}_v$ follows an exception, then we must have $F_{v,t}=0\ \forall t$, which implies $M_{v,t}=0\ \forall t$, and consequently $t^{\min}_v=t'^{\min}_v=1$, a contradiction.
\newline
If $t^{\min}_v$ does not follow an exception,
then $1 \leq t^{\min}_v $, and
together with our assumption we obtain $1 \leq t^{\min}_v < t'^{\min}_v $, which implies that $t'^{\min}_v$ cannot follow an exception (otherwise it would be equal to $1$).
Furthermore, if $M_{v,t^{\min}_v} = 0$, then $t^{\max}_v = t^{\min}_v$, which by \eqref{eq:CALDER_extinction_def_paper_1}, implies that $M_{v,t'^{\min}_v} = 0$, contradicting the definition of $t'^{\min}_v$.
Therefore,
$M_{v,t^{\min}_v} > 0$, and
 there exists $\tilde{t} \in (t^{\min}_v , t'^{\min}_v) $ such that $M_{v,\tilde{t}} = 0$.
 Since $M_{v,t}$ transitions from a positive value to zero somewhere between $t^{\min}_v$ and $\tilde{t}$, and since $t'^{\max}_v$ is the earliest time when such a transition occurs,
 we have $t'^{\max}_v\leq \tilde{t} < t'^{\min}_v$, which contradicts \eqref{eq:longitudional_condition_bento_1_paper_1}.

{\it Proof that} $t^{\max}_v \leq t'^{\max}_v$:
If $t'^{\max}_v$ follows an exception, then either $t'^{\max}_v=M_1$, in which case $t^{\max}_v \leq t'^{\max}_v$,
or $M_{v,t} = 0 \ \forall t$, which implies  $t'^{\min}_v=t'^{\max}_v= 1$, contradicting \eqref{eq:longitudional_condition_bento_1_paper_1}.
If $t'^{\max}_v$  does not follow an exception, then, since we already proved that $t^{\min}_v \leq t'^{\min}_v$, it follows from \eqref{eq:longitudional_condition_bento_1_paper_1} that $t^{\min}_v\leq t'^{\min}_v<t'^{\max}_v$.
Since $M_{v,t'^{\max}_v} = 0$, the definition of $t^{\max}_v$ implies that $t^{\max}_v \leq t'^{\max}_v$.

{\it Proof that} $t^{\max}_v \geq t'^{\max}_v$:
Assume, for the sake of contradiction, that there exists a mutant $v$ such that $t^{\max}_v < t'^{\max}_v$.
It cannot be that $t^{\max}_v$ follows an exception; otherwise, $t^{\max}_v = M_1 < t'^{\max}_v$, which is impossible.
It also cannot be that $t'^{\max}_v$ follows an exception.
Suppose otherwise; we show that each of the following cases leads to a contradiction:\newline
(a) If $M_{v,t} = 0$ for all $t$, then $t'^{\min}_v =t'^{\max}_v= 1$, contradicting \eqref{eq:longitudional_condition_bento_1_paper_1}. \newline
(b) If
$M_{v,t} > 0$ for all $t$, then
$t^{\max}_v = t'^{\max}_v= M_1$, contradicting the assumption that $t^{\max}_v < t'^{\max}_v$.\newline
(c) Otherwise, $M_{v,t}$ is neither identically zero nor strictly positive for all $t$.
If $t'^{\max}_v$ follows an exception, then $M_{v,1}$ cannot be zero;  because if $t'^{\max}_v$ follows an exception, then $M_{v,t}$ can never return to zero within the observation window,
implying that  $M_{v,t}$ is not always zero, which contradicts the definition of the exception.
Thus, $M_{v,t}$ must start at zero and at some point become non-zero and remain non-zero for the remainder of the observation window.
But this implies that  $t^{\max}_v$ follows an exception, which we have already ruled out.
Since we now know that $t'^{\max}_v$ does not follow an exception, it follows that $M_{v,t'^{\max}_v - 1} > 0$.
At the same time, since by assumption
$t'^{\max}_v - 1 \geq t^{\max}_v$,
 it follows from \eqref{eq:CALDER_extinction_def_paper_1}  that $M_{v,t'^{\max}_v - 1} =0$, which is a contradiction.
 \end{proof}

\begin{lem} \label{th:logitudional_equivalence_lemma_2_paper_1}
Conditions  \eqref{eq:CALDER_cont_def_paper_1}, \eqref{eq:CALDER_extinction_def_paper_1} and \eqref{eq:extra_condition_to_add_to_longitudional_conditions_paper_1} imply \eqref{eq:longitudional_condition_bento_1_paper_1} and \eqref{eq:longitudional_condition_bento_2_paper_1}, and vice versa.
\end{lem}

\begin{proof}\hspace{1 mm} We split the proof into the following cases.

    \emph{Proof that \eqref{eq:CALDER_extinction_def_paper_1} and \eqref{eq:extra_condition_to_add_to_longitudional_conditions_paper_1} imply \eqref{eq:longitudional_condition_bento_1_paper_1}}:
    We proceed by contradiction.
    Assume that \eqref{eq:CALDER_extinction_def_paper_1} and \eqref{eq:extra_condition_to_add_to_longitudional_conditions_paper_1} hold, but that condition \eqref{eq:longitudional_condition_bento_1_paper_1} does not hold; that is,  for some mutant $v$, $t'^{\min}_v \geq t'^{\max}_v$.
    We consider the following two possible situations.

In the first situation, we assume  that neither $t'^{\min}_v$ nor $t'^{\max}_v$ follows an exception.
\rev{ By definition of $t'^{\max}_v$ when there is no exception, $t'^{\max}_v$ is the earliest time at which $M_{v,t}$ drops from non-zero to zero; hence $M_{v,t} > 0$ for $t = t'^{\max}_v -1$.
}

\rev{
By \eqref{eq:simple_PPM_model_expanded_paper_1} it follows that  $F_{v,t'^{\max}_v- 1} > 0$ and hence $t^{\min}_v \leq t'^{\max}_v - 1< t'^{\max}_v$.
Since we’re assuming \eqref{eq:CALDER_extinction_def_paper_1}, we know that ${\bf M}_{v,:}$ is zero at and after $t^{\max}_v$, so we cannot have $t^{\max}_v \leq t'^{\max}_v - 1$, so it must be that  $t'^{\max}_v \leq t^{\max}_v$.}

\rev{
We also have $t^{\max}_v < t'^{\min}_v$ because $M_{v,t'^{\min}_v-1} = 0$, and $t^{\max}_v$ is defined as the earliest time at or after $t^{\min}_v$ at which $M_{v,t}=0$.}

 Condition \eqref{eq:CALDER_extinction_def_paper_1} then implies that $M_{v,{t'}^{\min}_v} = 0$, which contradicts the fact that, since ${t'}^{\min}_v$ does not follow an exception, $M_{v,{t'}^{\min}_v} > 0$.

Now assume that either $t'^{\min}_v$ or $t'^{\max}_v$ follows an exception.
The following four cases exhaust all possibilities. \newline
Case 1: Both $t'^{\min}_v$ and $t'^{\max}_v$ follow an exception and $M_{v,t}=0\ \forall t$. Then $t^{\min}_v=t^{\max}_v=1$, which contradicts \eqref{eq:extra_condition_to_add_to_longitudional_conditions_paper_1}. \newline
Case 2: Both $t'^{\min}_v$ and $t'^{\max}_v$ follow an exception and $M_{v,t} > 0\ \forall t$.
Then $t'^{\max}_v = M_1$ and $t'^{\min}_v = 1 < t'^{\max}_v$,  contradicting the assumption that $t'^{\min}_v \geq t'^{\max}_v$.\newline
Case 3: Only $t'^{\min}_v$ follows an exception. Then, $M_{v,1} > 0$ and $M_{v,t}$ drops to zero for the first time at some time $t'$ with $1<t'<n$.
Since $t'^{\min}_v$ follows an exception,  $M_{v,t}$ is never observed to transition from zero to non-zero and hence remains zero for all $t\leq t'$.
This implies that $t'^{\max}_v=t'$, and since $t'^{\min}_v$ follows an exception,  $t'^{\min}_v =1$. Hence $t'^{\min}_v < t'^{\max}_v$, a contradiction. \newline
Case 4: Only $t'^{\max}_v$ follows an exception. Then $M_{v,n} > 0$, and $M_{v,t}$ transitions from zero to a positive value latest at some time $t'$ with $1\leq t' < n$.
The fact that $t'$ is the latest such time implies that $M_{v,t} > 0$ for all $t \geq t'$.
Furthermore, since
$t'^{\max}_v$ follows an exception, $M_{v,t}$ is never observed to transition from non-zero to zero and hence
remains non-zero for all $t<t'$.
This implies that $t'^{\max}_v = M_1$ and $t'^{\min}_v = t' < n < M_1 = t'^{\max}_v$, a contradiction.

\emph{Proof that \eqref{eq:CALDER_extinction_def_paper_1}, \eqref{eq:CALDER_cont_def_paper_1}, and \eqref{eq:extra_condition_to_add_to_longitudional_conditions_paper_1}
imply \eqref{eq:longitudional_condition_bento_2_paper_1}}:
Assume that \eqref{eq:CALDER_extinction_def_paper_1},
\eqref{eq:CALDER_cont_def_paper_1}, and
\eqref{eq:extra_condition_to_add_to_longitudional_conditions_paper_1} hold.
By definition of $t^{\min}_v$, if $\mathcal{T}_{v,w}=1$ then $t^{\min}_w \geq t^{\min}_v$.
Therefore, \eqref{eq:CALDER_cont_def_paper_1} implies that $t^{\min}_w \in [t^{\min}_v,t^{\max}_v]$.
We have already proved that \eqref{eq:CALDER_extinction_def_paper_1} and \eqref{eq:extra_condition_to_add_to_longitudional_conditions_paper_1} imply \eqref{eq:longitudional_condition_bento_1_paper_1}, and hence \eqref{eq:longitudional_condition_bento_1_paper_1} holds.
Lemma \ref{th:logitudional_equivalence_lemma_1_paper_1} now implies that $t'^{\max}_v=t^{\max}_v$ and $t'^{\min}_v=t^{\min}_v$.
It therefore follows that
$t'^{\min}_w \in [t'^{\min}_v,t'^{\max}_v]$, that is, \eqref{eq:longitudional_condition_bento_2_paper_1} holds.

\emph{Proof that \eqref{eq:longitudional_condition_bento_1_paper_1} implies \eqref{eq:CALDER_extinction_def_paper_1}} and \eqref{eq:extra_condition_to_add_to_longitudional_conditions_paper_1}:\\
Assume that either $t'^{\min}_v$ or $t'^{\max}_v$ follows an exception.
In this case, there are only three possible scenarios for the values that ${\bf M}_{v,:}$ can take:
being always non-zero ($t^{\min}_v = 1;\ t^{\max}_v = M_1$),
going from zero to non-zero during the observation window and remaining non-zero until $n$ ($t^{\min}_v>1; t^{\max}_v=M_1$),
and going from non-zero to zero during the observation window and remaining zero until $n$ ($t^{\min}_v=1; t^{\max}_v\leq n$).
The scenario where ${\bf M}_{v,:}$ is always zero is excluded by the assumption that \eqref{eq:longitudional_condition_bento_1_paper_1} holds.
In each of these scenarios, \eqref{eq:CALDER_extinction_def_paper_1} and \eqref{eq:extra_condition_to_add_to_longitudional_conditions_paper_1} hold.
Now assume that neither $t'^{\min}_v$ nor $t'^{\max}_v$ follows an exception.
If \eqref{eq:longitudional_condition_bento_1_paper_1} holds, then $M_{v,t}$ starts at zero, becomes positive and, before $n$, goes back permanently to zero.
This implies that $t^{\min}_v\leq t'^{\min}_v< t'^{\max}_v\leq t^{\max}_v$, which in turn implies that \eqref{eq:CALDER_extinction_def_paper_1} and \eqref{eq:extra_condition_to_add_to_longitudional_conditions_paper_1} hold.

\emph{Proof that \eqref{eq:longitudional_condition_bento_1_paper_1} and \eqref{eq:longitudional_condition_bento_2_paper_1} imply \eqref{eq:CALDER_cont_def_paper_1}}:
If \eqref{eq:longitudional_condition_bento_1_paper_1} holds, then \eqref{eq:CALDER_extinction_def_paper_1} holds.
Lemma \ref{th:logitudional_equivalence_lemma_1_paper_1} then implies that $t'^{\max}_v=t^{\max}_v$ and $t'^{\min}_v=t^{\min}_v$.
If \eqref{eq:longitudional_condition_bento_2_paper_1} also holds, then $t'^{\max}_v=t^{\max}_v$ and $t'^{\min}_v=t^{\min}_v$, and hence \eqref{eq:CALDER_cont_def_paper_1} holds.
\end{proof}
\section{Proof of Lemma \ref{th:quasi_convex_long_cond_M_no_eps_paper_1}}\label{appd:proof_of_quasi_convex_long_cond_M_no_eps_paper_1}
\begin{proof}
Below we use the terms ``a mutant is born'' and ``a mutant dies" to mean that $M_{v,t}$ goes from zero to non-zero, or from non-zero to zero, respectively.
These terms do not imply or assume that a mutant can only be  born or die once; this is only true if the extended longitudinal conditions hold.
The purpose of the proof is precisely  to show that \eqref{eq:non_negative_M_each_entry_paper_1}--\eqref{eq:right_birth relative_to_father_birth_paper_1} are equivalent to the extended longitudinal conditions holding.
Below, all death times and birth times are defined according to
Definition \ref{defi:new_definition_of_birth_and_death_time_paper_1}.
Below we assume that all observations and statements about $M$ are made for $t$ between $1$ and $n$. To make this point more clear, sometimes we will, for example, write ``${\bf M}_{v,:}$ is \emph{observed} to be zero at some point'' instead of just ``${\bf M}_{v,:}$ is zero at some point'', to emphasize that ``observed'' refers to $t$ between $1$ and $n$.
We write ``mutant (type) $v$ is zero (resp. non-zero) at $t$'' to mean that $M_{v,t} = 0$ (resp. $>0$).

Throughout the proof of Lemma \ref{th:quasi_convex_long_cond_M_no_eps_paper_1}, when we assume that  \eqref{eq:non_negative_M_each_entry_paper_1} holds, i.e. $M \geq 0$, this implies that if $M_{i,t}$ is not positive, then it must be zero, and vice versa.

If $M \in \mathcal{L}$, then by definition of $\mathcal{L}$, $M \geq 0$, and \eqref{eq:non_negative_M_each_entry_paper_1} holds.
Hence, all that remains to prove is that \eqref{eq:non_negative_M_each_entry_paper_1}--\eqref{eq:right_birth relative_to_father_birth_paper_1} imply \eqref{eq:longitudional_condition_bento_1_paper_1}--\eqref{eq:longitudional_condition_bento_2_paper_1}, and that \eqref{eq:longitudional_condition_bento_1_paper_1}--\eqref{eq:longitudional_condition_bento_2_paper_1} imply \eqref{eq:non_zero_M_paper_1}--\eqref{eq:right_birth relative_to_father_birth_paper_1}.
We prove this in steps.

We will use Lemma \ref{th:aux_result_for_equiv_between_only_M_ELC_and_regular_LC_paper_1}, which we prove first, and which states that, when proving that \eqref{eq:non_negative_M_each_entry_paper_1}--\eqref{eq:right_birth relative_to_father_birth_paper_1} imply \eqref{eq:longitudional_condition_bento_1_paper_1}--\eqref{eq:longitudional_condition_bento_2_paper_1}, we can assume that \eqref{eq:right_birth relative_to_father_death_paper_1} is true.
\begin{lem}\label{th:aux_result_for_equiv_between_only_M_ELC_and_regular_LC_paper_1}
Equations \eqref{eq:non_negative_M_each_entry_paper_1}--\eqref{eq:right_birth relative_to_father_birth_paper_1} imply that
\begin{align}
M_{w,t} = 0 \text{ if } \sum^{t}_{k=1} M_{v,k} = 0, \quad  \forall v,w: \mathcal{T}_{v,w}=1,\; \forall t = 1,\dots,n ,\label{eq:right_birth relative_to_father_death_paper_1}
\end{align}
\end{lem}
\begin{proof}
Assume by contradiction that equations \eqref{eq:non_negative_M_each_entry_paper_1}--\eqref{eq:right_birth relative_to_father_birth_paper_1} hold and that there exists $t\geq 1$ and a mutant $v$ with child $w$ such that $M_{v,k} = 0$ for all $ 1\leq k \leq t$ and $M_{w,t} > 0$.

By equation \eqref{eq:non_zero_M_paper_1}, we know that ${\bf M}_{w,:} \neq {\bf 0}$ and hence, using \eqref{eq:non_negative_M_each_entry_paper_1}, there exists a first instant $t'$, where $1\leq t'<t\leq n$, such that ${\bf M}_{w,t'} > 0$.
If $t' > 1$, then $M_{w,t'-1} = 0$ and since $M_{w,t'} > 0$, equation \eqref{eq:right_birth relative_to_father_birth_paper_1} implies that either $M_{v,t'-1}$ or $M_{v,t'}$ is non-zero, which contradicts the fact that $M_{v,k} = 0$ for all $ 1\leq k \leq t$, including at $t'$ and $t'-1$.
If $t' = 1$, then $M_{w,1} > 0$, and \eqref{eq:right_birth relative_to_father_birth_paper_1} implies that $M_{v,1} > 0$, which again contradicts the fact that $M_{v,k} = 0$ for all $ 1\leq k \leq t$, including at $t' = 1$.
\end{proof}

{\em Proof that \eqref{eq:non_negative_M_each_entry_paper_1}, \eqref{eq:non_zero_M_paper_1}, and \eqref{eq:right_form_M_condition_paper_1} imply \eqref{eq:longitudional_condition_bento_1_paper_1}:}
\\
Expression \eqref{eq:right_form_M_condition_paper_1} represents multiple conditions, one for each $t$ and $v$.
Each condition from \eqref{eq:right_form_M_condition_paper_1}, for a fixed $v$ and fixed $t\geq 2$, means that if $v$ is alive sometime before and including $t-1$ but is dead at time $t-1$, then it must be dead at time $t$.
By induction on time for a fixed $v$, \eqref{eq:right_form_M_condition_paper_1} means that either (a) $v$ is never observed being born, i.e. $v$ is always dead or always alive, or (b) if $v$ is observed being born sometime between $t=1$ and $t = n$, then after its birth, if it dies, it dies forever. In other words, if it is observed being born, it is only born once.

The case where $v$ is always dead is excluded by the assumption that \eqref{eq:non_zero_M_paper_1} holds, so either (b) holds or (c) $v$ is always alive.
Note that (b) is an \emph{if statement}. It is possible that a mutant is never observed being born and is never always zero, if it is always non-zero during the observation window, in which case \eqref{eq:right_form_M_condition_paper_1} also holds. We will prove that (b) and (c) imply \eqref{eq:longitudional_condition_bento_1_paper_1}.
We consider all possible scenarios for ${\bf M}_{v,:}$ compatible with (b) or (c) and show that \eqref{eq:longitudional_condition_bento_1_paper_1} holds in all of them.
These scenarios are the following:
either ${\bf M}_{v,:}$ is observed always positive; or it starts positive and is observed going to zero and remains zero; or it starts at zero and is observed becoming positive at some point and remains positive; or it starts at zero and is observed becoming positive at some point and then goes back to zero at some point.

In scenario 1, $t'^{\min}_v=1 < M_1= t'^{\max}_v$.
In scenario 2, $t'^{\min}_v=1$ and $1<t'^{\max}_v\leq n$.
In scenario 3,  $1<t'^{\min}_v < t'^{\max}_v=M_1$.
In scenario 4,  $1<t'^{\min}_v < t'^{\max}_v\leq n$.

{\em Proof that \eqref{eq:longitudional_condition_bento_1_paper_1} implies  \eqref{eq:non_zero_M_paper_1} and \eqref{eq:right_form_M_condition_paper_1}:}
\\
If ${\bf M}_{v,:}$ is always zero then ${t'}^{\min}_v = {t'}^{\max}_v = 1$, which violates
\eqref{eq:longitudional_condition_bento_1_paper_1}. Hence \eqref{eq:longitudional_condition_bento_1_paper_1} implies \eqref{eq:non_zero_M_paper_1}.

If ${\bf M}_{v,:}$ is not always zero, we cannot have a zero between two non-zero instants during the observation window; otherwise, since the definition of ${t'}^{\min}_v$ involves a $\max$, this would lead to ${t'}^{\min}_v > {t'}^{\max}_v$, violating \eqref{eq:longitudional_condition_bento_1_paper_1}. Therefore, if ${\bf M}_{v,:}$ ever becomes zero after being non-zero, it must remain zero, which implies that \eqref{eq:right_form_M_condition_paper_1} holds.

{\em Proof that \eqref{eq:non_negative_M_each_entry_paper_1}, \eqref{eq:non_zero_M_paper_1}, \eqref{eq:right_form_M_condition_paper_1}, and \eqref{eq:right_birth relative_to_father_death_paper_1} imply $t'^{\min}_w \geq t'^{\min}_v$, i.e. the lower bound in  \eqref{eq:longitudional_condition_bento_2_paper_1}:}
\\
We have already proved above that \eqref{eq:non_zero_M_paper_1} and \eqref{eq:right_form_M_condition_paper_1} imply that if a mutant is observed to be born, then it is born only once.
Expression \eqref{eq:right_birth relative_to_father_death_paper_1} represents multiple conditions, one for each $t$ and $v$.
Each condition from \eqref{eq:right_birth relative_to_father_death_paper_1} for a fixed $v$, $w$ and $t$, states that if mutant $v$ is not observed to exist before or at time $t$, then its child $w$ must not exist at time $t$.
By induction on time, this implies that $w$ cannot exist before $v$ is observed to exist.

We now consider all possible scenarios for ${\bf M}_{v,:}$ and ${\bf M}_{w,:}$ that are compatible with $v$ and $w$ being observed to be born at most once and with neither being always zero, and show that in these scenarios either \eqref{eq:right_birth relative_to_father_death_paper_1} does not hold, or, if \eqref{eq:right_birth relative_to_father_death_paper_1} holds, then $t'^{\min}_w \geq t'^{\min}_v$.
These scenarios are:  both $v$ and $w$ are always non-zero; mutant $v$ is always non-zero and $w$ is observed being born; mutant $w$ is always non-zero and $v$ is observed being born; both $v$ and $w$ are observed being born.
Note that without assuming \eqref{eq:non_zero_M_paper_1} and \eqref{eq:right_form_M_condition_paper_1}, it could be that $v$ is observed being born multiple times, and the fact that $t'^{\min}_v$ is defined using a $\max$ would allow both \eqref{eq:right_birth relative_to_father_death_paper_1} to hold and $t'^{\min}_w < t'^{\min}_v$.

In scenario 1, $t'^{\min}_v = t'^{\min}_w = 1$. In scenario 2, $t'^{\min}_v = 1 < t'^{\min}_w$.
In scenario 3, $w$ is observed to exist at a time when $v$ does not exist, which violates \eqref{eq:right_birth relative_to_father_death_paper_1}.
In scenario 4, if \eqref{eq:right_birth relative_to_father_death_paper_1} holds, then $w$ cannot be born before $v$, otherwise, since both $v$ and $w$ are born only once, this would mean $v$ is still zero when $w$ is non-zero, contradicting \eqref{eq:right_birth relative_to_father_death_paper_1}. Therefore, $w$ becomes non-zero no sooner than $v$ becomes non-zero, and thus $t'^{\min}_w \geq t'^{\min}_v$.

{\em Proof that \eqref{eq:non_negative_M_each_entry_paper_1}, \eqref{eq:non_zero_M_paper_1}, \eqref{eq:right_form_M_condition_paper_1},
 \eqref{eq:right_birth relative_to_father_birth_paper_1}
and \eqref{eq:right_birth relative_to_father_death_paper_1}  imply $t'^{\min}_w \leq t'^{\max}_v$ (i.e., the upper bound in  \eqref{eq:longitudional_condition_bento_2_paper_1}):}
\\
We have already proved above that \eqref{eq:non_zero_M_paper_1} and \eqref{eq:right_form_M_condition_paper_1} imply that no mutant is observed to be always zero and that if a mutant is observed to be born, it does so only once. Conditions \eqref{eq:non_zero_M_paper_1} and \eqref{eq:right_form_M_condition_paper_1} also imply that, if a mutant is observed to die (i.e., going from non-zero to zero during the observation window), it does so only once, and remains zero afterwards.

Expression \eqref{eq:right_birth relative_to_father_birth_paper_1} represents multiple conditions, one for each $t$ and $v$.
Each condition from \eqref{eq:right_birth relative_to_father_birth_paper_1} for a fixed $v$, $w$ and $t$, states that if mutant $v$ has a child $w$ and if this child is observed to be born at time $t$ between $2$ and $n$, then its parent must exist at time $t-1$ or time $t$.

We now consider all possible scenarios for ${\bf M}_{v,:}$ and ${\bf M}_{w,:}$ that are compatible with no mutant being always zero and mutants being observed to be born or to die at most once.
These are:
$v$ is never observed to die and $w$ is never observed to be born;
$v$ is never observed to die and $w$ is observed to be born;
$v$ is observed to die and $w$ is never observed to be born;
$v$ is observed to die and $w$ is observed to be born;
We show that in each of these scenarios either
\eqref{eq:right_birth relative_to_father_birth_paper_1} does not hold, or if it does, then $t'^{\min}_w \leq t'^{\max}_v$.

In scenario 1, child $w$ is always non-zero so $t'^{\min}_w=1$.
At the same time, \eqref{eq:right_birth relative_to_father_death_paper_1} implies that $t'^{\min}_v \leq  t'^{\min}_w$, so $t'^{\min}_v =1$. Therefore, $v$ is always non-zero so
\eqref{eq:right_birth relative_to_father_birth_paper_1} holds and $t'^{\max}_v=M_1$.
Therefore, $t'^{\min}_w\leq t'^{\max}_v$.
In scenario 3, $w$ is always non-zero so \eqref{eq:right_birth relative_to_father_birth_paper_1} holds and $t'^{\min}_w=1$. Since $t'^{\max}_v \geq 1$ always, we have $t'^{\min}_w\leq t'^{\max}_v$.
In scenarios 2 and 4, $w$ is observed to be born, so
$2 \leq {t'}^{\min}_w \leq n$.
If \eqref{eq:right_birth relative_to_father_birth_paper_1} holds, then either $M_{v,{t'}^{\min}_w} > 0$ or $M_{v,{t'}^{\min}_w-1} > 0$. Therefore, since \eqref{eq:non_zero_M_paper_1} and \eqref{eq:right_form_M_condition_paper_1} imply that $v$ can die at most once and remains zero thereafter, it must be that at $t={t'}^{\min}_w-1$ mutant $v$ has not died yet.
Hence ${t'}^{\max}_v \geq {t'}^{\min}_w$.

{\em Proof that \eqref{eq:longitudional_condition_bento_1_paper_1} and $t'^{\min}_w \in [t'^{\min}_v, t'^{\max}_v]$, i.e. \eqref{eq:longitudional_condition_bento_2_paper_1} implies \eqref{eq:right_birth relative_to_father_birth_paper_1}:}
\\
We already proved that \eqref{eq:longitudional_condition_bento_1_paper_1} implies \eqref{eq:non_zero_M_paper_1} and \eqref{eq:right_form_M_condition_paper_1}, which imply that no mutant is observed to be always zero,  and that if a mutant is observed to be born, it does so at most once, and if it dies, it remains zero thereafter.

We now consider all possible scenarios for ${\bf M}_{v,:}$ and ${\bf M}_{w,:}$ that are compatible with $v$ and $w$ being observed to be born or to die at most once and not being always zero, and show that in these scenarios either \eqref{eq:longitudional_condition_bento_2_paper_1} does not hold, or, if \eqref{eq:longitudional_condition_bento_2_paper_1} holds, then \eqref{eq:right_birth relative_to_father_birth_paper_1} holds.
These scenarios are:  $v$ is never observed to die
and $w$ is never observed to be born;
$v$ is never observed to die and $w$ is observed to be born;
$v$ is observed to die and $w$ is never observed to be born;
$v$ is  observed to die and $w$ is observed to be born;

In scenario 1,
$v$ is always non-zero, hence \eqref{eq:right_birth relative_to_father_birth_paper_1} holds.

In scenario 3, since $w$ is never observed to be born, we only need to prove that \eqref{eq:right_birth relative_to_father_birth_paper_1} holds for $t = 1$.
If \eqref{eq:longitudional_condition_bento_2_paper_1} holds, then $t'^{\min}_v \leq t'^{\min}_w = 1$, so $t'^{\min}_v=1$.
This means that $v$ is alive at $t = 1$, and hence
\eqref{eq:right_birth relative_to_father_birth_paper_1} holds for $t = 1$.

In scenarios 2 and 4, we have that ${t'}^{\min}_w > 1$, and thus if ${t'}^{\min}_w \in [{t'}^{\min}_v,{t'}^{\max}_v]$, then at $t = {t'}^{\min}_w-1 \geq 1$ mutant $v$ has not died yet.
Thus, it is either alive or has not yet been born.
In the first case, $M_{v,t'^{\min}_w-1} > 0$, so \eqref{eq:right_birth relative_to_father_birth_paper_1} holds. In the second case, since $t'^{\min}_w \geq t'^{\min}_v$, it must be that $v$ is born at $t = t'^{\min}_w$, so $M_{v,{t'}^{\min}_w} > 0$, and \eqref{eq:right_birth relative_to_father_birth_paper_1} holds.

\end{proof}

\section{Alternative statement for last assumption in Assumption \ref{ass:continuity_of_M_for_calder_conditions_degeneracy_computations_paper_1}} \label{app:last_condition_forces_all_mutants_to_be_alive_is_the_same_as_not_doing_it_paper_1}

Lemma \ref{th:ass_1_with_all_non_zero_and_ass_1_with_a_few_non_zero_paper_1} shows that in Assumption \ref{ass:continuity_of_M_for_calder_conditions_degeneracy_computations_paper_1}, we could have stated the last assumption by conditioning on all the mutants in both $\mathcal{S}_1$ and $\mathcal{S}_2$ being alive, not just on at least one mutant being alive in each of $\mathcal{S}_1$ and in $\mathcal{S}_2$.

\begin{lem}\label{th:ass_1_with_all_non_zero_and_ass_1_with_a_few_non_zero_paper_1}
Let $\mathcal{S}_1$ and $\mathcal{S}_2$ be two disjoint multisets of mutants.
If the joint distribution of abundances of  $\mathcal{S}_1$ and $\mathcal{S}_2$ at time $t$ conditioned on $U^*$ and on all mutants being alive at time $t$ is absolutely continuous, then,
if $\mathbb{P}(  M^*_{\mathcal{S}_1,t},M^*_{\mathcal{S}_2,t}>0\mid U^*) > 0$, we have
$$\mathbb{P}( M^*_{\mathcal{S}_1,t} = M^*_{\mathcal{S}_2,t} \mid  M^*_{\mathcal{S}_1,t}, M^*_{\mathcal{S}_2,t}>0,U^*) = 0.$$
\end{lem}
\begin{proof}
We consider partitions of $\mathcal{S}_1$ of the form $\mathcal{S}_1 = \mathcal{A} \cup \mathcal{A}^c$, where $\mathcal{A}^c$ is the complement of $\mathcal{A}$ in $\mathcal{S}_1$, and likewise we partition $\mathcal{S}_2 = \mathcal{B} \cup \mathcal{B}^c$.
In the sums over $\mathcal{A}$ and $\mathcal{B}$ below, we consider all different ways of partitioning $\mathcal{S}_1$ such that at time $t$ the mutants in $\mathcal{A}$ are alive, those in  $\mathcal{A}^c$ are dead, and $\mathcal{A}$ is never empty, and similarly for $\mathcal{S}_2$. Note that if there are multiple copies of a given mutant in e.g. $\mathcal{S}_1$, it is not possible that one copy is in $\mathcal{A}$ and another one is in $\mathcal{A}^c$ .
We  define ${\bf M}^*_{\mathcal{S},t} \equiv \{M^*_{j,t}: j\in \mathcal{S} \}$; note the boldface to distinguish this from the previously defined scalar $M^*_{\mathcal{S},t}$.
Write
\begin{align*}
&\mathbb{P}( M^*_{\mathcal{S}_1,t} = M^*_{\mathcal{S}_2,t} \land  M^*_{\mathcal{S}_1,t} > 0 \land  M^*_{\mathcal{S}_2,t}>0\mid U^*) \\
&= \sum_{\mathcal{A},\mathcal{B}} \mathbb{P}(  M^*_{\mathcal{A},t} = M^*_{\mathcal{B},t} \land  {\bf M}^*_{\mathcal{A},t} > 0 \land  M^*_{\mathcal{A}^c,t} = 0 \land {\bf M}^*_{\mathcal{B},t} > 0\\
&\qquad\qquad\land  M^*_{\mathcal{B}^c,t}=0\mid U^*)\\
&\leq \sum_{\mathcal{A},\mathcal{B}}
\mathbb{P}(  M^*_{\mathcal{A},t} = M^*_{\mathcal{B},t} \land  {\bf M}^*_{\mathcal{A},t} > 0  \land {\bf M}^*_{\mathcal{B},t} > 0\mid U^*) \\
&=\sum_{\substack{\mathcal{A},\mathcal{B}:\\ \mathbb{P}({\bf M}^*_{\mathcal{A},t} > 0  \land {\bf M}^*_{\mathcal{B},t} > 0 \mid U^* ) > 0}} \hspace{-1cm}
\mathbb{P}( M^*_{\mathcal{A},t} -  M^*_{\mathcal{B},t} = 0\mid   {\bf M}^*_{\mathcal{A},t} ,  {\bf M}^*_{\mathcal{B},t} > 0,U^*)\\
&\qquad \qquad \times \mathbb{P}({\bf M}^*_{\mathcal{A},t} , {\bf M}^*_{\mathcal{B},t} > 0\mid U^*).
\end{align*}

Observe that, conditioned on ${\bf M}^*_{\mathcal{A},t} > 0  \land  {\bf M}^*_{\mathcal{B},t} > 0$, the random variable $M^*_{\mathcal{A},t} - M^*_{\mathcal{B},t}$ is a linear function of random variables with an absolutely continuous joint distribution and hence is also absolutely continuous.
Therefore, the probability that it takes any particular value is zero,
so
$$\mathbb{P}( M^*_{\mathcal{A},t} -  M^*_{\mathcal{B},t} = 0\mid   {\bf M}^*_{\mathcal{A},t} ,  {\bf M}^*_{\mathcal{B},t} > 0,U^*) = 0,$$
and thus the entire sum equals zero.
\end{proof}

\section{Main lemmas for proving Theorem \ref{th:probability_of_solution_without_long_cond_and_with_is_the_same_paper_1}} \label{app:auxiliary_results_for_proving_prob_equiv_calder_paper_1}

Theorem \ref{th:expected_number_of_solutions_without_long_cond_and_with_is_the_same_paper_1} follows from Theorem \ref{th:probability_of_solution_without_long_cond_and_with_is_the_same_paper_1}. To prove Theorem \ref{th:probability_of_solution_without_long_cond_and_with_is_the_same_paper_1} we make use of Lemmas \ref{th:M_not_zero_paper_1}, \ref{th:new_M_statistifes_death_longitudional_condition_paper_1}, \ref{th:M_infered_no_zero_when_Mstart_in_calder_is_non_negative_paper_1}, and \ref{th:proof_that_M_star_satisfies_t_min_child_t_max_father_relationship_implies_M_also_does_paper_1}.

\begin{lem}\label{th:M_not_zero_paper_1}
Let $M^*$ satisfy Assumption \ref{ass:continuity_of_M_for_calder_conditions_degeneracy_computations_paper_1}.
We have that
\begin{equation}
\mathbb{P}(\text{assumption \ref{ass:assumption_no_root_node_is_zero_paper_1} holds}
) = 1.
\end{equation}
\end{lem}
\begin{proof}
Condition \eqref{eq:extra_condition_to_add_to_longitudional_conditions_paper_1} in Assumption \ref{ass:continuity_of_M_for_calder_conditions_degeneracy_computations_paper_1} implies \eqref{eq:extra_long_constraint_no_zero_M_paper_1} (cf. Section \ref{sec:redefine_birth_death_time_original_paper_1}). Hence, the probability that ${\bf M}^*_{r^*,:} = {\bf 0}$ is zero, and thus Assumption \ref{ass:assumption_no_root_node_is_zero_paper_1} holds with probability $1$.
\end{proof}

\begin{lem}\label{th:new_M_statistifes_death_longitudional_condition_paper_1}
Let $M^*$ satisfy Assumption \ref{ass:continuity_of_M_for_calder_conditions_degeneracy_computations_paper_1}.
 Let $U^*$ and $U$ be fixed ancestry matrices, and $M = U^{-1} U^* M^*$.
 Then,
\begin{equation}
\mathbb{P}(M \text{ does not satisfy \eqref{eq:right_form_M_condition_paper_1}}
\;\land\;M \geq 0\mid U,U^*) = 0.
\end{equation}
\end{lem}

\begin{lem}\label{th:M_infered_no_zero_when_Mstart_in_calder_is_non_negative_paper_1}
Let $M^*$ satisfy Assumption \ref{ass:continuity_of_M_for_calder_conditions_degeneracy_computations_paper_1}.
 Let $U^*$ and $U$ be fixed ancestry matrices, and $M = U^{-1} U^* M^*$.
 Then,
\begin{equation}
\mathbb{P}(M \text{ violates \eqref{eq:non_zero_M_paper_1}}
\;\land\;M \geq 0\mid U, U^*) = 0.
\end{equation}
\end{lem}

\begin{lem}\label{th:proof_that_M_star_satisfies_t_min_child_t_max_father_relationship_implies_M_also_does_paper_1}
Let $M^*$ satisfy Assumption \ref{ass:continuity_of_M_for_calder_conditions_degeneracy_computations_paper_1}.
Let $U^*$ and $U$ be fixed ancestry matrices, and $M = U^{-1} U^* M^*$. Then,
\begin{equation}
\mathbb{P}(M \text{ does not satisfy \eqref{eq:right_birth relative_to_father_birth_paper_1} }
\;\land\; M \geq 0 \mid U,U^*) = 0.
\end{equation}
\end{lem}

To prove Lemmas \ref{th:new_M_statistifes_death_longitudional_condition_paper_1}, \ref{th:M_infered_no_zero_when_Mstart_in_calder_is_non_negative_paper_1}, and \ref{th:proof_that_M_star_satisfies_t_min_child_t_max_father_relationship_implies_M_also_does_paper_1}, we need a series of intermediary results.

\subsection{Intermediary results}

\rev{
We first introduce some notation. Define the set
$\mathcal{S}'^+_i \equiv \Delta^* i$ and define the multiset $\mathcal{S}'^-_i \equiv \cup_{k \in \partial i} \Delta^* k$ where the union of the different sets $\Delta^* k$ retains multiple copies of the same element.
In particular, the operator $\cup$ is such that, for example, $\{1,2,3\}\cup\{1,2,2,4\} = \{1,1,2,2,2,3,4\}$.
}

\rev{
Define the set $\mathcal{S}^+_i \equiv \mathcal{S}'^+_i - \mathcal{S}'^-_i$ and the multiset $\mathcal{S}^-_i \equiv \mathcal{S}'^-_i - \mathcal{S}'^+_i$ where the set difference respects the multiplicities of multisets.
In particular, the operator  $-$ is such that, for example, $\{1,2,2,4\} - \{1,2,3\} = \{1,2,4\}$ and $\{1,2,3\} - \{1,2,2,4\} = \{3\}$.
}

\rev{
\begin{rmk}\label{rmk:S_do_not_depend_on_time_paper_1}
By definition, the (multi)sets $\mathcal{S}^+_i,\mathcal{S}'^+_i$ and $\mathcal{S}^-_i,\mathcal{S}'^-_i$ are independent of $t$. They depend only on $i$, $U$ and $U^*$.
\end{rmk}
}

\rev{
Define also ${\bf M}_{+,i,:} \equiv \sum_{j \in \mathcal{S}^+_i} {\bf M}^*_{j,:}$ and ${\bf M}_{-,i,:} \equiv \sum_{j \in \mathcal{S}^-_i} {\bf M}^*_{j,:}$. We are using the convention that if $\mathcal{S}^+_i = \emptyset$ (resp. $\mathcal{S}^-_i = \emptyset$) then ${\bf M}_{+,i,:} = 0$ (resp. ${\bf M}_{-,i,:} = 0$).
}

\begin{lem}\label{th:representation_of_M_as_difference_paper_1}
Let $U^*$ and $U$ be fixed ancestry matrices, and let $M = U^{-1} U^* M^*$.
For any mutant $i$, we have
\begin{equation}
{\bf M}_{i,:} =
{\bf M}_{+,i,:} - {\bf M}_{-,i,:}.
\end{equation}

Furthermore, if $M^*$ satisfies Assumption \ref{ass:continuity_of_M_for_calder_conditions_degeneracy_computations_paper_1}, then for any mutant $i$ and time $t$, we have that ${M}_{+,i,t},{M}_{-,i,t} \geq 0$ and that $\mathbb{P}(M_{+,i,t} = M_{-,i,t} > 0\mid U,U^*) = 0$.
\end{lem}

\begin{proof}[Proof of Lemma~\ref{th:representation_of_M_as_difference_paper_1}]
By direct calculation,
\begin{align*}
{\bf M}_{i,:}
&= (U^{-1}U^* {M}^*)_{i,:}
= ((I - {T})U^* {M}^*)_{i,:} \\
&= \sum_{j \in \Delta^* i} {\bf M}^*_{j,:} - \sum_{k \in \partial i} \sum_{j \in \Delta^* k}{\bf M}^*_{j,:} \\
&= \sum_{j \in \mathcal{S}'^+_i} {\bf M}^*_{j,:} - \sum_{j \in \mathcal{S}'^-_i }{\bf M}^*_{j,:} ,\label{eq:proof_of_M_calder_satisfies_exctinction_eq1_paper_1}
\end{align*}
 where ${T}$ is the matrix representation of the operator that maps children to parents in $U$, $\partial i$ denotes the children of $i$ in $U$, and $\Delta^* i$ (resp. $\Delta^* k$) are the descendants of $i$ (resp. $k$) in $U^*$ including $i$ (resp. $k$) themselves.
 Some of the ${\bf M}^*_{j,:}$ terms  in $\sum_{j \in \mathcal{S}'^+_i} {\bf M}^*_{j,:}$ also appear in $\sum_{j \in \mathcal{S}'^-_i }{\bf M}^*_{j,:}$.
 After cancellations, the remaining terms being added sum to ${\bf M}_{+,i,:}$ and the remaining terms being subtracted sum to ${\bf M}_{-,i,:}$.

The fact that ${\bf M}_{-,i,:},{\bf M}_{+,i,:} \geq {\bf 0}$ follows from the fact that both are sums of entries of $M^*$, all of which are non-negative by Assumption \ref{ass:continuity_of_M_for_calder_conditions_degeneracy_computations_paper_1}.

To prove the last statement, observe that either $\mathbb{P}(M_{+,i,t} > 0 \land  M_{-,i,t} > 0\mid U,U^*) = 0$ or $\mathbb{P}(M_{+,i,t} > 0 \land  M_{-,i,t} > 0\mid U,U^*) > 0$.
In the first case, it immediately follows that $\mathbb{P}(M_{+,i,t} =  M_{-,i,t} > 0\mid U,U^*) = 0$.
In the second case, we can write
\begin{align*}
    \mathbb{P}&(M_{+,i,t} =  M_{-,i,t} > 0\mid U,U^*) \\
    &= \mathbb{P}(M_{+,i,t} =  M_{-,i,t} \land M_{+,i,t} ,  M_{-,i,t}> 0\mid U,U^*)  \\
    &= \frac{\mathbb{P}(M_{+,i,t}  =  M_{-,i,t} \mid M_{+,i,t} ,  M_{-,i,t}> 0, U,U^*)}{\mathbb{P}(M_{+,i,t} ,  M_{-,i,t}> 0\mid U,U^*)}.
\end{align*}
Since by their definitions, $\mathcal{S}^+$ and $\mathcal{S}^-$ involve disjoint sets of mutants, it follows from Assumption \ref{ass:continuity_of_M_for_calder_conditions_degeneracy_computations_paper_1} that
$$\mathbb{P}(M_{+,i,t} =  M_{-,i,t} \mid M_{+,i,t} ,  M_{-,i,t}> 0,U,U^*) = 0.$$
\end{proof}

\begin{lem}\label{th:i_is_in_S_plus_paper_1}
Let $M^*$ satisfy Assumption \ref{ass:continuity_of_M_for_calder_conditions_degeneracy_computations_paper_1}. Let $U^*$ and $U$ be fixed ancestry matrices, and $M = U^{-1} U^* M^*$. Consider any mutant $i$. Let $\mathcal{S}^+_i$ and $\mathcal{S}^-_i$ be defined as in Lemma \ref{th:representation_of_M_as_difference_paper_1}.
The sets $\mathcal{S}^+_i$ and $\mathcal{S}^-_i$ cannot both be empty.
Furthermore,
$\mathbb{P}({\bf M}_{i,:} \geq 0 \land i \notin \mathcal{S}^+_i\mid U,U^*) = 0$.
\end{lem}
\begin{proof}
First we prove by contradiction that it is not possible that $\mathcal{S}^+_i=\mathcal{S}^-_i = \emptyset$. The set $\mathcal{S}'^+_i$ is never empty (always contains at least $i$), so for
$\mathcal{S}^+_i=\mathcal{S}^-_i = \emptyset$ it must be that $i \in \mathcal{S}'^-_i$.
For this to be the case, there must exist a child $j$ of $i$ in $U$ that is an ancestor of $i$ in $U^*$.
However, since $j$ is not in $\mathcal{S}'^+_i$ but is in $\mathcal{S}'^-_i$, we must have that $j \in \mathcal{S}^-_i$, and hence $\mathcal{S}^-_i$ is not empty, which is a contradiction.

Now we prove the second part of the statement.
If $i \notin \mathcal{S}^+_i$, then since
$i \in \mathcal{S}'^+_i$, it must be that $\Delta^* i \subseteq \mathcal{S}'^-_i$, and hence
$\mathcal{S}^+_i=\emptyset$. This implies that $\mathcal{S}^-_i \neq \emptyset$, since we have already shown that the two multisets cannot both be empty.
It follows that
 ${\bf M}_{i,:} = - {\bf M}_{-,i,:} \leq 0$. Since the
 event under consideration requires ${\bf M}_{i,:} \geq 0$, this event implies that ${\bf M}_{-,i,:} = {\bf 0}$, which in turn implies that
 ${\bf M}^*_{j,:} = 0$ for some $j \in \mathcal{S}^- \neq \emptyset$. The probability of this last event is zero, since by assumption $M^*$ satisfies the longitudinal conditions, one of which is \eqref{eq:extra_condition_to_add_to_longitudional_conditions_paper_1}, which is equivalent to \eqref{eq:extra_long_constraint_no_zero_M_paper_1} by Lemma \ref{th:equivalent_of_extending_calder_long_conditions_in_one_or_in_another_way_paper_1} (see Appendix \ref{app:equivalent_of_extending_calder_long_conditions_in_one_or_in_another_way_paper_1}), and which states that no mutant described by $M^*$ is always dead.
\end{proof}

\begin{lem}\label{th:M_non_zero_M_zero_Mstar_non_zero_M_star_zero_paper_1}
Let $M^*$ satisfy Assumption \ref{ass:continuity_of_M_for_calder_conditions_degeneracy_computations_paper_1}.
 Let $U^*$ and $U$ be fixed ancestry matrices, and $M = U^{-1} U^* M^*$.
 For any mutant $i$ and time $t$, we have
\begin{equation}
\mathbb{P}( {\bf M}_{i,:} \geq {\bf 0} \land M_{i,t}  = 0 \land M^*_{i,t}  > 0\mid U,U^*) = 0.
\end{equation}
\end{lem}
\begin{proof}

If $\mathbb{P}( {\bf M}_{i,:} \geq {\bf 0} \mid U^*,U) = 0$, we are done. If $\mathbb{P}( {\bf M}_{i,:} \geq {\bf 0}\mid U^*,U) > 0$, then
\begin{align}
&\mathbb{P}( {\bf M}_{i,:} \geq {\bf 0}\land M_{i,t}  = 0 \land M^*_{i,t}  > 0  \mid U,U^*)  \nonumber\\
&=\mathbb{P}( M_{i,t}  = 0 \land M^*_{i,t}  > 0   \mid U,U^*,{\bf M}_{i,:} \geq {\bf 0})  \nonumber\\
&\qquad\times\mathbb{P}({\bf M}_{i,:} \geq {\bf 0}\mid U,U^*) \label{eq:before_i_in_Splus_in_lemma_paper_1}\\
&=\mathbb{P}( M_{i,t}  = 0 \land M^*_{i,t}  > 0  \land i \in \mathcal{S}^+ \mid U,U^*, {\bf M}_{i,:} \geq {\bf 0}) \nonumber\\
&\qquad\times \mathbb{P}({\bf M}_{i,:} \geq {\bf 0} \mid U,U^*)\label{eq:at_i_in_Splus_in_lemma_paper_1}\\
&=\mathbb{P}( (M_{+,i,t}=M_{-,i,t} = 0 \lor M_{+,i,t}=M_{-,i,t} > 0)\nonumber\\
&\qquad \land M^*_{i,t}  > 0  \land i \in \mathcal{S}^+ \mid U,U^*, {\bf M}_{i,:} \geq {\bf 0})\nonumber\\
&\qquad\times\mathbb{P}({\bf M}_{i,:} \geq {\bf 0} \mid U,U^*)\label{eq:union_bound_in_lemma_paper_1}\\
&=\mathbb{P}( M_{+,i,t}=M_{-,i,t} = 0 \land M^*_{i,t}  > 0  \land i \in \mathcal{S}^+ \mid U,U^*, {\bf M}_{i,:} \geq {\bf 0})\nonumber\\
&\qquad\times
\mathbb{P}({\bf M}_{i,:} \geq {\bf 0}\mid U,U^*)\nonumber\\
&+\mathbb{P}( M_{+,i,t}=M_{-,i,t} > 0 \land M^*_{i,t}  > 0  \land i \in \mathcal{S}^+ \mid U,U^*, {\bf M}_{i,:} \geq {\bf 0})\nonumber\\
&\qquad\times\mathbb{P}({\bf M}_{i,:} \geq {\bf 0} \mid U,U^*)\label{eq:union_bound_broken_in_lemma_paper_1}\\
&\leq \mathbb{P}( M^*_{i,t}= 0 \land M^*_{i,t}  > 0  \mid U,U^*)
\nonumber\\
&\qquad+ \mathbb{P}( M_{+,i,t}=M_{-,i,t} > 0 \mid U,U^*) = 0 \label{eq:last_line_in_lemma_paper_1}
\end{align}
where
\begin{itemize}
\item from \eqref{eq:before_i_in_Splus_in_lemma_paper_1} to \eqref{eq:at_i_in_Splus_in_lemma_paper_1}, we apply Lemma \ref{th:i_is_in_S_plus_paper_1}, which implies that $\mathbb{P}( i \in \mathcal{S}^+ \mid U,U^*, {\bf M}_{i,:} \geq {\bf 0}) = 1$;
\item from \eqref{eq:at_i_in_Splus_in_lemma_paper_1} to \eqref{eq:union_bound_in_lemma_paper_1}, we use the fact that, by Assumption \ref{ass:continuity_of_M_for_calder_conditions_degeneracy_computations_paper_1}, and conditioned on $U^*$,  we have $M^* \geq 0$ with probability $1$.
This implies that $M_{+,i,t} , M_{-,i,t} \geq 0$, and hence  $M_{i,t}=0$ implies that either $M_{+,i,t} = M_{-,i,t} = 0$ or $M_{+,i,t} = M_{-,i,t} > 0$;
\item from \eqref{eq:union_bound_in_lemma_paper_1} to \eqref{eq:union_bound_broken_in_lemma_paper_1}, we use the fact that if an event $A \subseteq B$, then $\mathbb{P}(A)\leq\mathbb{P}(B)$;
\item and in \eqref{eq:last_line_in_lemma_paper_1}, we apply Lemma \ref{th:representation_of_M_as_difference_paper_1}.
\end{itemize}

\end{proof}

\begin{lem}\label{th:th_tmin_tmax_relation_theorem_paper_1}
Let $M^*$ satisfy Assumption \ref{ass:continuity_of_M_for_calder_conditions_degeneracy_computations_paper_1}.
 Let $U^*$ and $U$ be fixed ancestry matrices, and $M = U^{-1} U^* M^*$.
 For any mutant $i$, let ${t'^*}^{\min}_i$ and ${t'^*}^{\max}_i$ be the birth and death time computed from $M^*$ as in Section \ref{sec:redefinition_of_longitudional_conditions_paper_1},  and let ${t'}^{\min}_i$ and ${t'}^{\max}_i$ be the birth and death time computed from $M$ as in Section \ref{sec:redefinition_of_longitudional_conditions_paper_1}.
 We have that
 \begin{align}
 &\mathbb{P}({\bf M}_{i,:} \geq {\bf 0} \land {t'^*}^{\max}_i > {t'}^{\max}_i\mid U,U^*) = 0\label{eq:th_tmin_tmax_relation_theorem_eq_1_paper_1},\\
 &\mathbb{P}({\bf M}_{i,:} \geq {\bf 0} \land {t'}^{\min}_i > {{t'}^*}^{\min}_i\mid U,U^*) = 0, \label{eq:th_tmin_tmax_relation_theorem_eq_2_paper_1}\\
 &\mathbb{P}({\bf M}_{i,:} \geq {\bf 0} \land {t'}^{\min}_i < {{t'}^*}^{\min}_i\mid U,U^*) = 0. \label{eq:th_tmin_tmax_relation_theorem_eq_3_paper_1}
 \end{align}
\end{lem}
\begin{proof}

\rev{We prove \eqref{eq:th_tmin_tmax_relation_theorem_eq_1_paper_1} last.}

To prove \eqref{eq:th_tmin_tmax_relation_theorem_eq_2_paper_1}, notice that if ${t'}^{\min}_i > {{t'}^*}^{\min}_i \geq 1$, then ${M}_{i,1} = 0$, and hence there exists some $1\leq t \leq n$ for which ${M}_{i,t} = 0$ but ${M}^*_{i,t} > 0$. Therefore,
\begin{align*}
\mathbb{P}&({\bf M}_{i,:} \geq {\bf 0} \land {t'}^{\min}_i > {{t'}^*}^{\min}_i\mid U,U^*) \nonumber\\
&\leq  \mathbb{P}({\bf M}_{i,:} \geq {\bf 0} \land (\exists \ 1\leq t \leq n: {M}_{i,t} = 0 \land {M}^*_{i,t} > 0 )  \mid U,U^*)  \\
&\leq \sum^n_{t=1} \mathbb{P}({\bf M}_{i,:} \geq {\bf 0} \land {M}_{i,t} = 0 \land {M}^*_{i,t} > 0   \mid U,U^*) =0,
\end{align*}
where in the last step we invoke Lemma \ref{th:M_non_zero_M_zero_Mstar_non_zero_M_star_zero_paper_1}.

To prove \eqref{eq:th_tmin_tmax_relation_theorem_eq_3_paper_1}, note that if ${t'}^{\min}_i < {{t'}^*}^{\min}_i$, then for $t = {t'}^{\min}_i$
we have ${M}_{i,t}>0$ but ${M^*}_{i,t}=0$. Using Lemma \ref{th:representation_of_M_as_difference_paper_1}, ${M}_{i,t}>0$ implies that ${M}_{+,i,t} > 0$, which implies that for some $j\in \mathcal{S}^+_i \subseteq \Delta^* i$ we have ${M}^*_{j,t}>0$.
By Assumption \ref{ass:continuity_of_M_for_calder_conditions_degeneracy_computations_paper_1}, $M^*$ satisfies \eqref{eq:longitudional_condition_bento_2_paper_1} with probability $1$.
In particular, when birth times are computed from $M^*$,
the birth time of a child is never before the birth time of its parent.
Hence, by induction on $U^*$, and when birth times are computed from $M^*$, the birth time of $j$ (a descendant of $i$ in $U^*$) is not before the birth time of $i$.
Since ${M}^*_{j,t}>0$ implies that $j$ is alive at time $t$, we have ${{t'}^*}^{\min}_i \leq t = {t'}^{\min}_i$ with probability $1$. Therefore,
\begin{align*}
&\mathbb{P}({\bf M}_{i,:} \geq {\bf 0} \land {t'}^{\min}_i < {{t'}^*}^{\min}_i\mid U,U^*) \\
&\leq \mathbb{P}({\bf M}_{i,:} \geq {\bf 0} \land {t'}^{\min}_i \geq {{t'}^*}^{\min}_i \land {t'}^{\min}_i < {{t'}^*}^{\min}_i\mid U,U^*) \\
&= 0.
\end{align*}

\rev{We finally prove \eqref{eq:th_tmin_tmax_relation_theorem_eq_1_paper_1}.}%
From \eqref{eq:th_tmin_tmax_relation_theorem_eq_2_paper_1} and \eqref{eq:th_tmin_tmax_relation_theorem_eq_3_paper_1} we can add to the event inside the probability in
\eqref{eq:th_tmin_tmax_relation_theorem_eq_1_paper_1} that $t'^{*\min}_i = t'^{\min}_i$. Having this in mind,
let us consider the following two possible scenarios: 1) If $M_1 > {t'^*}^{\max}_i > {t'}^{\max}_i$, then for some $1\leq t \leq n$, ${\bf M}_{i,t} = 0$ but ${\bf M}^*_{i,t} > 0$.
2) If $M_1 = t'^{*\max}_i > t'^{\max}_i$ then, since $M_1 > n$, it cannot be that
 $M^*_i$ is always zero, otherwise $t'^{*\max}_i = 1$ by definition. Now consider the following two sub-scenarios:
 2.1)  If $M^*_i$ is always non-zero then there exists a $t$ for which $M_{i,t} = 0$ but $M^*_{i,t} > 0$;
 2.2) If $M^*_i$ is sometimes zero and other times not zero then, since the ELC hold by Assumption \ref{ass:continuity_of_M_for_calder_conditions_degeneracy_computations_paper_1}, it must be that $t'^{\min}_i < t'^{\max}_i$, and hence $t'^{*\min}_i = t'^{\min}_i < t'^{\max}_i < t'^{*\max}_i = M_1 > n$. Hence at $t = t'^{\max}_i$ we have $M_{i,t} = 0$ but $M^*_{i,t} > 0$ because $M^*_i$ has not died yet, and it has been born.

In short, in all non-impossible scenarios there exists a $t$ for which $M_{i,t} = 0$ but $M^*_{i,t} > 0$. Therefore,
\begin{align*}
&\mathbb{P}({\bf M}_{i,:} \geq {\bf 0} \land {t'^*}^{\max}_i > {t'}^{\max}_i\mid U,U^*) \\
&=\mathbb{P}({\bf M}_{i,:} \geq {\bf 0} \land M={t'^*}^{\max}_i > {t'}^{\max}_i\mid U,U^*)
\\
&\qquad
+
\mathbb{P}({\bf M}_{i,:} \geq {\bf 0} \land M>{t'^*}^{\max}_i > {t'}^{\max}_i\mid U,U^*)\\
&\leq 2\ \mathbb{P}( {\bf M}_{i,:} \geq {\bf 0} \land (\exists \ 1\leq t\leq n: M_{i,t}  = 0 \land M^*_{i,t}  > 0) \mid U,U^*)\\
&\leq \sum^n_{t = 1} 2\ \mathbb{P}( {\bf M}_{i,:} \geq {\bf 0} \land M_{i,t}  = 0 \land M^*_{i,t}  > 0 \mid U,U^*) = 0,
\end{align*}
where in the last step we invoke%
Lemma~\ref{th:M_non_zero_M_zero_Mstar_non_zero_M_star_zero_paper_1}.
\end{proof}

\subsection{Proofs of Lemmas \ref{th:new_M_statistifes_death_longitudional_condition_paper_1},\ref{th:M_infered_no_zero_when_Mstart_in_calder_is_non_negative_paper_1} and \ref{th:proof_that_M_star_satisfies_t_min_child_t_max_father_relationship_implies_M_also_does_paper_1}}

\begin{lem*}[\ref{th:new_M_statistifes_death_longitudional_condition_paper_1}, restated]
Let $M^*$ satisfy Assumption \ref{ass:continuity_of_M_for_calder_conditions_degeneracy_computations_paper_1}.
 Let $U^*$ and $U$ be fixed ancestry matrices, and $M = U^{-1} U^* M^*$.
 Then,
\begin{equation}
\mathbb{P}(M \text{does not satisfy \eqref{eq:right_form_M_condition_paper_1}}
\;\land\;M \geq 0\mid U,U^*) = 0.
\end{equation}
\end{lem*}
\begin{proof}
Using a union bound over the mutants, it suffices to show that for any mutant $i$, $\mathbb{P}(E_i\mid U,U^*) = 0$, where $E_i = ({\bf M}_{i,:} \text{ does not satisfy \eqref{eq:right_form_M_condition_paper_1}}
\;\land\;{\bf M}_{i,:} \geq {\bf 0})$.
If $\mathbb{P}({\bf M}_{i,:} \geq {\bf 0}\mid U,U^*) = 0$, we are done. Hence, from now on we assume that $\mathbb{P}({\bf M}_{i,:} \geq {\bf 0}\mid U,U^*) > 0$.

We start by using the representation in Lemma \ref{th:representation_of_M_as_difference_paper_1} and write
\begin{equation}\label{eq:simple_compact_form_for_M_i_from_U_star_and_M_star_paper_1}
{\bf M}_{i,:} =
{\bf M}_{+,i,:} - {\bf M}_{-,i,:}.
\end{equation}
Define the event
\begin{align*}
\tilde{E}_i &= (\exists \ 1 \leq t_1 < t_2 < t_3\leq n :\\
&\qquad
{M}_{+,i,t_1} > {M}_{-,i,t_1} \geq 0, {M}_{+,i,t_3} > {M}_{-,i,t_3} \geq 0 \\
&\qquad \land {M}_{+,i,t_2} = {M}_{-,i,t_2}),
\end{align*}
and the disjoint events
\begin{align*}
\tilde{E}_{i,0} &= (\exists 1 \leq t_1 < t_2 < t_3\leq n :\\
&\qquad {M}_{+,i,t_1} > {M}_{-,i,t_1} \geq 0, {M}_{+,i,t_3} > {M}_{-,i,t_3} \geq 0 \\
&\qquad \land {M}_{+,i,t_2} = {M}_{-,i,t_2} = 0),
\end{align*}
and
\begin{align*}
\tilde{E}_{i,1}
&= (\exists 1 \leq t_1 < t_2 < t_3\leq n : \\
&\qquad {M}_{+,i,t_1} > {M}_{-,i,t_1} \geq 0, {M}_{+,i,t_3} > {M}_{-,i,t_3} \geq 0 \\
&\qquad \land {M}_{+,i,t_2} = {M}_{-,i,t_2} > 0).
\end{align*}
The event that ${\bf M}_{i,:}$ does not satisfy \eqref{eq:right_form_M_condition_paper_1} implies $\tilde{E}_i$.
Therefore, we can write
\begin{align}
&\mathbb{P}(E_i\mid U,U^*) \nonumber\\
&\leq \mathbb{P}(\tilde{E}_i\mid{\bf M}_{i,:} \geq {\bf 0},U,U^*)\mathbb{P}({\bf M}_{i,:} \geq {\bf 0}\mid U,U^*) \label{eq:dead_for_ever_proof_eq_1_paper_1}\\
& = \mathbb{P}(\tilde{E}_i \land i \in \mathcal{S}^+_i\mid{\bf M}_{i,:} \geq {\bf 0},U,U^*)\mathbb{P}({\bf M}_{i,:} \geq {\bf 0}\mid U,U^*) \label{eq:dead_for_ever_proof_eq_2_paper_1}\\
& = \mathbb{P}(\tilde{E}_{i,0} \land i \in \mathcal{S}^+_i\land{\bf M}_{i,:} \geq {\bf 0}\mid U,U^*) \label{eq:dead_for_ever_proof_eq_3_paper_1}\\
&\quad + \mathbb{P}(\tilde{E}_{i,1} \land i \in \mathcal{S}^+_i\mid {\bf M}_{i,:} \geq {\bf 0},U,U^*)\ \mathbb{P}({\bf M}_{i,:} \geq {\bf 0}\mid U,U^*) \nonumber \\
& = \mathbb{P}(\tilde{E}_{i,1} \land i \in \mathcal{S}^+_i \land {\bf M}_{i,:} \geq {\bf 0}\mid U,U^*)\label{eq:dead_for_ever_proof_eq_4_paper_1}\\
&\leq \mathbb{P}( \exists \ 1 < t_2 <  n : {M}_{+,i,t_2} = {M}_{-,i,t_2} > 0 \mid U,U^*) \label{eq:dead_for_ever_proof_eq_5_paper_1}\\
& \leq \sum^{n-1}_{t_2=2} \mathbb{P}( {M}_{+,i,t_2} = {M}_{-,i,t_2} > 0  \mid U,U^*) = 0,\label{eq:dead_for_ever_proof_eq_6_paper_1}
\end{align}
where each step is justified as follows:
\begin{itemize}
\item from \eqref{eq:dead_for_ever_proof_eq_1_paper_1} to \eqref{eq:dead_for_ever_proof_eq_2_paper_1}, we apply Lemma \ref{th:i_is_in_S_plus_paper_1}, which implies that $\mathbb{P}(i \in \mathcal{S}^+_i\mid {\bf M}_{i,:} \geq {\bf 0}) = 1$;
\item from \eqref{eq:dead_for_ever_proof_eq_2_paper_1} to \eqref{eq:dead_for_ever_proof_eq_3_paper_1}, we apply Lemma \ref{th:representation_of_M_as_difference_paper_1}, which states that ${\bf M}_{-,i,:},{\bf M}_{+,i,:} \geq {\bf 0}$.
This implies that these values are either zero or positive, never negative, and hence we can partition $\tilde{E}_i$ as
$\tilde{E}_i = \tilde{E}_{i,0} \cup \tilde{E}_{i,1}$ with probability $1$;
\item from \eqref{eq:dead_for_ever_proof_eq_3_paper_1} to \eqref{eq:dead_for_ever_proof_eq_4_paper_1}, we use the fact that the event $\tilde{\tilde{E}}_i =(\exists\, 1 \leq t_1<t_2<t_3\leq n: {M}_{+,i,t_1}>0,\  {M}_{+,i,t_3} > 0 \land {M}_{+,i,t_2} = 0)$ satisfies $\tilde{E}_{i,0} \subseteq \tilde{\tilde{E}}_i$
and that
$$\mathbb{P}(\tilde{\tilde{E}}_i \land i \in \mathcal{S}^+_i\land{\bf M}_{i,:} \geq {\bf 0}\mid U,U^*) = 0.$$
This fact is proved at the end of the proof.
This implies that
\begin{align*}
\mathbb{P}&({\tilde{E}}_{i,0} \land i \in \mathcal{S}^+_i\land{\bf M}_{i,:} \geq {\bf 0}\mid U,U^*) \\
&\leq  \mathbb{P}(\tilde{{\tilde{E}}}_{i} \land i \in \mathcal{S}^+_i\land{\bf M}_{i,:} \geq {\bf 0}\mid U,U^*)\ =\ 0;
\end{align*}
\item from \eqref{eq:dead_for_ever_proof_eq_4_paper_1} to \eqref{eq:dead_for_ever_proof_eq_5_paper_1}, we remove overlapping events to obtain an upper bound;
\item from \eqref{eq:dead_for_ever_proof_eq_5_paper_1} to \eqref{eq:dead_for_ever_proof_eq_6_paper_1}, we apply a union bound and invoke Lemma \ref{th:representation_of_M_as_difference_paper_1} to conclude that the expression equals zero.
\end{itemize}

It remains to prove that
$$\mathbb{P}(\tilde{\tilde{E}}_i \land i \in \mathcal{S}^+_i\land{\bf M}_{i,:} \geq {\bf 0}\mid U,U^*) = 0.$$
If there exists $1 \leq t_1 < t_2 < t_3\leq n$ such that ${M}_{+,i,t_1} ,{M}_{+,i,t_3} > 0$,
${M}_{+,i,t_2} = 0$,  ${\bf M}_{i,:} \geq 0$,
and $i \in \mathcal{S}^+$,
then
${M^*}_{j_1,t_3}  > 0$ for some mutant $j_1 \in \mathcal{S}^+_i$,
and ${M}^*_{j,t_2} = 0$ for all $j \in  \mathcal{S}^+_i$. Since $i,j_1 \in \mathcal{S}^+_i$, it follows that ${M}^*_{i,t_2} = {M}^*_{j_1,t_2} = 0$.
\newline
Since $M^*$ satisfies the longitudinal conditions with probability $1$, the fact that $i$ is dead at $t = t_2$ and $j_1$ is alive after $t_2$ but dead at $t_2$, implies that $t^{\max}_i \leq t_2 \leq t^{\min}_{j_1}$.
At the same time, since $i$ is an ancestor of $j_1$ (or possibly $j_1 = i$), the longitudinal conditions (cf. \eqref{eq:CALDER_cont_def_paper_1}) imply that $t^{\min}_{j_1} \leq t^{\max}_i$.
Therefore, we must have $t^{\min}_{j_1} = t^{\max}_i$, and  because of \eqref{eq:extra_condition_to_add_to_longitudional_conditions_paper_1}, this requires that $i \neq j_1$. Hence, $i$ is a strict ancestor of $j_1$.
Furthermore, by Assumption \ref{ass:continuity_of_M_for_calder_conditions_degeneracy_computations_paper_1}, the probability that a child is born exactly when its parent dies is zero.
By finite induction, the probability that $t^{\min}_{j_1} = t^{\max}_i$, i.e., that a mutant is born exactly when one of its ancestors dies, is zero.
Thus,
$$\mathbb{P}(\tilde{\tilde{E}}_i \land i \in \mathcal{S}^+_i\land{\bf M}_{i,:} \geq {\bf 0}\mid U,U^*),$$
which is bounded by the probability of the event that $t^{\min}_{j_1} = t^{\max}_i$,
is equal to zero.
\end{proof}

\begin{lem*}[\ref{th:M_infered_no_zero_when_Mstart_in_calder_is_non_negative_paper_1}, restated]
Let $M^*$ satisfy Assumption \ref{ass:continuity_of_M_for_calder_conditions_degeneracy_computations_paper_1}.
 Let $U^*$ and $U$ be fixed ancestry matrices, and let $M = U^{-1} U^* M^*$.
 Then,
\begin{equation}
\mathbb{P}(M \text{ violates \eqref{eq:non_zero_M_paper_1}}
\;\land\;M \geq 0\mid U, U^*) = 0.
\end{equation}
\end{lem*}
\begin{proof}

Using a union bound, it is sufficient to prove that
$\mathbb{P}({\bf M}_{i,:} = {\bf 0}\mid U, U^*) = 0$ for any mutant $i$.
If $\mathbb{P}({\bf M}_{i,:} \geq {\bf 0}\mid U, U^*) = 0$, then we are done. From now on, we assume that $\mathbb{P}({\bf M}_{i,:} \geq {\bf 0}\mid U, U^*) > 0$.

Use Lemma \ref{th:representation_of_M_as_difference_paper_1} to write
\begin{equation}
{\bf M}_{i,:} =
{\bf M}_{+,i,:} - {\bf M}_{-,i,:},
\end{equation}
where the definition of ${\bf M}_{+,i,:}$ depends on the set $\mathcal{S}^+_i$ defined in Lemma \ref{th:representation_of_M_as_difference_paper_1}.
Now write
\begin{align*}
    \mathbb{P}&({\bf M}_{i,:} = {\bf 0}\mid U, U^*)
    = \mathbb{P}({\bf M}_{i,:} = {\bf 0} \land {\bf M}_{i,:} \geq {\bf 0}\mid U, U^*)\\
    &= \mathbb{P}({\bf M}_{i,:} = {\bf 0} \mid {\bf M}_{i,:} \geq {\bf 0} , U, U^*) \mathbb{P}( {\bf M}_{i,:} \geq {\bf 0} \mid U, U^*) \\
    &= \mathbb{P}({\bf M}_{i,:} = {\bf 0} \land i \in \mathcal{S}^+_i\mid {\bf M}_{i,:} \geq {\bf 0} , U, U^*) \mathbb{P}( {\bf M}_{i,:} \geq {\bf 0} \mid U, U^*),
\end{align*}
where in the last step we apply Lemma \ref{th:i_is_in_S_plus_paper_1}, which implies that $\mathbb{P}(i \in \mathcal{S}^+_i\mid {\bf M}_{i,:} \geq {\bf 0} , U, U^*) = 1$.
Since by Assumption \ref{ass:continuity_of_M_for_calder_conditions_degeneracy_computations_paper_1}, $M^*$ satisfies the longitudinal conditions with probability $1$, we have that ${\bf M}^*_{i,:}\neq {\bf 0}$ with probability $1$.
Also, by Assumption \ref{ass:continuity_of_M_for_calder_conditions_degeneracy_computations_paper_1} we have that ${\bf M}^*_{i,:} \geq {\bf 0}$ with probability $1$. Hence, we can  write
\begin{align*}
\mathbb{P}&({\bf M}_{i,:} = {\bf 0}\mid U, U^*)  \\
& = \mathbb{P}({\bf M}_{i,:} = {\bf 0} \land i \in \mathcal{S}^+_i \land  {\bf M}^*_{i,:}\neq {\bf 0}\land  {\bf M}^*_{i,:}\geq {\bf 0}\mid U, U^*).
\end{align*}

Consider the event inside the last probability expression above.
Let $t$ be a time point such that $M^*_{i,t} > 0$.
For $M_{i,t}$ to be zero, it must be that ${M_{-,i,t}} = {M_{+,i,t}}$. Since $i \in \mathcal{S}^+_i$, we have ${M_{+,i,t}} \geq M^*_{i,t} > 0$. Hence,
$$\mathbb{P}({\bf M}_{i,:} = {\bf 0}\mid U, U^*) \leq \sum^n_{t=1} \mathbb{P}({M_{-,i,t}} = {M_{+,i,t}} > 0\mid U,U^*),$$
which, by Lemma \ref{th:representation_of_M_as_difference_paper_1} is equal to $0$.
\end{proof}

\begin{lem*}[\ref{th:proof_that_M_star_satisfies_t_min_child_t_max_father_relationship_implies_M_also_does_paper_1}, restated]
Let $M^*$ satisfy Assumption \ref{ass:continuity_of_M_for_calder_conditions_degeneracy_computations_paper_1}.
Let $U^*$ and $U$ be fixed ancestry matrices, and $M = U^{-1} U^* M^*$. Then,
\begin{equation}
\mathbb{P}(M \text{ does not satisfy \eqref{eq:right_birth relative_to_father_birth_paper_1} }
\;\land\; M \geq 0 \mid U,U^*) = 0.
\end{equation}
\end{lem*}

\begin{proof}
If $M$ does not satisfy \eqref{eq:right_birth relative_to_father_birth_paper_1}, then there exists a time $t$, and a mutant $v$ with child $w$ in $U$, such that $M_{w,t} > 0 \land M_{w,t-1} = 0 \land M_{v,t}+M_{v,t-1} = 0$.
It should be understood that if $t =  1$, then these expressions, as well as the expressions below, should be read with the convention $M_{w,t-1} = M_{v,t-1} = 0$.

We will first prove that for any $v$ and child $w$ in $U$ we have
\begin{align*}
\mathbb{P}&(M \geq 0 \land M_{w,t} > 0 \land M_{w,t-1} = 0 \land M_{v,t}+M_{v,t-1} = 0\\
&\mid U,U^*) = \ 0.
\end{align*}
Using a union bound over $v$ and $t$ then finishes the proof.

To prove $\mathbb{P}(M \geq 0 \land M_{w,t} > 0 \land M_{w,t-1} = 0 \land M_{v,t}+M_{v,t-1} = 0\mid U,U^*) = 0$, we write
\begin{align}
&\mathbb{P}(M \geq 0 \land M_{w,t} > 0 \land M_{w,t-1} = 0 \land M_{v,t}+M_{v,t-1} = 0\mid U,U^*) \label{th:proof_that_M_star_satisfies_t_min_child_t_max_father_relationship_implies_M_also_does_eq1_paper_1}\\
&\leq \mathbb{P}(M \geq 0 \land M_{+,w,t} > 0
\land M_{+,v,t}=M_{-,v,t}\mid U,U^*)\label{th:proof_that_M_star_satisfies_t_min_child_t_max_father_relationship_implies_M_also_does_eq2_paper_1}\\
& \leq  \mathbb{P}(M \geq 0 \land M_{+,w,t} > 0
\land M_{-,v,t} = 0\mid U,U^*)\label{th:proof_that_M_star_satisfies_t_min_child_t_max_father_relationship_implies_M_also_does_eq3_paper_1}\\
& \leq  \mathbb{P}(M \geq 0 \land M^*_{j,t} > 0 \text{ for some } j \in \mathcal{S}^+_w \nonumber\\
&\qquad\qquad\land M^*_{j,t} = 0 \forall j \in \mathcal{S}^-_v \mid U,U^*) \label{eq:proof_child_father_birth_time_eq_1_paper_1}
\end{align}
where
\begin{itemize}
\item from \eqref{th:proof_that_M_star_satisfies_t_min_child_t_max_father_relationship_implies_M_also_does_eq1_paper_1} to \eqref{th:proof_that_M_star_satisfies_t_min_child_t_max_father_relationship_implies_M_also_does_eq2_paper_1}, we (a) apply the representation from Lemma \ref{th:representation_of_M_as_difference_paper_1},  which implies that $M_{w,t} > 0$ entails $M_{+,w,t}>0$;
and (b) we note that since $M \geq 0$,
the condition $M_{v,t}+M_{v,t-1} = 0$ implies $M_{v,t} = 0$, which in turn implies  $M_{+,v,t} = M_{-,v,t}$ by Lemma \ref{th:representation_of_M_as_difference_paper_1};
\item from \eqref{th:proof_that_M_star_satisfies_t_min_child_t_max_father_relationship_implies_M_also_does_eq2_paper_1} to \eqref{th:proof_that_M_star_satisfies_t_min_child_t_max_father_relationship_implies_M_also_does_eq3_paper_1}, we use the fact stated in Lemma \ref{th:representation_of_M_as_difference_paper_1} that $\mathbb{P}(M_{+,i,t} = M_{-,i,t} > 0\mid U,U^*) = 0$;
\item from \eqref{th:proof_that_M_star_satisfies_t_min_child_t_max_father_relationship_implies_M_also_does_eq3_paper_1} to \eqref{eq:proof_child_father_birth_time_eq_1_paper_1}, we use the definition of $M_{+,w,t}$ and $M_{-,v,t}$ from Lemma \ref{th:representation_of_M_as_difference_paper_1},
and recall that since $M^* \geq 0$,
if a sum of components of $M^*$ is zero, then each component must be zero. We also recall (cf. Lemma \ref{th:representation_of_M_as_difference_paper_1} and Remark \ref{rmk:S_do_not_depend_on_time_paper_1}) that the (multi)sets $\mathcal{S}^+_w$ do not depend on time.
\end{itemize}

We now consider the following three possible scenarios:
\\
Scenario 1: $w$ is an ancestor of $v$ in $U^*$; \\
Scenario 2: $w$ is neither an ancestor nor a descendant of $v$ in $U^*$; \\
Scenario 3: $w$ is a descendant of $v$ in $U^*$.\\
Recall that by definition $\mathcal{S}'^-_v \supseteq \Delta^* w$ and $\mathcal{S}'^+_v  = \Delta^* v$ (there are no repeated elements in $\mathcal{S}'^+_v$ but there might be in $\mathcal{S}'^-_v$). Now, we analyze each scenario in detail.

\emph{Scenario 1}: If $w$ is an ancestor of $v$ in $U^*$, then $\Delta^* w \supseteq \Delta^* v$ and hence $\mathcal{S}^+_v = \mathcal{S}'^+_v - \mathcal{S}'^-_v = \emptyset$, so that ${\bf M}_{+,v,:} = {\bf 0}$.
In this case, $M \geq 0$ implies ${\bf M}_{v,:} ={\bf M}_{+,v,:}-{\bf M}_{-,v,:} = -{\bf M}_{-,v,:} \geq {\bf 0}$.
Since, conditioned on $U$ and $U^*$, $M^* \geq 0$ holds with probability $1$ by Assumption \ref{ass:continuity_of_M_for_calder_conditions_degeneracy_computations_paper_1}, this further implies that $ {\bf 0}\leq {\bf M}_{-,v,:} \leq {\bf 0}$ element-wise.
Consequently, there must exist some $j \in \mathcal{S}^-_v$ such that ${\bf M}^*_{j,:} = {\bf 0}$, which contradicts Assumption \ref{ass:continuity_of_M_for_calder_conditions_degeneracy_computations_paper_1}, stating that $M^*$ satisfies the longitudinal conditions that require all mutants to be observed to be alive for at least one time sample.
Therefore, we can bound \eqref{eq:proof_child_father_birth_time_eq_1_paper_1} by zero.

\emph{Scenario 2}: If $w$ is not related to $v$ in $U^*$, then $\mathcal{S}^-_v = \mathcal{S}'^-_v  - \mathcal{S}'^+_v = \mathcal{S}'^-_v - \Delta^* v \supseteq \Delta^* w$.
At the same time, we always have $\mathcal{S}^+_w \subseteq \Delta^* w$. Therefore, the event in \eqref{eq:proof_child_father_birth_time_eq_1_paper_1} implies that $M \geq 0 \land M^*_{j,t} = 0 \land M^*_{j,t} > 0 \text{ for some } j \in \mathcal{S}^+_w$, an event which occurs with probability zero.

\emph{Scenario 3}: Assume now that $w$ is a descendant of $v$ in $U^*$.
For any mutant $i$, define birth and death times ${t'^*}^{\min}_i$ and ${t'^*}^{\max}_i$ from $M^*$ and ${t'}^{\min}_i$ and ${t'}^{\max}_i$ from $M$ as  in Section \ref{sec:redefinition_of_longitudional_conditions_paper_1}.
Let $t$ be a time, and $v$ be a mutant with child $w$ in $U$, such that $ M_{w,t} > 0 \land M_{w,t-1} = 0 \land M_{v,t}+M_{v,t-1} = 0$.
It follows that $t'^{\min}_w  = t$, and that ${\bf M}_{v,:}$ is dead at both $t$ and $t-1$ if $t \geq 2$.
This last fact implies that if $t = 1$, then $t'^{\min}_v  > t$, and if $t > 2$, then either $t'^{\min}_v  > t$ or  $t'^{\max}_v < t$.

If $t =1$, we can write
\begin{align}
&\mathbb{P}(M \geq 0 \land M_{w,t} > 0 \land M_{v,t}= 0\mid U,U^*) \\
&\leq \mathbb{P}(  M \geq 0 \land t'^{\min}_v  >  t'^{\min}_w  \mid U,U^*) \label{eq:proof_that_M_star_satisfies_t_min_child_t_max_father_relationship_implies_M_also_does_eq44_paper_1}\\
&= \mathbb{P}(  M \geq 0 \land {t'^*}^{\min}_v  >  {t'^*}^{\min}_w  \mid U,U^*)\label{eq:proof_that_M_star_satisfies_t_min_child_t_max_father_relationship_implies_M_also_does_eq55_paper_1}\\
&\leq \mathbb{P}(  M \geq 0 \land {t'^*}^{\min}_v  >  {t'^*}^{\min}_w  \land  {t'^*}^{\min}_v  \leq  {t'^*}^{\min}_w  \mid U,U^*) \label{eq:proof_that_M_star_satisfies_t_min_child_t_max_father_relationship_implies_M_also_does_eq66_paper_1}\\
& =\ 0,
\end{align}
where
\begin{itemize}
    \item from \eqref{eq:proof_that_M_star_satisfies_t_min_child_t_max_father_relationship_implies_M_also_does_eq44_paper_1} to \eqref{eq:proof_that_M_star_satisfies_t_min_child_t_max_father_relationship_implies_M_also_does_eq55_paper_1}, we apply \eqref{eq:th_tmin_tmax_relation_theorem_eq_2_paper_1}-\eqref{eq:th_tmin_tmax_relation_theorem_eq_3_paper_1} in Lemma \ref{th:th_tmin_tmax_relation_theorem_paper_1}, which show that for any mutant $i$, the event $M\geq 0 \land t'^{\min}_i= {t'^*}^{\min}_i$ occurs with probability $1$;
    \item from \eqref{eq:proof_that_M_star_satisfies_t_min_child_t_max_father_relationship_implies_M_also_does_eq55_paper_1} to \eqref{eq:proof_that_M_star_satisfies_t_min_child_t_max_father_relationship_implies_M_also_does_eq66_paper_1}, we use the fact that $M^*$ satisfies Assumption \ref{ass:continuity_of_M_for_calder_conditions_degeneracy_computations_paper_1} with probability $1$, which implies that the birth time of a parent is never later than that of its child.
    By induction, this implies that, with probability $1$, ${t'^*}^{\min}_v  \leq  {t'^*}^{\min}_w$, since $w$ is a descendant of $v$ in $U^*$.
    Because both conditions ${t'^*}^{\min}_v  >  {t'^*}^{\min}_w$ and ${t'^*}^{\min}_v  \leq  {t'^*}^{\min}_w$ cannot hold simultaneously, the result follows.
\end{itemize}

For the rest of the proof, we consider $t \geq 2$. We can write
\begin{align}
&\mathbb{P}(M \geq 0 \land M_{w,t} > 0 \land M_{w,t-1} = 0 \land M_{v,t}+M_{v,t-1} = 0\mid U,U^*) \nonumber\\
&\leq \mathbb{P}(  M \geq 0 \land t'^{\min}_v  >  t'^{\min}_w  \mid U,U^*) \nonumber\\
&\qquad+ \mathbb{P}(  M \geq 0 \land t'^{\max}_v  <  t'^{\min}_w  \mid U,U^*)\label{eq:proof_that_M_star_satisfies_t_min_child_t_max_father_relationship_implies_M_also_does_eq4_paper_1}\\
&= \mathbb{P}(  M \geq 0 \land {t'^*}^{\min}_v  >  {t'^*}^{\min}_w  \mid U,U^*) \nonumber\\
&\qquad+ \mathbb{P}(  M \geq 0 \land t'^{\max}_v  <  t'^{\min}_w  \mid U,U^*)\label{eq:proof_that_M_star_satisfies_t_min_child_t_max_father_relationship_implies_M_also_does_eq5_paper_1}\\
&\leq \mathbb{P}(  M \geq 0 \land {t'^*}^{\min}_v  >  {t'^*}^{\min}_w  \land  {t'^*}^{\min}_v  \leq  {t'^*}^{\min}_w  \mid U,U^*) \nonumber\\
&\qquad+ \mathbb{P}(  M \geq 0 \land t'^{\max}_v  <  t'^{\min}_w  \mid U,U^*)\label{eq:proof_that_M_star_satisfies_t_min_child_t_max_father_relationship_implies_M_also_does_eq6_paper_1}\\
&= \mathbb{P}(  M \geq 0 \land t'^{\max}_v  <  t'^{\min}_w  \mid U,U^*)\label{eq:proof_that_M_star_satisfies_t_min_child_t_max_father_relationship_implies_M_also_does_eq7_paper_1}\\
&= \mathbb{P}(  M \geq 0 \land {t'^*}^{\max}_v \leq t'^{\max}_v  <  t'^{\min}_w = {t'^*}^{\min}_w \mid U,U^*)\label{eq:proof_that_M_star_satisfies_t_min_child_t_max_father_relationship_implies_M_also_does_eq8_paper_1},
\end{align}
where
\begin{itemize}
\item from \eqref{eq:proof_that_M_star_satisfies_t_min_child_t_max_father_relationship_implies_M_also_does_eq4_paper_1} to \eqref{eq:proof_that_M_star_satisfies_t_min_child_t_max_father_relationship_implies_M_also_does_eq5_paper_1}, we apply \eqref{eq:th_tmin_tmax_relation_theorem_eq_2_paper_1}-\eqref{eq:th_tmin_tmax_relation_theorem_eq_3_paper_1} from Lemma \ref{th:th_tmin_tmax_relation_theorem_paper_1}, which imply that for any mutant $i$, the event $M\geq 0 \land t'^{\min}_i= {t'^*}^{\min}_i$ holds with probability $1$;
\item from \eqref{eq:proof_that_M_star_satisfies_t_min_child_t_max_father_relationship_implies_M_also_does_eq5_paper_1} to \eqref{eq:proof_that_M_star_satisfies_t_min_child_t_max_father_relationship_implies_M_also_does_eq6_paper_1}, we use the fact that $M^*$ satisfies Assumption \ref{ass:continuity_of_M_for_calder_conditions_degeneracy_computations_paper_1} with probability $1$, which implies that the birth time of a parent cannot exceed that of its child.
Hence, by induction, ${t'^*}^{\min}_v  \leq  {t'^*}^{\min}_w$ with probability $1$, since $w$ is a descendant of $v$ in $U^*$;
\item from \eqref{eq:proof_that_M_star_satisfies_t_min_child_t_max_father_relationship_implies_M_also_does_eq7_paper_1} to \eqref{eq:proof_that_M_star_satisfies_t_min_child_t_max_father_relationship_implies_M_also_does_eq8_paper_1}, we apply \eqref{eq:th_tmin_tmax_relation_theorem_eq_1_paper_1} from Lemma \ref{th:th_tmin_tmax_relation_theorem_paper_1}.
\end{itemize}
Let $u$ be the father of $w$ on $U^*$. If $u=v$, since $M^*$ satisfies the longitudinal conditions (in particular  \eqref{eq:longitudional_condition_bento_2_paper_1} with probability $1$), we obtain
\begin{align*}
\mathbb{P}&(  M \geq 0 \land {t'^*}^{\max}_v \leq t'^{\max}_v  <  t'^{\min}_w = {t'^*}^{\min}_w \mid U,U^*) \\
&\leq \mathbb{P}( {t'^*}^{\max}_v <  {t'^*}^{\min}_w \mid U,U^*)\ =\ 0.
\end{align*}
If $u \neq v$, we also arrive at a contradiction, as shown next.

Assume thus that $w$ is a descendant of $v$ in $U^*$, $u$ is the father of $w$ in $U^*$, and that $u \neq v$.
Represent $M_v$ and $M_w$ using  Lemma \ref{th:representation_of_M_as_difference_paper_1} and consider the corresponding sets $\mathcal{S}^+_v,\mathcal{S}^-_v,\mathcal{S}^+_w$ and $\mathcal{S}^-_w$.
We show that it cannot simultaneously hold that (a) $w \notin \mathcal{S}^+_v \cup \mathcal{S}^-_v$ and (b) $u \notin \mathcal{S}^+_v \cup \mathcal{S}^-_v$, where the union of multisets preserves repeated elements.
To see this, recall that $\mathcal{S}'^+_v = \Delta^* v \supseteq \{w,u\}$,
that
$\mathcal{S}^+_v = \mathcal{S}'^+_v - \mathcal{S}'^-_v$ has no repeated elements, and that since $w$ is a child of $v$ in $U$, we have $\mathcal{S}'^-_v =\Delta^* w \cup (\cdots)$.

If (a) holds, then, since by definition $w \in \mathcal{S}'^-_v$, it must be that $w \in \mathcal{S}'^+_v$, and that $w$ appears only once in  $\mathcal{S}'^-_v$.
If (b) also holds, then, since by definition $u \in \mathcal{S}'^+_v$, it must be that $u$ appears only once in $ \mathcal{S}'^-_v$.
However, for $u$ to appear in $\mathcal{S}'^-_v$, it must belong to $\Delta^* a$ for some $a$ that is either $u$ itself or an ancestor of $u$ in $U^*$.
Hence, the multiset $\mathcal{S}'^-_v = \Delta^* a \cup
\Delta^* w \cup (\cdots)$, where the  union preserves repetitions. Therefore, since $w \in \Delta^* a$ and $w \in \Delta^* w$, the multiset $\mathcal{S}'^-_v$ contains $w$ at least twice.
This contradicts (a), which implies that  $\mathcal{S}'^-_v$ contains $w$ only once.

For the remainder of the proof, we assume that either (a) or (b) is true, and show that both lead to a contradiction.

Assume that $w \in \mathcal{S}^+_v \cup \mathcal{S}^-_v$.
If $v$ is dead at time $t$, then $M_{v,t} \equiv M_{+,v,t} -M_{-,v,t} = 0$, and either $M_{+,v,t}=M_{-,v,t} =0$ or $M_{+,v,t}=M_{-,v,t}>0$.
The event $M_{+,v,t}=M_{-,v,t}>0$ has probability zero by Lemma \ref{th:representation_of_M_as_difference_paper_1}.
Since $w \in \mathcal{S}^+_v \cup \mathcal{S}^-_v$ and $\mathcal{S}^+_v \cap \mathcal{S}^-_v = \emptyset$, $M^*_{w,t}$ is a term in exactly one of $M_{+,v,t}$ or $M_{-,v,t}$.
Thus, the event $M_{+,v,t}=M_{-,v,t} =0$ implies $M^*_{w,t} = 0$, as $M^* \geq 0$.
Hence, with probability $1$, the death time of ${\bf M}_{v,:}$ is not smaller than that of ${\bf M}^*_{w,:}$, i.e. $t'^{\max}_v \geq {t'^*}^{\max}_w$.
We can thus start from \eqref{eq:proof_that_M_star_satisfies_t_min_child_t_max_father_relationship_implies_M_also_does_eq7_paper_1} and write,
\begin{align*}
    \mathbb{P}&(  M \geq 0\ \land \ t'^{\max}_v  <  t'^{\min}_w) \\
    &=      \mathbb{P}(  M \geq 0\ \land\ t'^{\max}_v  <  t'^{\min}_w={t^*}'^{\min}_w
    \ \land\ t'^{\max}_v \geq {t^*}'^{\max}_w) \\
    &\leq  \mathbb{P}(  {t^*}'^{\max}_w \leq t'^{\max}_v  <  {t^*}'^{\min}_w ) = 0,
\end{align*}
where we have used Lemma \ref{th:th_tmin_tmax_relation_theorem_paper_1}, which states that $M \geq 0 \land t'^{\min}_w={t^*}'^{\min}_w$ holds with probability $1$, and the fact that $M^*$ satisfies the longitudinal conditions (in particular \eqref{eq:longitudional_condition_bento_1_paper_1}) with probability $1$.

Similarly, assume that $u \in \mathcal{S}^+_v \cup \mathcal{S}^-_v$.
By the same reasoning as above, we have $t'^{\max}_v \geq {t'^*}^{\max}_u$ with probability $1$.
We can thus start from \eqref{eq:proof_that_M_star_satisfies_t_min_child_t_max_father_relationship_implies_M_also_does_eq7_paper_1} and write,
\begin{align*}
    \mathbb{P}&(  M \geq 0 \land  t'^{\max}_v  <  t'^{\min}_w) \\
    & =      \mathbb{P}(  M \geq 0\ \land \ t'^{\max}_v  <  t'^{\min}_w={t^*}'^{\min}_w \ \land\ t'^{\max}_v \geq {t^*}'^{\max}_u)\\
    &\leq  \mathbb{P}(  {t^*}'^{\max}_u \leq t'^{\max}_v  <  {t^*}'^{\min}_w ) = 0,
\end{align*}
where we again use Lemma \ref{th:th_tmin_tmax_relation_theorem_paper_1}, which ensures $M \geq 0 \land t'^{\min}_w={t^*}'^{\min}_w$ with probability $1$, and  the fact that since $u$ is the father of $w$ and with probability $1$, and $M^*$ satisfies the longitudinal conditions,  \eqref{eq:longitudional_condition_bento_2_paper_1} holds, implying ${t^*}'^{\max}_u \geq  {t^*}'^{\min}_w$.

\end{proof}

\section{Auxiliary results to prove the main results in Section \ref{sec:effect_of_dynamic_restrictions_on_degeneracy_paper_1}}\label{app:definition_of_increments_paper_1}

To prove Theorem \ref{th:main_theorem_for_our_dynamic_constraints_paper_1} we use Lemmas \ref{th:analytical_bound_on_degeneracy_middle_paper_1}, \ref{th:ergodic_expression_goes_to_one_paper_1}, and \ref{th:bound_on_prob_M1_less_M2_under_brownian_motion_paper_1}.
In this section, we provide proofs of these lemmas.

\subsection{Proof of Lemma \ref{th:analytical_bound_on_degeneracy_middle_paper_1}}\label{app:proof_of_lemma_5}

\begin{proof}[Proof of Lemma \ref{th:analytical_bound_on_degeneracy_middle_paper_1}]

We have
    $$        E = 1 + \mathbb{E}_{U^*,M^*}\bigg(\sum_{U \in \mathcal{D}(U^*)} \mathbb{I}(U^{-1}{U^*}M^* \geq 0)\bigg), $$
and we will show that the right-hand side equals \eqref{eq:degeneracy_without_regu_for_one_leaf_perturb_paper_1}.
Similarly,
$$E' \leq 1 + \mathbb{E}_{U^*,M^*}\bigg(\sum_{U \in \mathcal{D}(U^*)} \mathbb{I}(\| U^{-1} U^* \dot{M}^*\| \leq \| \dot{M}^*\| )\bigg) $$
where we have dropped the constraint $U^{-1}{U^*}M^* \geq 0$,
and we will show that the right-hand side equals \eqref{eq:degeneracy_with_regu_for_one_leaf_perturb_paper_1}.

We first derive \eqref{eq:degeneracy_without_regu_for_one_leaf_perturb_paper_1}. By direct calculation and reinterpretation of terms, for any node $i$ we obtain,
\begin{equation}\label{eq:relationship_for_uinv_ustar_mstar_paper_1}
(U^{-1}U^* {M}^*)_{i,:}= ((I - {T})U^* {M}^*)_{i,:} = \sum_{j \in \Delta^* i} {\bf M}^*_{j,:} - \sum_{k \in \partial i} \sum_{j \in \Delta^* k}{\bf M}^*_{j,:},
\end{equation}
 where ${T}$ is the matrix representation of the operator that maps children to their parents in $U$, $\partial i$ denotes the children of $i$ in $U$, and $\Delta^* i$ (resp. $\Delta^* k$) denotes the descendants of $i$ (resp. $k$) in $U^*$, including the node $i$ (resp. $k$) itself.

Let $\partial^* i$ denote the children of $i$ in $U^*$. For all nodes $i$ for which $\partial i = \partial^* i$, we have from \eqref{eq:relationship_for_uinv_ustar_mstar_paper_1} and the non-negativity of $M^*$ that
$$(U^{-1}U^* {M}^*)_{i,:} = {\bf M}^*_{i,:} \geq {\bf 0}.$$
Since $U \in \mathcal{D}(U^*)$ differs from $U^*$ only by a single leaf displacement, there is exactly one node whose children differ, i.e. $\partial i \neq \partial^* i$, namely the new parent (in $U$) of the leaf $j$ in $U^*$ whose parent has changed.
For this node,
\begin{equation}\label{eq:1_left_diff_no_regu_expression_paper_1}
    (U^{-1}U^* {M}^*)_{i,:} = {\bf M}^*_{i,:} - {\bf M}^*_{j,:},
\end{equation}
which can be positive or negative depending on the realization of the random variable $M^*$. The nodes $i$ and $j$ depend on $U$ and $U^*$; however, in what follows, \rev{to simplify notation instead of writing $i(U^*,U)$ and $j(U^*,U)$ we write $i$ and $j$.}

Using the independence of $M^*$ and $U^*$, linearity of expectation, the expression in \eqref{eq:1_left_diff_no_regu_expression_paper_1}, and the fact that the distribution for $M^*$ is, by assumption, invariant to node-label permutations, we write
\begin{align}
        \mathbb{E}_{U^*,M^*}&\left(\sum_{U \in \mathcal{D}(U^*)} \mathbb{I}(U^{-1}{U^*}M^* \geq 0)\right) \\
        &= \mathbb{E}_{U^*}\left(\sum_{U \in \mathcal{D}(U^*)} \mathbb{P}_{M^*}(U^{-1}{U^*}M^* \geq 0)\right)\label{eq:intermedia_expression_deg_without_regu_start_paper_1}\\
        &=\mathbb{E}_{U^*}\left(\sum_{U \in \mathcal{D}(U^*)} \mathbb{P}_{M^*}({\bf M}^*_{i,:} - {\bf M}^*_{j,:} \geq 0)\right)\\
        & = \mathbb{E}_{U^*}\left(\sum_{U \in \mathcal{D}(U^*)} \mathbb{P}_{M^*}({\bf M}^*_{2,:} - {\bf M}^*_{1,:} \geq 0)\right)\\
        & = \mathbb{E}_{U^*}\left(|\mathcal{D}(U^*)| \mathbb{P}_{M^*}({\bf M}^*_{2,:} - {\bf M}^*_{1,:} \geq 0)\right)\\
        &=\mathbb{E}_{U^*}\left(|\mathcal{D}(U^*)| \right)\mathbb{P}_{M^*}({\bf M}^*_{1,:} \leq {\bf M}^*_{2,:} )\label{eq:intermedia_expression_deg_without_regu_end_paper_1}
    \end{align}

If $L(U^*)$ is the number of leaves in the tree associated with $U^*$, the size of $\mathcal{D}(U^*)$ is
$L(U^*)(q-2)$, because each leaf can be made a child of any node except itself and its current parent.
The probability that a node $i$ is a leaf in a random labeled tree on $q$ nodes is $(1-1/q)^{q-2}$~\cite{goldschmidt2016short}.
Since in a rooted tree the root cannot be a leaf, we obtain
$$\mathbb{E}(|\mathcal{D}(U^*)|)=(q-2)\mathbb{E}(L(U^*))=(q-2)(q-1)(1-1/q)^{q-2}.$$
Substituting this expectation into the right-hand side of \eqref{eq:intermedia_expression_deg_without_regu_end_paper_1} and adding $1$ yields \eqref{eq:degeneracy_without_regu_for_one_leaf_perturb_paper_1}.

We now prove \eqref{eq:degeneracy_with_regu_for_one_leaf_perturb_paper_1} using an analogous argument.
If $U \in \mathcal{D}(U^*)$, by the same reasoning as above, for all nodes $i$ for which $\partial i = \partial^* i$,
$$(U^{-1}U^* \dot{{M}}^*)_{i,:} = \dot{{\bf M}}^*_{i,:},$$
and for the unique node $i$ for which $\partial i \neq \partial^* i$,
$$
    (U^{-1}U^* \dot{{M}}^*)_{i,:} = \dot{{\bf M}^*_{i,:}} - \dot{{\bf M}^*_{j,:}}.
$$
Therefore,
\begin{align*}
    \| U^{-1} U^* \dot{M}^*\|^2 - \| \dot{M}^*\|^2  &= \|\dot{{\bf M}^*_{i,:}} - \dot{{\bf M}^*_{j,:}}\|^2 - \|\dot{{\bf M}^*_{i,:}}\|^2 \\
    &= \dot{{\bf M}^*_{j,:}} ^{\top}(\dot{{\bf M}^*_{j,:}} - 2\dot{{\bf M}^*_{i,:}}).
\end{align*}%
We can now adapt the argument in \eqref{eq:intermedia_expression_deg_without_regu_start_paper_1}-\eqref{eq:intermedia_expression_deg_without_regu_end_paper_1} to complete the proof.
\end{proof}

\subsection{Proof of Lemma \ref{th:bound_on_prob_M1_less_M2_under_brownian_motion_paper_1}}

To prove Lemma \ref{th:bound_on_prob_M1_less_M2_under_brownian_motion_paper_1} we first require an intermediate lemma, which we now state and prove.

\begin{lem}
    Let ${\bf M}^*$ be distributed according to Assumption \ref{ass:brownian_motion_paper_1}. Conditioned on fixed values  ${M^*}_{1,0} < {M^*}_{2,0}$ we have that
    \begin{align*}
       \mathbb{P}_{M^*}({\bf M}^*_{1,:} \leq {\bf M}^*_{2,:} \mid {M^*}_{1,0} , {M^*}_{2,0}) &=
\frac{ {M^*_{2,0}}^2-{M^*_{1,0}}^2}{6 n} +O \left(\frac{1}{n^2}
\right)
    \end{align*}
\end{lem}
\begin{proof}
To shorten notation, throughout this proof all probabilities are implicitly
conditioned on fixed values $M^*_{1,0} < M^*_{2,0}$. For example, $\mathbb{P}_{\tilde{\bf M}}(\tilde{\bf M}_{1,:}< 0)$ denotes ``$\mathbb{P}_{\tilde{\bf M}}(\tilde{\bf M}_{1,:}< 0\mid M^*_{1,0} , M^*_{2,0})$ where $ M^*_{1,0} < M^*_{2,0}$.''

The event  ${\bf M}^*_{1,:} \leq {\bf M}^*_{2,:}$ is the same as the event $ (\tilde{\bf M}_{2,:} - \tilde{\bf M}_{1,:})/\sqrt{2} \geq 0$ conditioned on $\tilde{\bf M}_{1,:} , \tilde{\bf M}_{2,:}\geq 0$.
We compute separately (a) the probability that $\tilde{\bf M}_{1,:} , \tilde{\bf M}_{2,:}\geq 0$ and (b) the probability that $(\tilde{\bf M}_{2,:} - \tilde{\bf M}_{1,:})/\sqrt{2} \geq 0 \land \tilde{\bf M}_{1,:} \geq 0$. The ratio of (b) to (a) yields the desired probability.

The event $\tilde{\bf M}_{1,:}\geq 0$ is independent of the event $\tilde{\bf M}_{2,:}\geq 0$. Their probabilities are identical and are well known in literature from the distribution of the hitting time of a standard Brownian motion at a boundary; see e.g., \cite{borodin2012handbook}, and can be derived using a reflection principle. In particular, we have that
\begin{align*}
    \mathbb{P}_{\tilde{ M}}(\tilde{\bf M}_{1,:}< 0)&= \mathbb{P}_{\tilde{\bf M}_{1,:}}\Big(\max_{t\in[0,n]} M^*_{1,0}-\tilde{ M}_{1,t}> M^*_{1,0}\Big)  \\
    &= \mathbb{P}\Big(\max_{t\in[0,n]} { B}_t> M^*_{1,0}\Big),
\end{align*}
where $B_t$ is a standard Brownian motion started at zero. We recall that the observation interval is $t\in[0,n]$, where $n$ represents the ``number'' of samples.
Following, for example, \cite{morters2010brownian}, and using a reflection argument, we obtain $\mathbb{P}(\max_{t\in[0,n]} B_t> M^*_{1,0}) = 2\mathbb{P}(B_n > M^*_{1,0}) = 2 (1-\Phi(M^*_{1,0}/\sqrt{2T}))$,
where $\Phi$ is the cumulative distribution function of a standard Gaussian random variable.
Therefore,
$$\mathbb{P}_{\tilde{\bf M}_{1,:}}(\tilde{\bf M}_{1,:}\geq 0) = 2\Phi(M^*_{1,0}/\sqrt{2n})-1,$$
and
$$ \mathbb{P}_{\tilde{ M}}(\tilde{\bf M}_{1,:},\tilde{\bf M}_{2,:}\geq 0) = (2\Phi(M^*_{1,0}/\sqrt{2n}) - 1)(2\Phi(M^*_{2,0}/\sqrt{2n}) - 1).$$
Let
$${Y}_t \equiv (\tilde{ M}_{2,0} - \tilde{ M}_{1,0})/\sqrt{2} - (\tilde{M}_{2,t} - \tilde{M}_{1,t})/\sqrt{2},$$
and
$${ X}_t \equiv \tilde{ M}_{1,0}- \tilde{M}_{1,t}.$$
The processes ${X}_t$ and ${Y}_{t}$ are correlated standard Brownian motions with zero drift.
Specifically, their correlation coefficient at $t=1$ is equal to $\rho = -1/\sqrt{2}$, their variance at $t=1$ equals to $1$, and both start at zero.
The value of (b) is the probability that the two-dimensional Brownian motion $(X_t,Y_t)$ does not hit the line $x = a \equiv M^*_{1,0} > 0$ nor the line $y = b \equiv (\tilde{ M}_{2,0} - \tilde{ M}_{1,0})/\sqrt{2} >0 $ before $t=n$.
This problem has been studied extensively; see, e.g., \cite{he1998double, debicki2021finite, rogers2006correlation, kou2016first, shao2013estimates, metzler2010first, iyengar1985hitting, buckholtz1979first, zhou2001analysis}.
Many of these works follow a similar approach, which we briefly summarize here. Namely, if $f(x,y,t)$ is the joint density that $(X_t,Y_t)$ is at $(x,y)$ and that neither line $x=a,y=b$ is hit before time $t$, then $f$ satisfies the Kolmogorov forward equation
\begin{align}\label{eq:kolmogorov_eq_for_our_problem_paper_1}
    \frac{\partial f}{\partial t} = \frac{1}{2} \frac{\partial^2 f}{\partial x^2} + \frac{\partial^2 f}{\partial y^2} + 2 \rho \frac{\partial^2 f}{\partial x\partial y},
\end{align}
subject to the boundary and initial conditions
\begin{align*}
    &f(a,y,t)=f(x,b,t)=f(-\infty,b,t)=f(a,-\infty,t) = 0, \\
&f(x,y,0)=\delta(x)\delta(y),
\end{align*}
where $\delta$ is an indicator functions, and
\begin{align}
&\int^a_{-\infty} \int^b_{-\infty} f(x,y,t) {\rm d}x{\rm d}y \leq 1 ,\; t > 0 \label{eq:last_constraint_kolmogorov_paper_1}.
\end{align}
Integrating $f$ over the region $x \leq a$ and $y \leq b$, which is the left-hand side of \eqref{eq:last_constraint_kolmogorov_paper_1}, yields the value of (b).
After a change of variables, \eqref{eq:kolmogorov_eq_for_our_problem_paper_1} becomes a heat equation that can be solved in polar coordinates. Integrating $f$ over $x \leq a$ and $y \leq b$ leads to the following series expression for (b), reproduced from \cite{metzler2010first},
\begin{align}
    \mathbb{P}&(\max_{t\in[0,n]} X_t \leq a \land \max_{t\in[0,n]} Y_t \leq b)  \nonumber\\
    &\qquad= e^{-\frac{r^2_0}{4n}}\frac{2 {r_0}}{\sqrt{2 \pi  n}} \sum_{k=1,3,5,...}\frac{1}{k}\bigg(I_{\frac{1}{2} \left(\frac{k \pi }{\alpha }-1\right)}\Big(\frac{{r_0}^2}{4 n}\Big)\nonumber\\
    &\qquad+ I_{\frac{1}{2} \left(\frac{\pi  k}{\alpha }+1\right)}\Big(\frac{{r_0}^2}{4 n}\Big)\bigg)\sin \Big(\frac{\pi  {\theta_0} k}{\alpha }\Big),\label{eq:full_series_expression_hit_time_paper_1}
\end{align}
where, using $\rho < 0$ and that both $X_{t=1}$ and $Y_{t=1}$ have unit variance,
\begin{align*}
    r_0 &= \sqrt{\frac{a^2 + b^2 - 2\rho a b}{1-\rho^2}}  \\
    \alpha &= \arctan\left(-\frac{\sqrt{1-\rho^2}}{\rho}\right)\\
    \theta_0 &= \arctan\left(\frac{b\sqrt{1-\rho^2}}{a-b\rho}\right),
\end{align*}
and $I_a(x)$ is the  modified Bessel function of the first kind of order $a$.
Since $\rho = -1/\sqrt{2}$, we havethat $1 - \rho^2 = 1/2$, $r_0 = \sqrt{{M^*}^2_{1,0}+{M^*}^2_{2,0}}$, $\alpha = \pi/4$, and $\theta_0 = \arctan(1-\frac{2 a}{2 a+b\sqrt{2} })$, and the entire expression depends only on $n$, $M^*_{1,0}$, and $M^*_{2,0}$.

Defining $z \equiv r_0/(2\sqrt{n})$, and using the series representation
$$I_a(x) = \left(\frac{x}{2}\right)^a \sum^\infty_{j=0} \frac{(x/2)^{2j}}{j! \Gamma(n+j+1)},$$
we continue \eqref{eq:full_series_expression_hit_time_paper_1} as shown in
\eqref{eq:expression_before_expansion_of_small_z_paper_1},
where,  in the last line, we used the identity $\Gamma(x+1)=x\Gamma(x)$.
Implicit in the derivation is the fact that the order of summation may be exchanged, since the double series converges absolutely.
\begin{figure*}[h]
    \begin{equation}
        \begin{aligned}
            &e^{-z^2}\frac{4 z}{\sqrt{2 \pi}} \sum^{\infty}_{k=0}\frac{1}{2k+1}
 \Big(I_{-\frac{1}{2} + 2(2k+1)}\big(z^2\big) +I_{\frac{1}{2} + 2(2k+1)}\big(z^2\big)\Big)\sin(4(2k+1)\theta_0) \\
 &= e^{-z^2}\frac{4 z}{\sqrt{2 \pi}} \sum^{\infty}_{k=0}\frac{\sin(4(2k+1)\theta_0)}{2k+1}\sum^{\infty}_{j=0}\bigg( \frac{(z^4/4)^j\ (z^2/2)^{-\frac{1}{2} + 2(2k+1)}}{j!\Gamma(-\frac{1}{2} + 2(2k+1) + j + 1)} +\frac{(z^4/4)^j\ (z^2/2)^{\frac{1}{2} + 2(2k+1)}}{j!\Gamma(\frac{1}{2} + 2(2k+1) + j + 1)}  \bigg)\\
 &=\frac{4 z e^{-z^2}}{\sqrt{2 \pi}} \sum^{\infty}_{k=0}\sum^{\infty}_{j=0}\frac{\sin(4(2k+1)\theta_0)}{2k+1}\frac{(z^2/2)^{2j + 2(2k+1)}}{j! \Gamma(\frac{1}{2} + 2(2k+1) + j )}\bigg( \Big(\frac{z^2}{2}\Big)^{-1/2} + \frac{(z^2/2)^{1/2}}{\frac{1}{2} + 2(2k+1) + j }  \bigg).\label{eq:expression_before_expansion_of_small_z_paper_1}
        \end{aligned}
    \end{equation}
\end{figure*}

To see this, note that the expression being summed is bounded in absolute value by $$\frac{(z^4/4)^j }{j!}\frac{(z^4/4)^{s}}{s!}((z^2/2)^{1/2}+(z^2/2)^{-1/2}),$$ where $s=2k+1$.
Consequently, the sum of absolute values is bounded by
\begin{align*}
&\sum^{\infty}_{s=0} \sum^{\infty}_{j=0} \frac{(z^4/4)^j }{j!}\frac{(z^4/4)^{s}}{s!}\big((z^2/2)^{1/2}+(z^2/2)^{-1/2}\big) \\
&= e^{z^4/2}\big((z^2/2)^{1/2}+(z^2/2)^{-1/2}\big) \ <\ \infty,
\end{align*}
since $z > 0$.
As a consequence of this bound, if we truncate the summation at order $z^s$,
the truncation error is bounded by $C(z) z^{s+1}$ for some function $C(z)$ that is a power-series convergent in absolute value for which $C(0) < +\infty$.
In other words, for $z$ sufficiently small, the error is uniformly bounded by $C z^{s+1}$ for some constant $C$.

If we expand $e^{-z^2}$ and collect the lowest-order powers of $z$ in \eqref{eq:expression_before_expansion_of_small_z_paper_1}, we can express (b) as
\begin{equation*}
\mathbb{P}(0 \leq \tilde{\bf M}_{1,:} \leq \tilde{\bf M}_{2,:}) = \frac{4 \sin(4\theta_0) z^4}{3 \pi} \left(1-\frac{4 z^2}{5}+\frac{2 z^4}{5}+ O(z^6)\right),
\end{equation*}
where, after a few trigonometric manipulations and using the definition of $a$ and $b$,
\begin{align*}
\sin(4\theta_0)&=\frac{a b \left(\sqrt{2} a+2 b\right) \left(2 a+\sqrt{2} b\right)}{\left(a^2+\sqrt{2} a b+b^2\right)^2}\\
&= \frac{4 {M^*}_{1,0} {M^*}_{2,0}^3-4 {M^*}_{1,0}^3 {M^*}_{2,0}}{\left({M^*}_{1,0}^2+{M^*}_{2,0}^2\right){}^2}.
\end{align*}
At the same time, expanding (a) yields
\begin{align*}
\mathbb{P}&(\tilde{\bf M}_{1,:},  \tilde{\bf M}_{2,:} \geq 0)
\\
&=
(2\Phi(M^*_{1,0}/\sqrt{2n}) - 1)(2\Phi(M^*_{2,0}/\sqrt{2n}) - 1)\\ &=(2\Phi(z r_1) - 1)(2\Phi(z r_2) - 1) \\
&= \frac{4 r_1 r_2 z^2}{\pi }\left(1+\frac{1}{3} \left(-r_1^2-r_2^2\right) z^2+ O(z^4) \right)
\end{align*}
where $r_1=(\sqrt{2}M^*_{1,0}/r_0)$ and $r_2=(\sqrt{2}M^*_{2,0}/r_0)$ .

taking the ratio of (a) and (b), we obtain
\begin{align*}
 &\mathbb{P}_{{ M}^*}({\bf M}^*_{1,:} \leq {\bf M}^*_{2,:} | {M^*}_{1,0} , {M^*}_{2,0}) = \frac{\sin(4 \theta_0) z^2}{3 r_1 r_2} + O(z^4),
 \end{align*}
 where we also have that
 \begin{align*}
 &\frac{\sin(4 \theta_0) z^2}{3 r_1 r_2}=\frac{\sin(4 \theta_0) r^2_0 z^2}{6 M^*_{1,0}M^*_{2,0}} \\
 &=\frac{ r^2_0 z^2}{6 M^*_{1,0}M^*_{2,0}}
\frac{4 {M^*}_{1,0} {M^*}_{2,0}^3-4 {M^*}_{1,0}^3 {M^*}_{2,0}}{\left({M^*}_{1,0}^2+{M^*}_{2,0}^2\right){}^2}\\
&=\frac{ r^2_0 z^2}{6 }
\frac{4  {M^*}_{2,0}^2-4 {M^*}_{1,0}^2 }{\left({M^*}_{1,0}^2+{M^*}_{2,0}^2\right){}^2}=\frac{ r^4_0 }{6 }
\frac{  {M^*}_{2,0}^2- {M^*}_{1,0}^2 }{n\left({M^*}_{1,0}^2+{M^*}_{2,0}^2\right){}^2} \\
&=
\frac{  {M^*}_{2,0}^2- {M^*}_{1,0}^2 }{6 n}.
\end{align*}
\end{proof}

\begin{proof}[Proof of Lemma \ref{th:bound_on_prob_M1_less_M2_under_brownian_motion_paper_1}]
Since $M^*_{1,0}$ and $M^*_{2,0}$ are i.i.d. and uniform on $[0,1]$, we have that
\begin{align*}
&\mathbb{P}_{M^*}({\bf M}^*_{1,:} \leq {\bf M}^*_{2,:} )\\
&=\int^1_0 \int^1_{M^*_{1,0}}
\mathbb{P}_{M^*}({\bf M}^*_{1,:} \leq {\bf M}^*_{2,:} \mid {M^*}_{1,0} , {M^*}_{2,0})
\ {\rm d} M^*_{2,0}\ {\rm d} M^*_{1,0}\\
&=\int^1_0 \int^1_{M^*_{1,0}}
\frac{ {M^*_{2,0}}^2-{M^*_{1,0}}^2}{6 n}
{\rm d} M^*_{2,0}{\rm d}M^*_{1,0} + O(1/n^2) \\
&= 1/(36n) + O(1/n^2).\numberthis
\end{align*}
\end{proof}

\rev{
\section{Relaxing the independence assumption}\label{appendix:relaxing_independence_assumption}
If one assumes that (a) one is working with any distribution of $U^*$ (not necessarily uniform) where $U^*$ has a number of leaves bounded by, say, $Cq$ for large $q$, and that (b) for any two mutants the probability (over a random $U^*$ and random trajectory $M^*$, possibly interdependent) that one is more abundant than the other for the full duration of the observation window $[0,n]$ is increasingly unlikely and shrinks at least like $1/n$, i.e.,~\eqref{eq:shrinking_with_n} holds, and that (c) the rates of change of any two mutants is enough uncorrelated, so that the eq.\ in Lemma~6 holds, then Theorem~3 still holds. With the alternative assumptions (a), (b) and (c) we do not need to make any assumption on the independence of $U^*$ and $M^*$, or any direct assumptions on the distributions of $U^*$ or $M^*$.
}

\end{appendices}

\printbibliography[category=appendixonly,title={Appendix-only References}]

\end{document}